\documentclass[11pt]{article}
\usepackage[margin=25mm]{geometry}
\usepackage[utf8]{inputenc}
\usepackage{amsmath}
\usepackage{amsfonts}
\usepackage{amssymb}
\usepackage{graphicx}
\usepackage{verbatim}

\usepackage{microtype}
\usepackage{graphicx}
\usepackage{booktabs} 
\usepackage{amsmath}
\usepackage{amsthm}
\usepackage{amssymb}
\usepackage{amsfonts}
\usepackage{fixfoot}
\usepackage{algorithm}
\usepackage{algorithmic}
\usepackage{graphicx}
\usepackage[dvipsnames]{xcolor}
\usepackage{parskip}
\usepackage{diagbox}
\usepackage{multirow}
\usepackage[caption=false]{subfig}

\usepackage{amsfonts}
\usepackage[toc,page,header]{appendix}

\allowdisplaybreaks

\newtheorem{innercustomgeneric}{\customgenericname}
\providecommand{\customgenericname}{}
\newcommand{\newcustomtheorem}[2]{%
  \newenvironment{#1}[1]
  {%
   \renewcommand\customgenericname{#2}%
   \renewcommand\theinnercustomgeneric{##1}%
   \innercustomgeneric
  }
  {\endinnercustomgeneric}
}

\newcustomtheorem{customthm}{Theorem}
\newcustomtheorem{customlemma}{Lemma}

\usepackage{hyperref}
\usepackage{url}

\newcommand{\R}{\mathbb{R}}
\newcommand{\N}{\mathbb{N}}

\newcommand{\nnz}{\text{nnz}}
\newcommand{\eps}{\varepsilon}
\newcommand{\argmin}{\text{argmin}}

\newcommand{\poly}{\text{poly}}
\newcommand{\calE}{\mathcal{E}}
\newcommand{\calN}{\mathcal{N}}

\newcommand{\E}{\mathbb{E}}

\usepackage[english]{babel}
\usepackage{wrapfig}
\usepackage{textcomp}
\newtheorem{theorem}{Theorem}
\newtheorem{lemma}{Lemma}

\theoremstyle{definition}
\newtheorem{definition}{Definition}[section]

\title{Streaming and Distributed Algorithms for Robust Column Subset Selection\footnote{Proceedings of the $38^{th}$ International Conference on Machine Learning, 2021.}}
\author{Shuli Jiang, Dongyu Li, Irene Mengze Li, Arvind V. Mahankali, David P. Woodruff\\
    \small School of Computer Science, Carnegie Mellon University
    \footnote{Authors are listed in alphabetical order. Correspondence to: David P. Woodruff (\href{dwoodruf@andrew.cmu.edu}{dwoodruf@andrew.cmu.edu}), Shuli Jiang (\href{shulij@andrew.cmu.edu}{shulij@andrew.cmu.edu}), Arvind V. Mahankali (\href{amahanka@andrew.cmu.edu}{amahanka@andrew.cmu.edu}).}
}
\date{July 2021} 

\begin{document}
\maketitle

\begin{abstract}
We give the first single-pass streaming algorithm for Column Subset Selection with respect to the entrywise $\ell_p$-norm with $1 \leq p < 2$. 
We study the $\ell_p$ norm loss since it is often considered more robust to noise than the standard Frobenius norm.
Given an input matrix $A \in \mathbb{R}^{d \times n}$ ($n \gg d$), our algorithm achieves a multiplicative $k^{\frac{1}{p} - \frac{1}{2}}\poly(\log nd)$-approximation to the error with respect to the \textit{best possible column subset} of size $k$.
Furthermore, the space complexity of the streaming algorithm is optimal up to a logarithmic factor.
Our streaming algorithm also extends naturally to a 1-round distributed protocol with nearly optimal communication cost.
A key ingredient in our algorithms is a reduction to column subset selection in the $\ell_{p,2}$-norm, which corresponds to the $p$-norm of the vector of Euclidean norms of each of the columns of $A$. 
This enables us to leverage strong coreset constructions for the Euclidean norm, which previously had not been applied in this context.
We also give the first provable guarantees for greedy column subset selection in the $\ell_{1, 2}$ norm, which can be used as an alternative, practical subroutine in our algorithms.
Finally, we show that our algorithms give significant practical advantages on real-world data analysis tasks.
\end{abstract}
\hspace{10pt}

\newpage
\tableofcontents

\newpage

\section{Introduction}

Column Subset Selection ($k$-CSS) is a widely studied approach for low-rank approximation and feature selection. 
In $k$-CSS, on an input data matrix $A \in \mathbb{R}^{d \times n}$, we seek a small subset $A_I$ of $k$ columns from $A$ such that $\min_V \|A_I V - A\|$ is minimized for some norm $\|\cdot \|$. 
In contrast to general low-rank approximation, where one finds $U \in \R^{d \times k}$ and $V \in \R^{k \times n}$ such that $\|UV - A\|$ is minimized \cite{cw13_lra_nnz_time, w14_sketching_as_a_tool}, 
$k$-CSS outputs an actual subset of the columns of $A$ as the left factor $U$.
The main advantage of $k$-CSS over general low-rank approximation is that the resulting factorization is more interpretable. 
For instance, the subset $A_I$ can represent salient features of $A$, while in general low-rank approximation, the left factor $U$ may not be easy to relate to the original dataset. In addition, the subset $A_I$ preserves sparsity of the original matrix $A$.

$k$-CSS has been extensively studied in the Frobenius norm \cite{guruswami2012optimal,boutsidis2014near,boutsidis2017optimal,boutisidis18} and also the operator norm \cite{operator_norm_lra, sketching_as_a_tool}. 
A number of recent works ~\cite{swz2017lral1norm, cgklpw2017lplra, dwzzr2019opt_anal_css_lplra, bbbklw19_ptas_for_lp_lra, mw2019opt_l1_css_lra} studied this problem in the $\ell_p$ norm for $1 \leq p < 2$, due to its robustness properties. 
The $\ell_1$ norm, especially, is less sensitive to outliers, and better at handling missing data and non-Gaussian noise, than the Frobenius norm~\cite{swz2017lral1norm}. 
Using the $\ell_1$ norm loss has been shown to lead to improved performance in many real-world applications of low-rank approximation, such as structure-from-motion \cite{kk05_sfm_problem} and image denoising \cite{yzd12_image_denoising}. 

In this work, we give algorithms for $k$-CSS in the $\ell_p$ norm, or $k$-CSS$_p$, in the streaming and distributed settings for $1 \leq p < 2$.
The streaming algorithm can be used on small devices with memory constraints, when the elements of a large dataset are arriving one at a time and storing the entire stream is not possible.
The distributed algorithm is useful when a large dataset is partitioned across multiple devices. 
Each device only sees a subset of the entire dataset, and it is expensive to communicate or transmit data across devices.

\subsection{Background: $k$-CSS$_p$ in the Streaming Model}

We study the column-update streaming model (See Definition~\ref{def:column_update_model}) where the dataset $A$ is a $d \times n$ matrix with $n \gg d$, and the columns of $A$ arrive one by one. 
In this setting, our goal is to maintain a good subset of columns of $A$, while using space that can be linear in $d$ and $k$, but sublinear in $n$. 
This model is relevant in settings where memory is limited, and the algorithm can only make one pass over the input data, e.g.~\cite{dk03_pass_efficient}.

$k$-CSS and low-rank approximation with the Frobenius norm loss have been extensively studied in the column-update streaming model (or equivalently the row-update streaming model), e.g.~\cite{cw09_nla_in_streaming_model, liberty13_row_update_PCA, gp14_row_update_PCA, w14_row_update_lower_bound_frobenius, abfmrz2016greedycssfrobenius, BWZ16}. 
However, $k$-CSS in the streaming model has not been studied in the more robust $\ell_p$ norm, and it is not clear how to adapt existing offline $k$-CSS$_p$ or $\ell_p$ low-rank approximation algorithms into streaming $k$-CSS$_p$ algorithms. 
In fact, the only work which studies $\ell_p$ low-rank approximation in the column-update streaming model is \cite{swz2017lral1norm}. 
They obtain a $\poly(k, \log d)$-approximate algorithm with $\widetilde{O}(dk)$ space, and using the techniques of \cite{ww19lpobliviousemb} this approximation factor can be further improved. 
However, it is not clear how to turn this streaming algorithm of \cite{swz2017lral1norm} into a streaming $k$-CSS$_p$ algorithm.

In addition, it is not clear how to adapt many known algorithms for $\ell_p$ low-rank approximation and column subset selection into one-pass streaming algorithms for $k$-CSS$_p$. 
For instance, one family of offline algorithms for $k$-CSS$_p$, studied by \cite{cgklpw2017lplra, dwzzr2019opt_anal_css_lplra, swz2019zeroonelawcss, mw2019opt_l1_css_lra}, requires $O(\log n)$ iterations.
In each of these iterations, the algorithm selects $\widetilde{O}(k)$ columns of $A$ and uses their span to approximate a constant fraction of the remaining columns, which are then discarded. 
Since these rounds are adaptive, this algorithm would require $O(\log n)$ passes over the columns of $A$ if implemented as a streaming algorithm, making it unsuitable.

\subsection{Streaming $k$-CSS$_p$: Our Results}

In this work, we give the first non-trivial single-pass algorithm for $k$-CSS$_p$ in the column-update streaming model (Algorithm~\ref{alg:streaming}). 
The space complexity of our algorithm, $\widetilde{O}(kd)$, is nearly optimal up to a logarithmic factor, since for $k$ columns, each having $d$ entries, $\Omega(dk)$ words of space are needed. Our algorithm is bi-criteria, meaning it outputs a column subset of size $\widetilde{O}(k)$ (where $\widetilde{O}$ hides a $\poly(\log k)$ factor) for a target rank $k$. Bi-criteria relaxation is standard in most low-rank approximation algorithms to achieve polynomial running time, since obtaining a solution with rank exactly $k$ or exactly $k$ columns can be NP hard, e.g., in the $\ell_1$ norm case~\cite{l1nphard}.
Furthermore, we note that our algorithm achieves an $\widetilde{O}(k^{1/p - 1/2})$-approximation
to the error from \textit{the optimal column subset} of size $k$, instead of the optimal rank-$k$ approximation error (i.e., $\min_{\textrm{rank}-k A_k}\|A_k - A\|_p$). \cite{swz2017lral1norm} shows any matrix $A \in \R^{d \times n}$ has a subset of $O(k \log k)$ columns which span an $\widetilde{O}(k^{1/p - 1/2})$-approximation to the optimal rank-$k$ approximation error. Thus our algorithm is able to achieve an $\widetilde{O}(k^{2/p-1})$-approximation relative to the optimal rank-$k$ approximation error.

\subsection{Streaming $k$-CSS$_p$: Our Techniques}

Our first key insight is that we need to maintain a small subset of columns with size independent of $n$ that globally approximates all columns of $A$ well throughout the stream under the desired norm. A good data summarization technique for this is a \textit{strong coreset}, which can be a subsampled and reweighted subset of columns that preserves the cost of projecting onto all subspaces (i.e., the span of the columns we ultimately choose in our subset). However, \textit{strong coresets} are not known to exist in the $\ell_p$ norm for $1 \leq p < 2$. Thus, we reduce to low rank approximation in the $\ell_{p, 2}$ norm, which is the sum of the $p$-th powers of the Euclidean norms of all columns (see Definition~\ref{def:lp2_norm}).
Strong coresets in the $\ell_{p, 2}$ norm have been extensively studied and developed, see, e.g.,  \cite{sw2018strongcoreset}. 
However, to the best of our knowledge, $\ell_{p, 2}$ strong coresets have not been used for developing  $\ell_p$ low-rank approximation or $k$-CSS algorithms prior to our work.

The next question to address is how to construct and maintain \textit{strong coresets} of $A$ during the stream. First we observe that \textit{strong coresets} are mergeable. If two coresets $C_1$, $C_2$ provide a $(1\pm \epsilon)$-approximation to the $\ell_{p, 2}$-norm cost of projecting two sets of columns $A_M$ and $A_N$ respectively, to any subspace, then the coreset of $C_1 \cup C_2$ gives a $(1 \pm \epsilon)^2$-approximation to the cost of projecting $A_M \cup A_N$ onto any subspace. Thus, we can construct \textit{strong coresets} for batches of input columns from $A$ and merge these coresets to save space while processing the stream. In order to reduce the number of merges, and hence the approximation error, we apply the Merge-and-Reduce framework (see, e.g.,~\cite{mcgregor14_graph_stream_survey}). Our algorithm greedily merges the \textit{strong coresets} in a binary tree fashion, where the leaves correspond to batches of the input columns and the root is the single \textit{strong coreset} remaining at the end. This further enables us to maintain only $O(\log n)$ coresets of size $\widetilde{O}(k)$ throughout the stream, and hence achieve a space complexity of $\widetilde{O}(kd)$ words.

One problem with reducing to the $\ell_{p, 2}$ norm is that this leads to an approximation factor of $d^{1/p - 1/2}$. Our second key insight is to apply a dimensionality reduction technique in the $\ell_p$ norm, which reduces the row dimension from $d$ to $k\poly(\log nd)$ via data-independent (i.e., oblivious) sketching matrices of i.i.d. $p$-stable random variables (see Definition~\ref{def:pstables}), which only increases the approximation error by a factor of $O(\log nd)$. The overall approximation error is thus reduced to $O(k^{1/p-1/2}\poly(\log nd))$. 

As a result, our algorithm constructs coresets for the sketched columns instead of the original columns. However, since we do not know which subset of columns will be selected a priori, we need approximation guarantees of dimensionality reduction via oblivious skecthing matrices for all possible subsets of columns.
We combine a net argument with a union bound over all \textit{possible} subspaces spanned by column subsets of $A$ of size $\widetilde{O}(k)$ (see Lemma~\ref{lemma:pstable}). Previous arguments involving sketching for low-rank approximation algorithms, such as those by \cite{swz2017lral1norm, bbbklw19_ptas_for_lp_lra, mw2019opt_l1_css_lra}, only consider a single subspace at a time.

At the end of the stream we will have a single coreset of size $k \cdot \poly(\log nd)$. To further reduce the size of the set of columns output, we introduce an $O(1)$-approximate bi-criteria column subset selection algorithm in the $\ell_{p, 2}$ norm ($k$-CSS$_{p, 2}$; see Section~\ref{subsection:lp2_css_algorithm}) that selects $k\poly(\log k)$ columns from the coreset as the final output.

\subsection{Distributed $k$-CSS$_p$: Results and Techniques}
Our streaming algorithm and techniques can be extended to an efficient one-round distributed protocol for $k$-CSS$_p$ ($1 \leq p < 2$). 
We consider the column partition model in the distributed setting (see Definition~\ref{def:column_partition_model}), where $s$ servers communicate to a central coordinator via 2-way channels. This model can simulate arbitrary point-to-point communication by having the coordinator forward a message from one server to another; this increases the total communication by a factor of $2$ and an additive $\log s$ bits per message to identify the destination server.

Distributed low-rank approximation arises naturally when a dataset is too large to store on one machine, takes a prohibitively long time for a single machine to compute a rank-$k$ approximation, or is collected simultaneously on multiple machines. 
The column partition model arises naturally in many real world scenarios such as federated learning \cite{fegk13_distributed_css, abfmrz2016greedycssfrobenius, lbkw14_distributed_pca}. Despite the flurry of recent work on $k$-CSS$_p$, this problem remains largely unexplored in the distributed setting. This should be contrasted to Frobenius norm column subset selection and low-rank approximation, for which a number of results in the distributed model are known, see, e.g., \cite{abfmrz2016greedycssfrobenius,BLSW015,BSW016,BWZ16}. 

In this work, we give the first one-round distributed protocol for $k$-CSS$_p$ (Algorithm~\ref{alg:protocol}). Each server sends the coordinator a \textit{strong coreset} of columns. To reduce the number of columns output, our protocol applies the $O(1)$-approximate bi-criteria $k$-CSS$_{p,2}$ algorithm to give $k\poly(\log k)$ output columns, independent of $s,n$, and $d$. 
The communication cost of our algorithm, $\widetilde{O}(sdk)$ is optimal up to a logarithmic factor. Our distributed protocol is also a bi-criteria algorithm outputting $\widetilde{O}(k)$ columns and achieving an $\widetilde{O}(k^{1/p - 1/2})$-approximation relative to the error of the optimal column subset.

\subsection{Comparison with Alternative Approaches in the Distributed Setting}
If one only wants to obtain a good left factor $U$, and not necessarily a column subset of $A$, in the column partition model, one could simply sketch the columns of $A_i$ by applying an oblivious sketching matrix $S$ on each server. Each server sends $A_i \cdot S$ to the coordinator. 
The coordinator obtains $U = AS$ as a column-wise concatenation of the $A_iS$.  
\cite{swz2017lral1norm} shows that $AS$ achieves an $\widetilde{O}(\sqrt{k})$ approximation to the optimal rank-$k$ error, and this protocol only requires $\widetilde{O}(sdk)$ communication, $O(1)$ rounds and polynomial running time. However, while $AS$ is a good left factor, it does not correspond to an actual subset of columns of $A$. 

Obtaining a subset of columns that approximates $A$ well with respect to the $p$-norm in a distributed setting is non-trivial. One approach due to \cite{swz2017lral1norm} is to take the matrix $AS$ described above, sample rows according to the Lewis weights \cite{cp2015lewisweights} of $AS$ to get a right factor $V$, which is in the row span of $A$, and then use the Lewis weights of $V$ to sample columns of $A$. Unfortunately, this protocol only achieves a loose $\widetilde{O}(k^{3/2})$ approximation to the optimal rank-$k$ error \cite{swz2017lral1norm}. Moreover, it is not known how to do Lewis weight sampling in a distributed setting. Alternatively, one could first apply $k$-CSS$_p$ on $A_i$ to obtain factors $U_i$ and $V_i$ on each server, and then send the coordinator all the $U_i$ and $V_i$. The coordinator then column-wise stacks the $U_iV_i$ to obtain $U \cdot V$ and selects $\widetilde{O}(k)$ columns from $U \cdot V$. Even though this protocol applies to all $p\geq 1$, it achieves a loose $O(k^2)$ approximation to the optimal rank-$k$ error and requires a prohibitive $O(n + d)$ communication cost\footnote{We give this protocol and the analysis in the supplementary.}. One could instead try to just communicate the matrices $U_i$ to the coordinator, which results in much less communication, but this no longer gives a good approximation. Indeed, while each $U_i$ serves as a good approximation locally, there may be columns that are locally not important, but become globally important when all of the matrices $A_i$ are put together. What is really needed here is a small \emph{strong coreset} $C_i$ for each $A_i$ so that if one concatenates all of the $C_i$ to obtain $C$, \emph{any} good column subset of the coreset $C$ corresponds to a good column subset for $A$.

\subsection{Greedy $k$-CSS and Empirical Evaluations}
We also propose an offline, greedy algorithm to select columns in the $\ell_{p, 2}$ norm, $\forall p \in [1, 2)$ (see Section~\ref{section:greedy_analysis}), which can be used as an alternative subroutine in both of our algorithms, and show the provable additive error guarantees for this algorithm.
Similar error guarantees were known for the Frobenius norm \cite{abfmrz2016greedycssfrobenius}, though nothing was known for the $\ell_{p,2}$ norm. 
We implement both of our streaming and distributed algorithms in the $\ell_1$ norm and experiment with real-world text document and genetic analysis applications. We compare the $O(1)$-approximate bi-criteria $k$-CSS$_{1,2}$ (denoted regular CSS$_{1, 2}$) and the greedy $k$-CSS$_{1,2}$ as subroutines of our streaming and distributed algorithms, and show that greedy $k$-CSS$_{1, 2}$ yields an improvement in practice. Furthermore, we compare our $O(1)$-approximate $k$-CSS$_{1,2}$ subroutine against one $k$-CSS$_2$ algorithm as active learning algorithms on a noisy image classification task to show that the $\ell_1$ norm loss is indeed more robust for $k$-CSS to non-Gaussian noise.

Note that regular CSS$_{p, 2}$ gives a stronger multiplicative $O(1)$ approximation to the optimal rank-$k$ error, while greedy CSS$_{p, 2}$ gives an additive approximation to the error from the best column subset. However, in practice, we observe that greedy CSS$_{1, 2}$ gives lower approximation error than regular CSS$_{1, 2}$ in the $\ell_1$ norm, though it can require significantly longer running time than regular CSS$_{1, 2}$. An additional advantage of greedy CSS$_{p, 2}$ is that it is simpler and easier to implement.

\subsection{Technical Novelty}
We highlight several novel techniques and contributions: \begin{itemize}
    \item \textbf{Non-standard Net Argument:} 
    To apply dimensionality reduction techniques via $p$-stable random variables to reduce the final approximation error, on the lower bound side, we need to show that $p$-stable random variables do not reduce the approximation error (i.e. no contraction), with respect to \textit{all possible} column subsets of $A$, with high probability, in Lemma~\ref{lemma:pstable}. 
    While the net arguments are widely used, our proof 
    of Lemma~\ref{lemma:pstable} is non-standard: we union bound over all possible subspaces defined on \textit{subsets} of size $\widetilde{O}(k)$. This is more similar to the Restricted Isometry Property (RIP), which is novel in this context, but we only need a one-sided RIP since we only require no contraction; on the upper bound side, we just argue with constant probability the \textit{single} optimal column subset and its corresponding right factor do not dilate much.
    
    \item \textbf{\textit{Strong Coresets} for CSS: } Strong coresets for the $\ell_{p, 2}$ norm have not been used for entrywise $\ell_p$ norm column subset selection, or {\em even for $\ell_p$ low rank approximation}, for $p \neq 2$. This is perhaps because strong coresets for subspace approximation, i.e., coresets that work for all query subspaces simultaneously, are not known to exist for sums of $p$-th powers
    of $\ell_p$-distances for $p \neq 2$; our work provides a workaround to this. By switching to the Euclidean norm in a low-dimensional space we can use strong coresets for the $\ell_{p,2}$ norm with a small distortion. 

    \item \textbf{Greedy $\ell_{p, 2}$-norm CSS: }
    We give the first provable guarantees for greedy $\ell_{p, 2}$-norm column subset selection. 
    We show that the techniques used to derive guarantees for greedy CSS in the Frobenius norm from~\cite{abfmrz2016greedycssfrobenius} can be extended to the $\ell_{p, 2}$ norms, $\forall p \in [1, 2)$. A priori, it is not clear this should work, since for example, $\ell_{1, 2}$ norm low rank approximation is NP-hard~\cite{cw2015subsapproxl2} while Frobenius norm low rank approximation can be solved in polynomial time.
\end{itemize}

\section{Problem Setting}
\label{section:problemsetup}
\begin{definition}[Column-Update Streaming Model~\cite{swz2017lral1norm}] 
\label{def:column_update_model}
    Let $A_{*1}, A_{*2}, \cdots, A_{*n}$ be a set of columns from the input matrix $A \in \mathbb{R}^{d \times n}$. In the column-update model, each column of $A$ will occur in the stream exactly once, but the columns can be in an arbitrary order. An algorithm in this model is only allowed a single pass over the columns. At the end of the stream, the algorithm stores some information about $A$. The space of the algorithm is the total number of words required to store this information during the stream. Here, each word is $O(\log nd)$ bits.
\end{definition}

\begin{definition}[Column Partition Distributed Model~\cite{swz2017lral1norm}]
\label{def:column_partition_model}
There are $s$ servers, the $i$-th of which holds matrix $A_i \in \R^{d \times n_i}$ as the input. Suppose $n = \sum_{i=1}^s n_i$, and the global data matrix is denoted by $A = [A_1, A_2, \dots, A_s]$. $A$ is column-partitioned and distributed across $s$ machines. Furthermore, there is a coordinator. The model only allows communication between the servers and the coordinator. The communication cost in this model is the total number of words transferred between machines and the coordinator. Each word is $O(\log snd)$ bits.
\end{definition}

\begin{definition}[$\ell_{p, 2}$ norm]
\label{def:lp2_norm}
For matrix $A \in \mathbb{R}^{d \times n}$, $\|A\|_{p, 2} = (\sum_{j=1}^n \|A_{*j}\|_2^p)^{1/p}$, where $A_{*j}$ denotes the $j$-th column. 
\end{definition}

\begin{definition}[$p$-Stable Distribution and Random Variables]
\label{def:pstables}
Let $X_1, \dots, X_d$ be random variables drawn i.i.d. from some distribution $\mathcal{D}$. $\mathcal{D}$ is called $p$-stable if for an arbitrary vector $v \in \mathbb{R}^d$, $\langle v, X\rangle = \|v\|_pZ$ for some $Z$ drawn from $\mathcal{D}$, where $X = [X_1, \dots, X_d]^T$. $\mathcal{D}$ is called a $p$-stable distribution --- these exist for $p \in (0, 2]$. Though there is no closed form expression for the $p$-stable distribution in general except for a few values of $p$, we can efficiently generate a single $p$-stable random variable in $O(1)$ time using the following method due to \cite{cms76_simulating_pstable}: if $\theta \in [-\frac{\pi}{2}, \frac{\pi}{2}]$ and $r \in [0, 1]$ are sampled uniformly at random, then, $\frac{\sin(p\theta)}{\cos^{1/p}\theta}(\frac{\cos(\theta(1-p))}{\ln(\frac{1}{r})})^{\frac{1-p}{p}}$ follows a $p$-stable distribution.
\end{definition}

\section{Preliminaries}
\label{section:preliminaries}

In this section, we introduce the dimensionality reduction techniques we make use of: strong coresets for $\ell_{p, 2}$-norm low-rank approximation and oblivious sketching using $p$-stable random matrices. We begin with a standard relationship between the $\ell_p$ norm and the $\ell_{p, 2}$ norm:

\begin{lemma}\label{lemma:norm}
For a matrix $A \in \mathbb{R}^{d \times n}$ and $p \in [1, 2)$,
$\|A\|_{p,2} \leq \|A\|_p \leq d^{\frac{1}{p}-\frac{1}{2}} \|A\|_{p,2}$.
\end{lemma}

\subsection{Dimensionality Reduction in the $\ell_p$ norm}

To reduce the row dimension $d$, we left-multiply $A$ by an oblivious sketching matrix $S$ with i.i.d. $p$-stable entries so that our approximation error only increases by an $\widetilde{O}(k^{\frac{1}{p} - \frac{1}{2}})$ factor instead of an $O(d^{\frac{1}{p}- \frac{1}{2}})$ factor when we switch to the $\ell_{p, 2}$ norm. The following lemma shows that for all column subsets $A_T$ and right factors $V$, the approximation error when using these to fit $A$ does not shrink after multiplying by $S$ (i.e., this holds simultaneously for all $A_T$ and $V$):

\begin{lemma}[Sketched Error Lower Bound]
\label{lemma:pstable}
Let $A \in \R^{d \times n}$ and $k \in \N$. Let $t = k \cdot \poly(\log (nd))$, and let $S \in \R^{t \times d}$ be a matrix whose entries are i.i.d. standard $p$-stable random variables, rescaled by $\Theta(1/t^{\frac{1}{p}})$. Then, with probability $1 - o(1)$, for \textbf{all} $T \subset [n]$ with $|T| = k \cdot \poly(\log k)$ and for \textit{all} $V \in \R^{|T| \times n}$,
$$\|A_TV - A\|_p \leq \|SA_TV - SA\|_p$$ 
\end{lemma}

We also recall the following upper bound for oblivious sketching from \cite{swz2017lral1norm} for a \textit{fixed} subset of columns $A_T$ and a \textit{fixed} $V$.

\begin{lemma}[Sketched Error Upper Bound (Lemma E.11 of \cite{swz2017lral1norm})]
\label{lemma:bestcss}
Let $A \in \R^{d \times n}$ and $k \in \N$. Let $t = k \cdot \poly(\log (nd))$, and let $S \in \mathbb{R}^{t \times d}$ be a matrix whose entries are i.i.d. standard $p$-stable random variables, rescaled by $\Theta(1/t^{\frac{1}{p}})$. Then, for a fixed subset $T \subset [n]$ of columns with $|T| = k\cdot \poly(\log k)$ and a fixed $V \in \R^{|T| \times n}$, with probability $1 - o(1)$, we have 
$$\min_V \|SA_TV - SA\|_p \leq \min_V O(\log^{1/p}(nd))\|A_TV - A\|_p$$
\end{lemma}

\subsection{Strong Coresets in the $\ell_{p, 2}$ Norm}

As mentioned above, strong coresets for $\ell_{p, 2}$-norm low-rank approximation are reweighted column subsets which preserve the $\ell_{p, 2}$-norm approximation error incurred by any rank-$k$ projection. Our construction of strong coresets follows~\cite{sw2018strongcoreset}, which is based on Lewis weights~\cite{cp2015lewisweights} sampling. Note that~\cite{sw2018strongcoreset} only states such strong coresets hold with constant probability. But in our applications, we need to union bound over multiple constructions of strong coresets, so need a lower failure probability. 
The only reason the coresets of ~\cite{sw2018strongcoreset} hold only with constant
probability is because they rely on the sampling result of ~\cite{cp2015lewisweights}, which is stated for constant probability. However, the results of ~\cite{cp2015lewisweights} are a somewhat
arbitrary instantiation of the failure probabilities of the earlier $\ell_p$-Lewis weights sampling
results in ~\cite{blm1989zenotope} in the functional analysis literature. That work performs
$\ell_p$-Lewis weight sampling and we show how to obtain failure probability $\delta$ with a
$\log(1/\delta)$ dependence in Section B.1 of the supplementary material.

\begin{lemma}[Strong Coresets in $\ell_{p,2}$ norm \cite{sw2018strongcoreset}]\label{lemma:coreset}
Let $A \in \R^{d \times n}$, $k \in \N$, $p \in [1, 2)$, and $\eps, \delta \in (0, 1)$. Then, in $\widetilde{O}(nd)$ time, one can find a sampling and reweighting matrix $T$ with
$O(\frac{d}{\eps^2} \poly(\log(d/\eps), \log(1/\delta)))$ columns, such that, with probability $1 - \delta$, for all rank-$k$ matrices $U$, 
\begin{align*}
    \min_{\textrm{rank-k }V} \|UV - AT\|_{p,2} = (1 \pm \eps)\min_{\textrm{rank-k }V}\|UV - A\|_{p,2}
\end{align*}
where $AT$ is called a \textbf{strong coreset} of $A$.
\end{lemma}

\subsection{$O(1)$-approximate Bi-criteria $k$-CSS$_{p, 2}$}
\label{subsection:lp2_css_algorithm}
We introduce 
an $O(1)$-approximation bi-criteria algorithm, which is a modification of the algorithm from \cite{cw2015subsapproxl2}. The number of output columns is $\widetilde{O}(k)$ instead of $O(k^2)$ since we use $\ell_p$ Lewis weights instead of $\ell_p$ leverage scores.
Details are in the supplementary material. 

\begin{theorem}[Bicriteria $O(1)$-Approximation Algorithm for $k$-CSS$_{p, 2}$]
\label{thm:bicriteria_k_css_p2}
Let $A \in \R^{d \times n}$ and $k \in \N$. There is an algorithm with $(\nnz(A) + d^2)\cdot k \cdot \poly(\log k)$ runtime that outputs a rescaled subset of columns $U \in \R^{d \times \widetilde{O}(k)}$ of $A$ and a right factor $V \in \R^{\widetilde{O}(k) \times n}$ for which $V = \min_{V}\|UV - A\|_{p, 2}$, such that with probability $1 - o(1)$, 
$$\|UV - A\|_{p, 2} \leq O(1) \cdot \min_{\text{ rank-k}\ A_k} \|A_k - A\|_{p, 2}$$
\end{theorem}

\section{A Streaming Algorithm for $k$-CSS$_p$}
\label{section:streaming_algo}

\begin{algorithm}[tb]
\caption{A one-pass streaming algorithm for bi-criteria $k$-CSS$_p$ in the column-update streaming model.}
\label{alg:streaming}
\begin{algorithmic}
\STATE \textbf{Input: } A matrix $A \in \R^{d \times n}$ whose columns arrive one at a time, $p \in [1, 2)$, rank $k \in \N$ and batch size $r$.
\STATE \textbf{Output: } A subset of $\widetilde{O}(k)$ columns $A_I$.
\STATE Generate a dense $p$-stable sketching matrix $S \in \mathbb{R}^{k\poly(\log(nd)) \times d}$.
\STATE A list of strong coresets and the level number $C \leftarrow \{\}$.
\STATE A list of columns corresponding to the list of strong coresets and the level number $D \leftarrow \{\}$.
\STATE A list of sketched columns $M \leftarrow \{\}$.
\STATE A list of columns $L \leftarrow \{\}$.
\FOR{Each column $A_{*j}$ seen in the data stream}
\STATE $M \leftarrow M \cup SA_{*j}$
\STATE $L \leftarrow L \cup A_{*j}$
\IF{length of $M$ == r}
\STATE $C \leftarrow C \cup (M, 0), D \leftarrow D \cup L$
\STATE $C, D \leftarrow$ \textbf{Recursive Merge}($C, D$) \COMMENT{// Algorithm~\ref{alg:recursive_merge}}
\STATE $M \leftarrow \{\}, L \leftarrow \{\}$
\ENDIF
\ENDFOR
\STATE $C \leftarrow C \cup (M, 0), D \leftarrow D \cup L$
\STATE $C, D \leftarrow$ \textbf{Recursive Merge}($C, D$) \COMMENT{// Algorithm~\ref{alg:recursive_merge}}
\STATE Apply $k$-CSS$_{p, 2}$ on the single strong coreset left in $C$ to obtain the indices $I$ of the subset of selected columns with size $O(k \times \poly(\log k))$. Recover the original columns of $A$ by mapping indices $I$ to columns in $D$ to get the subset of columns $A_I$.
\end{algorithmic}
\end{algorithm}

\begin{algorithm}
\caption{Recursive Merge}
\label{alg:recursive_merge}
\begin{algorithmic}
\STATE \textbf{Input: } A list $C$ of strong coresets and their corresponding level numbers. A list $D$ of (unsketched) columns of $A$ corresponding to the sketched columns in $C$.
\STATE \textbf{Output: } New $C$, where the list of strong coresets is greedily merged, and the corresponding new $D$.
\IF{length of $C$ == 1}
\STATE Return $C, D$.
\ELSE
\STATE Let $(C_{-2}, l_{-2}), (C_{-1}, l_{-1})$ be the second to last and last sets of columns $C_{-2}, C_{-1}$ with their corresponding level $l_{-2}, l_{-1}$ from list $C$.
\IF{$l_{-2}$ == $l_{-1}$}
\STATE Remove $(C_{-2}, l_{-2}), (C_{-1}, l_{-1})$ from $C$.
\STATE Remove the corresponding $D_{-2}, D_{-1}$ from $D$.
\STATE Compute a strong coreset $C_0$ of (i.e., select columns from) $C_{-2} \cup C_{-1}$. Record the indices $I$ of the columns selected in $C_0$.
\STATE Map indices $I$ to columns in $D_{-2} \cup D_{-1}$ to form a new subset of columns $D_0$.
\STATE $C \leftarrow C \cup (C_0, l_{-1} + 1), D \leftarrow D \cup D_0$.
\STATE \textbf{Recursive Merge}($C, D$).
\ELSE
\STATE Return $C, D$.
\ENDIF
\ENDIF
\end{algorithmic}
\end{algorithm}

Our one-pass streaming algorithm (Algorithm~\ref{alg:streaming}) is based on the Merge-and-Reduce framework. The $n$ columns of the input matrix $A$ are partitioned into $\lceil n/r \rceil$ batches of length $r = k \cdot \poly(\log nd)$.
See the supplementary material for an illustration. 
These $\lceil n/r \rceil$ batches can be viewed as the leaves of a binary tree and are considered to be at level $0$ of the tree. 
A merge operation (Algorithm~\ref{alg:recursive_merge}) is used as a subroutine of Algorithm~\ref{alg:streaming} --- it computes a strong coreset of two sets of columns corresponding to the two children nodes. 
Each node in the binary tree represents a set of columns. 
Starting from level $0$, every pair of neighboring batches of columns will be merged, until there is only one coreset of columns left at the root, i.e. level $\log (n/k)$. 
During the stream, the nodes are greedily merged. 
The streaming algorithm constructs strong coresets and merges the sketched columns $SA_{*j}$ (list $C$), while keeping a list of corresponding columns $A_{*j}$ (list $D$) at the same time, in order to recover the original columns of $A$ as the final output.

\begin{theorem}[A One-pass Streaming Algorithm for $k$-CSS$_p$]
\label{thm:streaming_algo}
    In the column-update streaming model, let $A \in \mathbb{R}^{d \times n}$ be the data matrix whose columns arrive one at each time in a data stream. Given $p \in [1, 2)$ and a desired rank $k \in \mathbb{N}$, Algorithm~\ref{alg:streaming} outputs a subset of columns $A_I \in \mathbb{R}^{d \times k\poly(\log(k))}$ in $\widetilde{O}(\nnz(A)k + nk + k^3)$ time, such that with probability $1 - o(1)$,
    \begin{align*}
        \min_{V}\|A_IV - A\|_p \leq \widetilde{O}(k^{1/p - 1/2})\min_{L \subset [n], |L|=k}\|A_LV - A\|_p
    \end{align*}
    Moreover, Algorithm \ref{alg:streaming} only needs to process all columns of $A$ once and uses $\widetilde{O}(dk)$ space throughout the stream.
\end{theorem}

\begin{proof}
We give a brief sketch of the proof (the full proof is in the supplementary). We first need \textbf{Lemma}~\ref{lemma:approx_merge} below to show how the approximation error propagates through each level induced by the merge operator. It gives the approximation error of a strong coreset $C_0$ computed at level $l$ with respect to the union of all sets of columns represented as the leaves of the subtree rooted at $C_0$.

\begin{lemma}[Approximation Error from Merging]
\label{lemma:approx_merge}
Let $C_0$ be the strong coreset of size $\widetilde{O}(k)$ (See Lemma~\ref{lemma:coreset}) at level $l$ constructed from a union of its two children $C_{-1}\cup C_{-2}$, with $\frac{k}{\gamma^2} \cdot \poly(\log(nd/\gamma))$ columns, where $\gamma \in (0, 1)$. Then with probability $1 - \frac{1}{n^2}$, for all rank-k matrices $U$,\\ $$\min_{V}\|UV - C_0\|_{p, 2} = (1 \pm \gamma) \min_{V} \|UV - (C_{-1}\cup C_{-2})\|_{p, 2}$$\\
Let $M$ be the union of all sets of sketched columns represented as the leaves of the subtree rooted at $C_0$ (and assume the subtree has size $q$).
Then with probability $1 - \frac{q}{n^2}$, for all rank-k matrices $U$,\\ $$\min_{V}\|UV - C_0\|_{p, 2} = (1 \pm \gamma)^{l}\min_{V} \|UV - M\|_{p, 2}$$
\end{lemma}

If at the leaves (level 0), we construct $(1 \pm \frac{\eps}{\log n})$-approximate coresets for the input columns from the stream, the final single coreset left at the root level (level $\log(n/k)$) will be a $(1 \pm \eps)$-approximate coreset for \textit{all} columns of $SA$. The $k$-CSS$_{p, 2}$ algorithm that selects $O(k \poly(\log k))$ columns from this coreset gives an $O(1)$-approximation by Theorem~\ref{thm:bicriteria_k_css_p2}. 
By Lemmas \ref{lemma:norm}, \ref{lemma:pstable} and \ref{lemma:bestcss}, the final approximation error of this algorithm is dominated by the one from the relaxation to the $\ell_{p, 2}$ norm, which leads to an overall $\widetilde{O}(k^{1/p - 1/2})$ approximation factor. Note that the space complexity is $\widetilde{O}(dk)$ since each coreset has $\widetilde{O}(k)$ columns and we only keep coresets for at most $O(\log n)$ of the nodes of the tree, at a single time. The running time is dominated by the $O(n/k)$ merging operators throughout the stream and the $k$-CSS$_{p, 2}$ algorithm. A detailed analysis
is in the supplementary.
\end{proof}

\section{A Distributed Protocol for $k$-CSS$_p$}
\label{section:distributed_protocol}

\begin{theorem}[A One-round Protocol for Distributed $k$-CSS$_p$]
\label{thm:analysis_of_distributed_protocol}
In the column partition model, let $A \in \R^{d \times n}$ be the data matrix whose columns are partitioned across $s$ servers and suppose server $i$ holds a subset of columns $A_i \in \mathbb{R}^{d \times n_i}$, where $n = \sum_{i \in [s]} n_i$. Then, given $p \in [1, 2)$ and a desired rank $k \in \N$, Algorithm \ref{alg:protocol} outputs a subset of columns $A_I \in \R^{d \times k \poly(log(k))}$ 
in $\widetilde{O}(\nnz(A)k + kd + k^3)$ time, such that with probability $1 - o(1)$,
$$\min_V \|A_I V - A\|_p \leq \widetilde{O}(k^{1/p - 1/2}) \min_{L \subset [n], |L| = k} \|A_L V - A\|_p$$
Moreover, Algorithm \ref{alg:protocol} uses one round of communication and $\widetilde{O}(sdk)$ words of communication.
\end{theorem}

\begin{algorithm}[tb]
   \caption{A one-round protocol for bi-criteria $k$-CSS$_p$ in the column partition model}
   \label{alg:protocol}
\begin{algorithmic}
    \STATE \textbf{Initial State:} \\
    Server $i$ holds matrix $A_i \in \mathbb{R}^{d \times n_i}$, $\forall i \in [s]$.
   \STATE {\bfseries Coordinator:}\\
     Generate a dense $p$-stable sketching matrix $S \in \mathbb{R}^{k \textrm{ poly}(\log (nd)) \times d}$.\\ Send $S$ to all servers.
   \STATE {\bfseries Server $i$:}\\ Compute $SA_i$. \\
   Let the number of samples in the coreset be $t = O(k \cdot \textrm{poly}(\log (nd)))$.
    Construct a coreset of $SA_i$ under the $\ell_{p,2}$ norm by applying a sampling matrix $D_i$ of size $n_i \times t$ and a diagonal reweighting matrix $W_i$ of size $t \times t$. \\
    Let $T_i = D_iW_i$.
    Send $SA_iT_i$ along with $A_iD_i$ to the coordinator.
    \STATE {\bfseries Coordinator:}\\ Column-wise stack $SA_iT_i$ to obtain $SAT = [SA_1T_1, SA_2T_2, \dots, SA_sT_s]$.\\
    Apply $k$-CSS$_{p, 2}$ on $SAT$ to obtain the indices $I$ of the subset of selected columns with size $O(k \cdot \poly(\log k))$. \\
    Since $D_i$'s are sampling matrices, the coordinator can recover the original columns of $A$ by mapping indices $I$ to $A_iD_i$'s. \\
    Denote the final selected subset of columns by $A_I$.
    Send $A_I$ to all servers.
    \STATE {\bfseries Server $i$:}\\
    Solve $\min_{V_i} \|A_IV_i - A_i\|_p$ to obtain the right factor $V_i$. $A_I$ and $V$ will be factors of a rank-$k \cdot \poly(\log k)$ factorization of $A$, where $V$ is the (implicit) column-wise concatenation of the $V_i$.
\end{algorithmic}
\end{algorithm}

The analysis of the protocol is similar to the analysis of our streaming algorithm (Section~\ref{section:streaming_algo}). We give a detailed analysis in the supplementary.

\section{Greedy $k$-CSS$_{p, 2}$}
\label{section:greedy_analysis}

We propose a greedy algorithm for $k$-CSS$_{p, 2}$ (Algorithm~\ref{algorithm:fast_greedy_kCSSp2}).
In each iteration, the algorithm samples a subset $A_C$ of $O(\frac{n}{k}\log(\frac{1}{\delta}))$ columns from the input matrix $A \in \mathbb{R}^{d \times n}$ and picks the column among $A_C$ that reduces the approximation error the most.
We give the first provable guarantees for this algorithm below, the proof of which is in the supplementary. \footnotemark \footnotetext{The analysis is based on the analysis by \cite{abfmrz2016greedycssfrobenius} of greedy $k$-CSS$_2$.}
We also empirically compare this algorithm to the $k$-CSS$_{p, 2}$ algorithm mentioned above, in Section \ref{section:exp}.

\begin{algorithm}
\caption{Greedy $k$-CSS$_{p, 2}$. }
\label{algorithm:fast_greedy_kCSSp2}
\begin{algorithmic}
\STATE \textbf{Input:} The data matrix $A \in \R^{d \times n}$. A desired rank $k \in \mathbb{N}$ and $p \in [1, 2)$. The number of columns to be selected $r \leq n$. Failure probability $\delta \in (0, 1)$.
\STATE \textbf{Output:} A subset of $r$ columns $A_T$. 
\STATE Indices of selected columns $T \leftarrow \{\}$.
\FOR {$i = 1$ to $r$}
    \STATE $C \leftarrow$ Sample $\frac{n}{k}\log(\frac{1}{\delta})$ indices from $\{1, 2, \dots, n\} \setminus T$ uniformly at random.
    \STATE Column index $j^* \gets \argmin_{j \in C} (\min_{V} \|A_{T \cup j}V - A\|_{p, 2}$)
    \STATE $T \leftarrow T \cup j^*$. 
\ENDFOR
\STATE Map indices $T$ to get the selected columns $A_T$.
\end{algorithmic}
\end{algorithm}

\begin{theorem}[Greedy $k$-CSS$_{p, 2}$]
\label{thm:convergence_of_fast_greedy_kCSSp2}
    Let $p \in [1, 2)$.
    Let $A \in \mathbb{R}^{d \times n}$ be the data matrix and $k \in \mathbb{N}$ be the desired rank. Let $A_L$ be the best possible subset of $k$ columns, i.e., $A_L = \argmin_{A_L} \min_{V}\|A_LV - A\|_{p,2}$. Let $\sigma$ be the minimum non-zero singular value of the matrix $B$ of normalized columns of $A_L$, (i.e., the $j$-th column of $B$ is $B_{*j} = (A_L)_{*j}/\|(A_L)_{*j}\|_2$).
    Let $T \subset [n]$ be the subset of output column indices selected by \textbf{Algorithm}~\ref{algorithm:fast_greedy_kCSSp2}, for $\epsilon, \delta \in (0, 1)$, for $|T| = \Omega(\frac{k}{p\sigma^2 \epsilon^2})$, with probability $1 - \delta$,  $$\mathbb{E}[\min_{V} \|A_TV - A\|_{p,2}] \leq \min_V \|A_LV - A\|_{p,2}  + \epsilon \|A\|_{p,2}$$ 
    The overall running time is $O(\frac{n}{p\sigma^2\epsilon^2}\log(\frac{1}{\delta})\cdot (\frac{dk^2}{p^2\sigma^4\epsilon^4} + \frac{ndk}{p\sigma^2\epsilon^2}))$.
\end{theorem}

\section{Experiments}
\label{section:exp}
The source code is available at:
\href{https://github.com/11hifish/robust\_css}{\color{RubineRed} https://github.com/11hifish/robust\_css}.
\subsection{Streaming and Distributed $k$-CSS$_1$}
We implement both of our streaming and distributed $k$-CSS$_1$ algorithms, with subroutines $\textbf{regular}$ $k$-CSS$_{1,2}$ (Section \ref{subsection:lp2_css_algorithm}) and \textbf{greedy} $k$-CSS$_{1,2}$ (Section \ref{section:greedy_analysis}). Given a target rank $k$, we set the number of output columns to be $k$. We compare against a commonly used baseline for low-rank approximation~\cite{swz19avgcasecssl1, cgklpw2017lplra}, \textbf{SVD} (rank-$k$ singular value decomposition), and a \textbf{uniform} random baseline. In the streaming setting, the \textbf{uniform} baseline first collects $k$ columns from the data stream and on each of the following input columns, it decides whether to keep or discard the new column with equal probability. If it keeps the new columns, it will pick one existing column to replace uniformly at random. In the distributed setting, the \textbf{uniform} baseline simply uniformly at random selects $k$ columns.

We apply the proposed algorithms on one synthetic data to show when SVD (and hence all the existing $k$-CSS$_2$ algorithms) fails to find a good subset of columns in the $\ell_1$ norm. We also apply the proposed algorithms to two real-world applications, where
$k$-CSS$_2$ was previously used to analyze the most representative set of words among a text corpus or the most informative genes from genetic sequences, e.g.~\cite{mp2009curmatrixdecomp, bmd2008feat_sel_pca}. 

\textbf{Datasets.} 1) \texttt{Synthetic} has a matrix $A$ of size $(k + n) \times (k + n)$ with rank $k$ and a fixed number $n$, where the top left $k \times k$ submatrix is the identity matrix multiplied by $n^{\frac{3}{2}}$, and the bottom right $n \times n$ submatrix has all $1$'s. The optimal $k$ columns consist of one of the last $n$ columns along with $k - 1$ of the first $k$ columns, incurring an error of $n^{\frac{3}{2}}$ in the $\ell_1$ norm. SVD, however, will not cover any of the last $n$ columns, and thus will get an $\ell_1$ error of $n^2$. We set $n = 1000$ in the experiments. 2) \texttt{TechTC}\footnote{\url{http://gabrilovich.com/resources/data/techtc/techtc300/techtc300.html}} contains $139$ documents processed in a bag-of-words representation with a dictionary of $18446$ words, which naturally results in a sparse matrix.
3) $\texttt{Gene}$\footnote{\url{https://archive.ics.uci.edu/ml/datasets/gene+expression+cancer+RNA-Seq}} contains $5000$ different RNA-Seq gene expressions from $400$ cancer patients, which gives a dense data matrix with $\geq 85\%$ non-zero entries.

\textbf{Setup.} For an input data matrix $A \in \mathbb{R}^{d \times n}$, we set the number of rows of our 1-stable (Cauchy) sketching matrix to be $0.5d$ in both settings. In the streaming setting, we set the batch size to be $5k$ and maintain a list of coresets of size $2k$.
In the distributed setting, we set the number of servers to be $5$ and each server sends the coordinator a coreset of size $2k$. 
For each of our experimetnts, we conduct $10$ random runs and report the mean $\ell_1$ error ratio $\frac{\min_V \|A_IV - A\|_1}{\|A\|_1}$ and the mean time (in seconds) to obtain $A_I$ along with one standard deviation, where $A_I$ is the output set of columns. Note that the input columns of the data matrix are randomly permuted in the streaming setting for each run.

\textbf{Implementation Details.} Our algorithms are implemented with Python Ray\footnote{\url{https://docs.ray.io/en/master/}}, a high-level framework for parallel and distributed computing, and are thus highly scalable. 
All the experiments are conducted on AWS EC2 c5a.8xlarge machines with 32 vCPUs and 64GB EBS memory.

\textbf{Results. } The results for the streaming setting and the distributed setting are presented in Figure~\ref{fig:exp_streaming} and Figure~\ref{fig:exp_distributed} respectively. We note that
\textbf{SVD} works in neither streaming nor distributed settings and thus the running time of \textbf{SVD} is not directly comparable to the other algorithms. 
The performance of \textbf{SVD} and \textbf{uniform} is highly dependent on the actual data and does not have worst case guarantees, while the performance of our algorithm is stable and gives relatively low $\ell_1$ error ratio across different datasets. \textbf{greedy} $k$-CSS$_{1,2}$ gives lower error ratio compared to \textbf{regular} $k$-CSS$_{1,2}$ as one would expect, but the time it takes significantly increases as the number of output columns increases, especially in the distributed setting, while \textbf{regular} $k$-CSS$_{1,2}$ takes comparable time to the \textbf{uniform} random baseline in most settings and is thus more scalable.

\begin{figure}[htp]
\subfloat{%
  \includegraphics[width=\columnwidth]{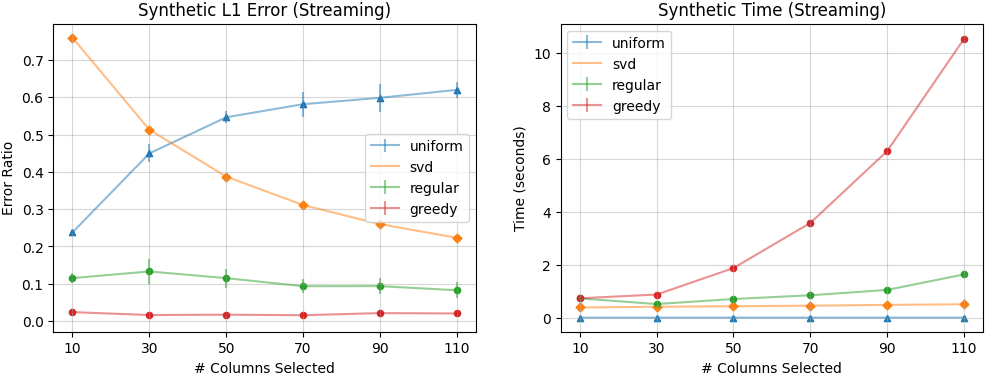}%
}\\
\subfloat{%
  \includegraphics[width=\columnwidth]{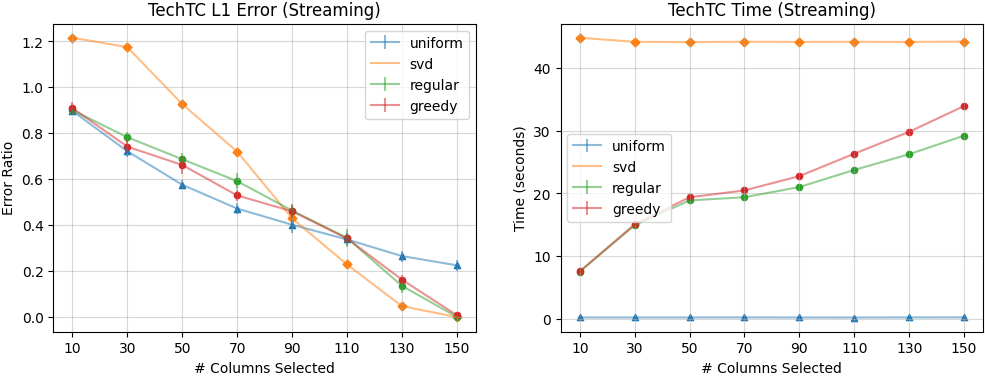}%
}\\
\subfloat{%
  \includegraphics[width=\columnwidth]{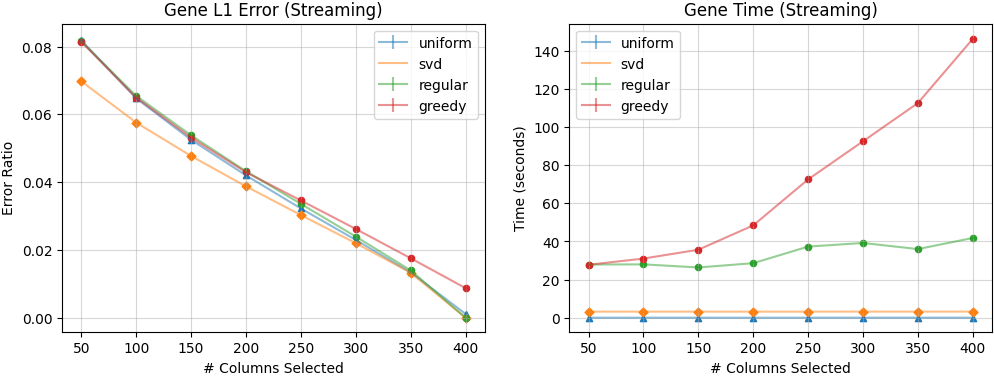}%
}
\caption{Streaming results.}
\label{fig:exp_streaming}
\end{figure}

\begin{figure}[htp]
\subfloat{%
  \includegraphics[width=\columnwidth]{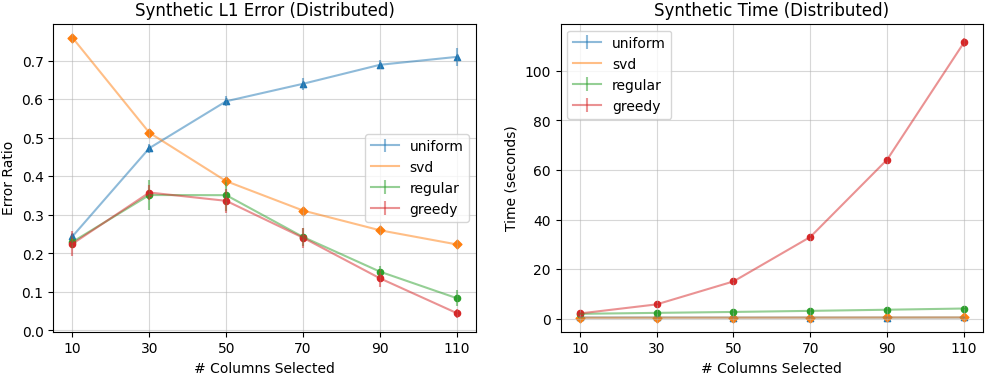}%
}\\
\subfloat{%
  \includegraphics[width=\columnwidth]{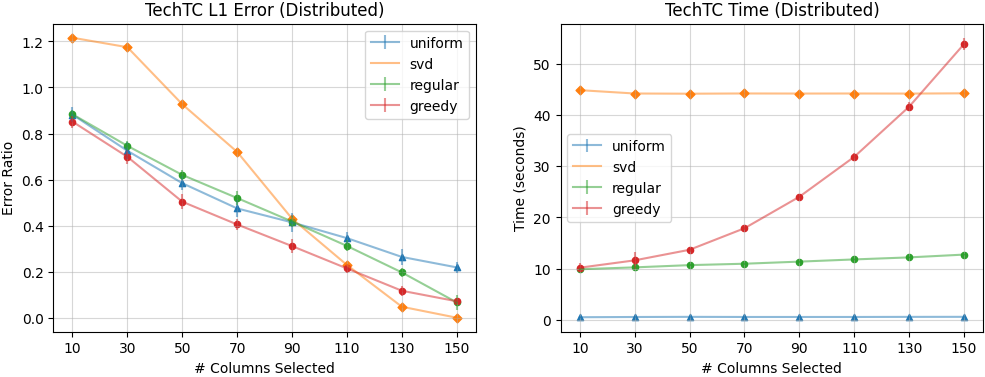}%
}\\
\subfloat{%
  \includegraphics[width=\columnwidth]{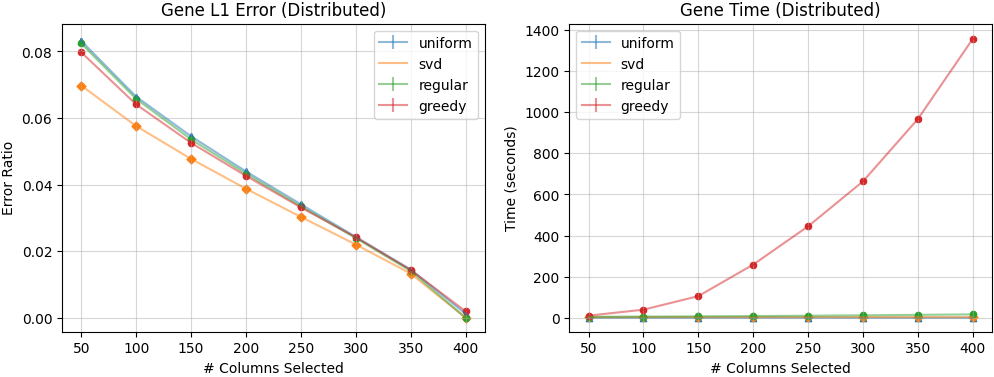}%
}
\caption{Distributed results.}
\label{fig:exp_distributed}
\end{figure}

\subsection{Robustness of the $\ell_1$ norm}
We further highlight the robustness of the $\ell_1$ norm loss to non-Gaussian noise compared to the Frobenius norm loss for $k$-CSS on an image classification task.
$k$-CSS was previously used as an active learning algorithm to select the most representative samples to acquire training labels in supervised learning, when acquiring such labels is  expensive~\cite{sjjrzy2011css_active_learning, kaushal2018learning_css}. 

We apply one Frobenius norm $k$-CSS$_2$ algorithm~\cite{bmd2008feat_sel_pca}, our regular $k$-CSS$_{1,2}$ algorithm (Section~\ref{subsection:lp2_css_algorithm}) and a random baseline to select a subset of $k$ training samples to train a linear regression model to classify images from the \texttt{COIL20}\footnote{\url{http://www.cs.columbia.edu/CAVE/software/softlib/coil-20.php}} dataset. \texttt{COIL20} contains a total of $1440$ images.
Each image has $32 \times 32$ pixels with $256$ gray levels per pixel. We randomly split the dataset into 80\% training set (1152 samples) and 20\% testing set (288 samples).
We then mask $40\%$ of the pixels of each training sample with a uniformly random noise value in $[0, 256)$. We conduct 20 random runs and report the average Mean Squared Error (MSE) and one standard deviation for each algorithm on the testing set. 

The results are summarized in Table~\ref{tab:exp_robustness}. $k$-CSS$_{1,2}$ gives a slightly lower average MSE score and a lower variance, compared to $k$-CSS$_2$ and the random baseline on noisy training data. This suggests that the $\ell_1$ norm loss is more robust compared to the Frobenius norm loss for $k$-CSS, which agrees with  previous observations from other algorithms and applications.

\begin{table}[]
    \centering
    \scalebox{0.75}{
    \begin{tabular}{|c|c|c|c|}
    \hline
     \backslashbox{$k$ samples}{Algorithm} & Random & $k$-CSS$_2$ & $k$-CSS$_{1,2}$ \\
     \hline
     200 & $16.29\pm 2.94$ & $17.64 \pm 2.97$ & $\mathbf{15.99 \pm 2.55}$\\
     \hline
     300 & $15.94 \pm 2.17$ & $17.92 \pm 3.59$ & $\mathbf{15.41 \pm 2.44}$ \\
     \hline
     400 & $14.27 \pm 1.68$ & $16.49 \pm 4.04$ & $\mathbf{13.84 \pm 1.55}$ \\
     \hline
     500 & $14.03 \pm 1.33$ & $14.59 \pm 1.85$ & $\mathbf{13.53 \pm 1.16}$ \\
     \hline
     600 & $14.45 \pm 1.91$ & $15.56 \pm 3.19$ & $\mathbf{13.68 \pm 1.40}$\\
     \hline
    \end{tabular}}
    \caption{Image classification results: average MSE and one std.}
    \label{tab:exp_robustness}
\end{table}

\section{Conclusion}
In this work, we give the first one-pass streaming algorithm for $k$-CSS$_p$ ($1 \leq p < 2$) in the column-update model and the first one-round distributed protocol in the column-partition model. Both of our algorithms achieve $\widetilde{O}(k^{1/p - 1/2})$-approximation to the optimal column subset. The streaming algorithm uses nearly optimal space complexity of $\widetilde{O}(kd)$, and the distributed protocol uses nearly optimal $\widetilde{O}(sdk)$ communication cost. We introduce novel analysis techniques for $k$-CSS. To achieve a good approximation factor, we use dense $p$-stable sketching and work with the $\ell_{p, 2}$ norm, which enables us to use an efficient construction of strong coresets and an $O(1)$-approximation bi-criteria $k$-CSS$_{p, 2}$ algorithm as a subroutine of our algorithms. We further propose a greedy alternative for $k$-CSS$_{p, 2}$ and show the first additive error upper bound. Our experimental results confirm that our algorithms give stable low $\ell_1$ error in both distributed and streaming settings. We further demonstrate the robustness of the $\ell_1$ norm loss for $k$-CSS. 

\section*{Acknowledgements}
D. Woodruff would like to thank partial support from NSF grant No. CCF-1815840, Office of Naval Research grant N00014-18-1-256, and a Simons Investigator Award. A. Mahankali would like to thank partial support from the SURF award from CMU's Undergraduate Research Office.

\bibliographystyle{unsrt}
\bibliography{main}

\begin{thebibliography}{10}

\bibitem{cw13_lra_nnz_time}
Kenneth~L. Clarkson and David~P. Woodruff.
\newblock Low rank approximation and regression in input sparsity time.
\newblock In Dan Boneh, Tim Roughgarden, and Joan Feigenbaum, editors, {\em
  Symposium on Theory of Computing Conference, STOC'13, Palo Alto, CA, USA,
  June 1-4, 2013}, pages 81--90. {ACM}, 2013.

\bibitem{w14_sketching_as_a_tool}
David~P. Woodruff.
\newblock Sketching as a tool for numerical linear algebra.
\newblock {\em Found. Trends Theor. Comput. Sci.}, 10(1-2):1--157, 2014.

\bibitem{guruswami2012optimal}
Venkatesan Guruswami and Ali~Kemal Sinop.
\newblock Optimal column-based low-rank matrix reconstruction.
\newblock In {\em Proceedings of the twenty-third annual ACM-SIAM symposium on
  Discrete Algorithms}, pages 1207--1214. SIAM, 2012.

\bibitem{boutsidis2014near}
Christos Boutsidis, Petros Drineas, and Malik Magdon-Ismail.
\newblock Near-optimal column-based matrix reconstruction.
\newblock {\em SIAM Journal on Computing}, 43(2):687--717, 2014.

\bibitem{boutsidis2017optimal}
Christos Boutsidis and David~P Woodruff.
\newblock Optimal cur matrix decompositions.
\newblock {\em SIAM Journal on Computing}, 46(2):543--589, 2017.

\bibitem{boutisidis18}
Christos Boutsidis, Michael~W. Mahoney, and Petros Drineas.
\newblock An improved approximation algorithm for the column subset selection
  problem.
\newblock {\em CoRR}, abs/0812.4293, 2008.

\bibitem{operator_norm_lra}
Nathan Halko, Per-Gunnar Martinsson, and Joel~A Tropp.
\newblock Finding structure with randomness: Probabilistic algorithms for
  constructing approximate matrix decompositions.
\newblock {\em SIAM review}, 53(2):217--288, 2011.

\bibitem{sketching_as_a_tool}
David~P. Woodruff.
\newblock Sketching as a tool for numerical linear algebra.
\newblock {\em Foundations and Trends in Theoretical Computer Science},
  10(1-2):1--157, 2014.

\bibitem{swz2017lral1norm}
Zhao Song, David~P. Woodruff, and Peilin Zhong.
\newblock Low rank approximation with entrywise l1-norm error.
\newblock In {\em Proceedings of the 49th Annual ACM SIGACT Symposium on Theory
  of Computing}, STOC 2017, page 688–701, New York, NY, USA, 2017.
  Association for Computing Machinery.

\bibitem{cgklpw2017lplra}
Flavio Chierichetti, Screenivas Gollapudi, Ravi Kumar, Silvio Lattanzi, Rina
  Panigrahy, and David~P. Woodruff.
\newblock {Algorithms for $\ell_p$ Low-Rank Approximation}.
\newblock 2017.

\bibitem{dwzzr2019opt_anal_css_lplra}
Chen Dan, Hong Wang, Hongyang Zhang, Yuchen Zhou, and Pradeep~K Ravikumar.
\newblock Optimal analysis of subset-selection based l\_p low-rank
  approximation.
\newblock In H.~Wallach, H.~Larochelle, A.~Beygelzimer, F.~d\textquotesingle
  Alch\'{e}-Buc, E.~Fox, and R.~Garnett, editors, {\em Advances in Neural
  Information Processing Systems 32}, pages 2541--2552. Curran Associates,
  Inc., 2019.

\bibitem{bbbklw19_ptas_for_lp_lra}
Frank Ban, Vijay Bhattiprolu, Karl Bringmann, Pavel Kolev, Euiwoong Lee, and
  David~P. Woodruff.
\newblock A {PTAS} for $\ell_p$-low rank approximation.
\newblock In Timothy~M. Chan, editor, {\em Proceedings of the Thirtieth Annual
  {ACM-SIAM} Symposium on Discrete Algorithms, {SODA} 2019, San Diego,
  California, USA, January 6-9, 2019}, pages 747--766. {SIAM}, 2019.

\bibitem{mw2019opt_l1_css_lra}
Arvind~V Mahankali and David~P Woodruff.
\newblock Optimal $\ell_1$ column subset selection and a fast ptas for low rank
  approximation.
\newblock In {\em Proceedings of the 2021 ACM-SIAM Symposium on Discrete
  Algorithms (SODA)}, pages 560--578. SIAM, 2021.

\bibitem{kk05_sfm_problem}
Qifa Ke and Takeo Kanade.
\newblock Robust $l_1$ norm factorization in the presence of outliers and
  missing data by alternative convex programming.
\newblock In {\em 2005 {IEEE} Computer Society Conference on Computer Vision
  and Pattern Recognition {(CVPR} 2005), 20-26 June 2005, San Diego, CA,
  {USA}}, pages 739--746. {IEEE} Computer Society, 2005.

\bibitem{yzd12_image_denoising}
Linbin Yu, Miao Zhang, and Chris H.~Q. Ding.
\newblock An efficient algorithm for l1-norm principal component analysis.
\newblock In {\em 2012 {IEEE} International Conference on Acoustics, Speech and
  Signal Processing, {ICASSP} 2012, Kyoto, Japan, March 25-30, 2012}, pages
  1377--1380. {IEEE}, 2012.

\bibitem{dk03_pass_efficient}
Petros Drineas and Ravi Kannan.
\newblock Pass efficient algorithms for approximating large matrices.
\newblock In {\em Proceedings of the Fourteenth Annual {ACM-SIAM} Symposium on
  Discrete Algorithms, January 12-14, 2003, Baltimore, Maryland, {USA}}, pages
  223--232. {ACM/SIAM}, 2003.

\bibitem{cw09_nla_in_streaming_model}
Kenneth~L. Clarkson and David~P. Woodruff.
\newblock Numerical linear algebra in the streaming model.
\newblock In Michael Mitzenmacher, editor, {\em Proceedings of the 41st Annual
  {ACM} Symposium on Theory of Computing, {STOC} 2009, Bethesda, MD, USA, May
  31 - June 2, 2009}, pages 205--214. {ACM}, 2009.

\bibitem{liberty13_row_update_PCA}
Edo Liberty.
\newblock Simple and deterministic matrix sketching.
\newblock In Inderjit~S. Dhillon, Yehuda Koren, Rayid Ghani, Ted~E. Senator,
  Paul Bradley, Rajesh Parekh, Jingrui He, Robert~L. Grossman, and Ramasamy
  Uthurusamy, editors, {\em The 19th {ACM} {SIGKDD} International Conference on
  Knowledge Discovery and Data Mining, {KDD} 2013, Chicago, IL, USA, August
  11-14, 2013}, pages 581--588. {ACM}, 2013.

\bibitem{gp14_row_update_PCA}
Mina Ghashami and Jeff~M. Phillips.
\newblock Relative errors for deterministic low-rank matrix approximations.
\newblock In Chandra Chekuri, editor, {\em Proceedings of the Twenty-Fifth
  Annual {ACM-SIAM} Symposium on Discrete Algorithms, {SODA} 2014, Portland,
  Oregon, USA, January 5-7, 2014}, pages 707--717. {SIAM}.

\bibitem{w14_row_update_lower_bound_frobenius}
David~P. Woodruff.
\newblock Low rank approximation lower bounds in row-update streams.
\newblock In Zoubin Ghahramani, Max Welling, Corinna Cortes, Neil~D. Lawrence,
  and Kilian~Q. Weinberger, editors, {\em Advances in Neural Information
  Processing Systems 27: Annual Conference on Neural Information Processing
  Systems 2014, December 8-13 2014, Montreal, Quebec, Canada}, pages
  1781--1789, 2014.

\bibitem{abfmrz2016greedycssfrobenius}
Jason Altschuler, Aditya Bhaskara, Gang Fu, Vahab Mirrokni, Afshin
  Rostamizadeh, and Morteza Zadimoghaddam.
\newblock Greedy column subset selection: New bounds and distributed
  algorithms.
\newblock In {\em Proceedings of the 33rd International Conference on
  International Conference on Machine Learning - Volume 48}, ICML’16, page
  2539–2548. JMLR.org, 2016.

\bibitem{BWZ16}
Christos Boutsidis, David~P. Woodruff, and Peilin Zhong.
\newblock Optimal principal component analysis in distributed and streaming
  models.
\newblock In Daniel Wichs and Yishay Mansour, editors, {\em Proceedings of the
  48th Annual {ACM} {SIGACT} Symposium on Theory of Computing, {STOC} 2016,
  Cambridge, MA, USA, June 18-21, 2016}, pages 236--249. {ACM}, 2016.

\bibitem{ww19lpobliviousemb}
Ruosong Wang and David~P. Woodruff.
\newblock Tight bounds for $\ell_p$ oblivious subspace embeddings.
\newblock In {\em Proceedings of the Thirtieth Annual {ACM-SIAM} Symposium on
  Discrete Algorithms, {SODA} 2019, San Diego, California, USA, January 6-9,
  2019}, pages 1825--1843, 2019.

\bibitem{swz2019zeroonelawcss}
Zhao Song, David~P. Woodruff, and Peilin Zhong.
\newblock Towards a zero-one law for column subset selection.
\newblock In {\em Advances in Neural Information Processing Systems 32: Annual
  Conference on Neural Information Processing Systems 2019, NeurIPS 2019, 8-14
  December 2019, Vancouver, BC, Canada}, pages 6120--6131, 2019.

\bibitem{l1nphard}
Nicolas Gillis and Stephen~A. Vavasis.
\newblock On the complexity of robust {PCA} and $\ell_1$-norm low-rank matrix
  approximation.
\newblock {\em CoRR}, abs/1509.09236, 2015.

\bibitem{sw2018strongcoreset}
Christian Sohler and David~P. Woodruff.
\newblock Strong coresets for k-median and subspace approximation: Goodbye
  dimension.
\newblock In Mikkel Thorup, editor, {\em 59th {IEEE} Annual Symposium on
  Foundations of Computer Science, {FOCS} 2018, Paris, France, October 7-9,
  2018}, pages 802--813. {IEEE} Computer Society, 2018.

\bibitem{mcgregor14_graph_stream_survey}
Andrew McGregor.
\newblock Graph stream algorithms: A survey.
\newblock {\em SIGMOD Rec.}, 43(1):9–20, May 2014.

\bibitem{fegk13_distributed_css}
Ahmed~K. Farahat, Ahmed Elgohary, Ali Ghodsi, and Mohamed~S. Kamel.
\newblock Distributed column subset selection on mapreduce.
\newblock In Hui Xiong, George Karypis, Bhavani~M. Thuraisingham, Diane~J.
  Cook, and Xindong Wu, editors, {\em 2013 {IEEE} 13th International Conference
  on Data Mining, Dallas, TX, USA, December 7-10, 2013}, pages 171--180. {IEEE}
  Computer Society, 2013.

\bibitem{lbkw14_distributed_pca}
Yingyu Liang, Maria{-}Florina Balcan, Vandana Kanchanapally, and David~P.
  Woodruff.
\newblock Improved distributed principal component analysis.
\newblock In Zoubin Ghahramani, Max Welling, Corinna Cortes, Neil~D. Lawrence,
  and Kilian~Q. Weinberger, editors, {\em Advances in Neural Information
  Processing Systems 27: Annual Conference on Neural Information Processing
  Systems 2014, December 8-13 2014, Montreal, Quebec, Canada}, pages
  3113--3121, 2014.

\bibitem{BLSW015}
Maria{-}Florina Balcan, Yingyu Liang, Le~Song, David~P. Woodruff, and Bo~Xie.
\newblock Distributed kernel principal component analysis.
\newblock {\em CoRR}, abs/1503.06858, 2015.

\bibitem{BSW016}
Maria{-}Florina Balcan, Yingyu Liang, Le~Song, David~P. Woodruff, and Bo~Xie.
\newblock Communication efficient distributed kernel principal component
  analysis.
\newblock In Balaji Krishnapuram, Mohak Shah, Alexander~J. Smola, Charu~C.
  Aggarwal, Dou Shen, and Rajeev Rastogi, editors, {\em Proceedings of the 22nd
  {ACM} {SIGKDD} International Conference on Knowledge Discovery and Data
  Mining, San Francisco, CA, USA, August 13-17, 2016}, pages 725--734. {ACM},
  2016.

\bibitem{cp2015lewisweights}
Michael~B. Cohen and Richard Peng.
\newblock Lp row sampling by lewis weights.
\newblock In {\em Proceedings of the Forty-Seventh Annual ACM Symposium on
  Theory of Computing}, STOC ’15, page 183–192, New York, NY, USA, 2015.
  Association for Computing Machinery.

\bibitem{cw2015subsapproxl2}
Kenneth~L. Clarkson and David~P. Woodruff.
\newblock Input sparsity and hardness for robust subspace approximation.
\newblock In {\em Proceedings of the 2015 IEEE 56th Annual Symposium on
  Foundations of Computer Science (FOCS)}, FOCS ’15, page 310–329, USA,
  2015. IEEE Computer Society.

\bibitem{cms76_simulating_pstable}
J.~M. Chambers, C.~L. Mallows, and B.~W. Stuck.
\newblock A method for simulating stable random variables.
\newblock {\em Journal of the American Statistical Association},
  71(354):340--344, 1976.

\bibitem{blm1989zenotope}
J.~Bourgain, J.~Lindenstrauss, and V.~Milman.
\newblock {Approximation of zonoids by zonotopes}.
\newblock {\em Acta Mathematica}, 162(none):73 -- 141, 1989.

\bibitem{swz19avgcasecssl1}
Zhao Song, David~P. Woodruff, and Peilin Zhong.
\newblock Average case column subset selection for entrywise $\ell_1$-norm
  loss.
\newblock In {\em Advances in Neural Information Processing Systems 32: Annual
  Conference on Neural Information Processing Systems 2019, NeurIPS 2019, 8-14
  December 2019, Vancouver, BC, Canada}, pages 10111--10121, 2019.

\bibitem{mp2009curmatrixdecomp}
Michael Mahoney and Petros Drineas.
\newblock Cur matrix decompositions for improved data analysis.
\newblock {\em Proceedings of the National Academy of Sciences of the United
  States of America}, 106:697--702, 02 2009.

\bibitem{bmd2008feat_sel_pca}
Christos Boutsidis, Michael~W. Mahoney, and Petros Drineas.
\newblock Unsupervised feature selection for principal components analysis.
\newblock In {\em Proceedings of the 14th ACM SIGKDD International Conference
  on Knowledge Discovery and Data Mining}, KDD '08, page 61–69, New York, NY,
  USA, 2008. Association for Computing Machinery.

\bibitem{sjjrzy2011css_active_learning}
Jianfeng Shen, Bin Ju, Tao Jiang, Jingjing Ren, Miao Zheng, Chengwei Yao, and
  Lanjuan Li.
\newblock Column subset selection for active learning in image classification.
\newblock {\em Neurocomputing}, 74:3785--3792, 11 2011.

\bibitem{kaushal2018learning_css}
Vishal Kaushal, Anurag Sahoo, Khoshrav Doctor, Narasimha Raju, Suyash Shetty,
  Pankaj Singh, Rishabh Iyer, and Ganesh Ramakrishnan.
\newblock Learning from less data: Diversified subset selection and active
  learning in image classification tasks, 2018.

\bibitem{STOC-2013-MengM}
Xiangrui Meng and Michael~W. Mahoney.
\newblock {Low-distortion subspace embeddings in input-sparsity time and
  applications to robust linear regression}.
\newblock In {\em {Proceedings of the 45th Annual ACM Symposium on Theory of
  Computing}}, pages 91--100. {ACM}, 2013.

\bibitem{cww2019tukey_regression}
Kenneth Clarkson, Ruosong Wang, and David Woodruff.
\newblock Dimensionality reduction for tukey regression.
\newblock In Kamalika Chaudhuri and Ruslan Salakhutdinov, editors, {\em
  Proceedings of the 36th International Conference on Machine Learning},
  volume~97 of {\em Proceedings of Machine Learning Research}, pages
  1262--1271. PMLR, 09--15 Jun 2019.

\bibitem{ddhkm2008lp_sampling}
Anirban Dasgupta, Petros Drineas, Boulos Harb, Ravi Kumar, and Michael~W.
  Mahoney.
\newblock Sampling algorithms and coresets for $\ell_p$ regression.
\newblock In {\em Proceedings of the Nineteenth Annual ACM-SIAM Symposium on
  Discrete Algorithms}, SODA '08, page 932–941, USA, 2008. Society for
  Industrial and Applied Mathematics.

\bibitem{ddhkm09_well_conditioned_bases}
Anirban Dasgupta, Petros Drineas, Boulos Harb, Ravi Kumar, and Michael~W.
  Mahoney.
\newblock Sampling algorithms and coresets for $\ell_p$ regression.
\newblock {\em {SIAM} J. Comput.}, 38(5):2060--2078, 2009.

\bibitem{nn2013OSNAP}
Jelani Nelson and Huy~L. Nguyen.
\newblock Osnap: Faster numerical linear algebra algorithms via sparser
  subspace embeddings.
\newblock In {\em Proceedings of the 2013 IEEE 54th Annual Symposium on
  Foundations of Computer Science}, FOCS ’13, page 117–126, USA, 2013. IEEE
  Computer Society.

\bibitem{pvz_dvoretzky_lp_norm}
Grigoris Paouris, Petros Valettas, and Joel Zinn.
\newblock Random version of dvoretzky’s theorem in $\ell_p^n$.
\newblock {\em Stochastic Processes and their Applications}, 127(10):3187 --
  3227, 2017.

\bibitem{ww2019lpobliviousemb}
Ruosong Wang and David~P. Woodruff.
\newblock Tight bounds for lp oblivious subspace embeddings.
\newblock In {\em Proceedings of the Thirtieth Annual ACM-SIAM Symposium on
  Discrete Algorithms}, SODA '19, page 1825–1843, USA, 2019. Society for
  Industrial and Applied Mathematics.

\bibitem{ycrm2018wsgd}
Jiyan Yang, Yin-Lam Chow, Christopher R{{\'e}}, and Michael~W. Mahoney.
\newblock Weighted sgd for $\ell_p$ regression with randomized preconditioning.
\newblock {\em Journal of Machine Learning Research}, 18(211):1--43, 2018.

\bibitem{mbkvk15_lazier_than_lazy_greedy}
Baharan Mirzasoleiman, Ashwinkumar Badanidiyuru, Amin Karbasi, Jan Vondr\'{a}k,
  and Andreas Krause.
\newblock Lazier than lazy greedy.
\newblock In {\em Proceedings of the Twenty-Ninth AAAI Conference on Artificial
  Intelligence}, AAAI'15, page 1812–1818. AAAI Press, 2015.

\end{thebibliography}

\newpage
\begin{appendices}

\section{Proof of Preliminaries (Section 3)}

\subsection{Norms (Lemma 1)}
\begin{customlemma}{1}
\label{lemma:norm}
For a matrix $A \in \mathbb{R}^{d \times n}$ and $p \in [1, 2)$,
$\|A\|_{p,2} \leq \|A\|_p \leq d^{\frac{1}{p}-\frac{1}{2}} \|A\|_{p,2}$.
\end{customlemma}

\begin{proof}
Let $x \in\mathbb{R}^{d}$.
For $0 < p < r$,
$$\|x\|_r\leq \|x\|_p\leq d^{\frac{1}{p}-\frac{1}{r}}\|x\|_r$$
Let $r=2$. Then we have
$$\|x\|_2\leq \|x\|_p\leq d^{\frac{1}{p}-\frac{1}{2}}\|x\|_2$$
Note that $\|A\|_{p,2}=\left(\sum_j \|A_{*j}\|_2^p \right)^\frac{1}{p}$ and $\|A\|_p=\left(\sum_j \|A_{*j}\|_p^p \right)^\frac{1}{p}$.

Therefore,
$$\|A\|_{p,2}=\left(\sum_j \|A_{*j}\|_2^p \right)^\frac{1}{p}\leq \left(\sum_j \|A_{*j}\|_p^p \right)^\frac{1}{p}=\|A\|_p$$
and
$$\|A\|_p=\left(\sum_j \|A_{*j}\|_p^p \right)^\frac{1}{p}\leq d^{\frac{1}{p}-\frac{1}{2}} \left(\sum_j \|A_{*j}\|_2^p \right)^\frac{1}{p}=d^{\frac{1}{p}-\frac{1}{2}}\|A\|_{p,2}
$$
\end{proof}

\subsection{Sketched Error Lower Bound (Lemma 2)}

We show a lower bound on the approximation error for a sketched subset of columns, $\|SA_TV - SA\|_p$, in terms of $\|A_TV - A\|_p$. The lower bound holds simultaneously for any arbitrary subset $A_T$ of chosen columns, and for any arbitrary right factor $V$. 

We begin the proof by first showing that applying a dense $p$-stable sketch to a vector will not shrink its $p$-norm. This is done in  \textbf{Lemma}~\ref{lemma:pstablenocontraction}. We further observe that although $p$-stable random variables are heavy-tailed, we can still bound their tail probabilities by applying Lemma 9 from~\cite{STOC-2013-MengM}. We note this in \textbf{Lemma}~\ref{lemma:pstableuppertail}.
Note that the $X_i$'s do not need to be independent in this lemma.

Equipped with \textbf{Lemma}~\ref{lemma:pstablenocontraction}, \textbf{Lemma}~\ref{lemma:pstableuppertail} and a net argument, 
we can now establish a lower bound on $\|SA_TV - SA\|_p$. We first show in \textbf{Lemma}~\ref{lemma:cssboundvector} that, with high probability, for any arbitrarily selected subset $A_T$ of columns and for an arbitrary column $A_{*j}$, the error incurred to fit $SA_{*j}$ using the columns of $SA_T$ is no less than the error incurred to fit $A_{*j}$ using the columns of $A_T$. We then apply a union bound over all subsets $T \subset [n]$ and columns $j \in [n]$ to conclude our lower bound in \textbf{Lemma}~\ref{lemma:pstable}.

\begin{customlemma}{2.1}(No Contraction of $p$-stable Sketch)
\label{lemma:pstablenocontraction}
    Given a matrix $S \in \mathbb{R}^{t \times m}$ whose entries are i.i.d. p-stable random variables rescaled by $\Theta\left(\frac{1}{t^{\frac{1}{p}}}\right)$, where $1 \leq p < 2$, for any fixed $y \in \mathbb{R}^m$, with probability $1 - \frac{1}{e^{t}}$, the following holds:
    \begin{align*}
        \|Sy\|_p \geq \|y\|_p
    \end{align*}
\end{customlemma}

\begin{proof}
By $p$-stability, we have $\|Sy\|_p^p = \sum_{i=1}^t \Big(\|y\|_p\frac{|Z_i|}{t^{\frac{1}{p}}}\Big)^{p}$, where the $Z_i$ are i.i.d. $p$-stable random variables. Since $Pr[|Z_i| = \Omega(1)] > \frac{1}{2}$, by applying a Chernoff bound (to the indicators $1_{|Z_i| \geq C}$ for a sufficiently small constant $C$), we have $\sum_{i=1}^t |Z_i|^p = \Omega(t)$ with probability $1 - \frac{1}{e^t}$. Therefore, with probability $1 - \frac{1}{e^t}$, $\|Sy\|_p \geq \|y\|_p$.
\end{proof}

\begin{customlemma}{2.2}(Upper Tail Inequality for $p$-stable Distributions) \label{lemma:pstableuppertail}
Let $p \in (1, 2)$, and $m > 3$. For $i \in [m]$, let $X_i$ be a standard $p$-stable random variable, and let $\gamma_i > 0$ and $\gamma = \sum_{i = 1}^m \gamma_i$. Let $X = \sum_{i = 1}^m \gamma_i |X_i|^p$. Then, for any $t \geq 1$, $Pr[X \geq t\alpha_p\gamma] \leq \frac{2\log(mt)}{t}$, where $\alpha_p > 0$ is a constant that is at most $2^{p - 1}$.
\end{customlemma}

\begin{proof}
Lemma 9 from~\cite{STOC-2013-MengM} for $p \in (1, 2)$.
\end{proof}

\begin{customlemma}{2.3}(No Contraction for All Sketched Subsets and Columns) \label{lemma:cssboundvector}
Let $A \in \R^{d \times n}$, and $k \in \N$. Let $t = k \cdot \poly(\log nd)$, and let $S \in \R^{t \times d}$ be a matrix whose entries are i.i.d. standard $p$-stable random variables, rescaled by $\Theta(1/t^{\frac{1}{p}})$. Finally, let $m = k \cdot \poly(\log k)$. Then, with probability $1 - \frac{1}{\poly(nd)}$, for all $T \subset [n]$ with $|T| = m$, for all $j \in [n]$, and for all $y \in \R^{|T|}$,
$$\|A_Ty - A_{*j}\|_p \leq \|S(A_Ty - A_{*j})\|_p$$
\end{customlemma}

\begin{proof}
\textbf{Step 1:}
We first extend \textbf{Lemma}~\ref{lemma:pstablenocontraction} and use a net argument to show that applying a $p$-stable sketching matrix $S \in \mathbb{R}^{t \times d}$ will not shrink the norm of \textit{any} vector, i.e. $\|Sy\|_p \geq \|y\|_p$ simultaneously for \textit{all} $y$ in the column span of $[A_T, A_j] =: A_{T, j}$, for any fixed $T \subset [n]$ with $|T| = k \cdot \poly(\log k)$, and $j \in [n]$.

For our net argument, we begin by showing that with high probability all entries of $S$ are bounded. Let $D > 0$, which we will choose appropriately later. For convenience, let $\widetilde{S} \in \R^{t \times n}$ be equal to $S$ without the rescaling by $\Theta(1/t^{1/p})$ (that is, the entries of $\widetilde{S}$ are i.i.d. $p$-stable random variables, and the entries of $S$ are those of $\widetilde{S}$ but rescaling by $\Theta(1/t^{1/p})$. Consider the following two cases:

\textbf{Case 1:} $p = 1$: The $\widetilde{S}_{ij}$ are standard Cauchy random variables. Consider the half-Cauchy random variables $X_{i, j} = |S_{i, j}|$. The cumulative distribution function of a half-Cauchy random variable $X$ is $F(x) = \int_0^x \frac{2}{\pi(t^2 + 1)} dt = 1 - \Theta(\frac{1}{x})$. Thus, for any $i \in [t]$ and $j \in [d]$, $\Pr[|\widetilde{S_{ij}}| \leq D] = 1 - \Theta(\frac{1}{D})$, and $\Pr[|S_{ij}| \leq D] = 1 - \Theta(\frac{1}{tD}) \geq 1 - \Theta(\frac{1}{D})$.

\textbf{Case 2:} $p \in (1, 2)$: We apply the upper tail bound for $p$-stable random variables in \textbf{Lemma}~\ref{lemma:pstableuppertail}. For any fixed $i \in [t]$ and $j \in [d]$, $\Pr[|\widetilde{S}_{ij}|^p \leq D^p] \geq 1 - \Theta(\frac{\log D}{D^p})$, which implies that $\Pr[|\widetilde{S_{ij}}| \leq D] \geq 1 - \Theta(\frac{1}{D})$, since $p > 1$. In addition, by the same argument, $\Pr[|S_{ij}| \leq D] = \Pr[|\widetilde{S}_{ij}| \leq t^{1/p}D] = 1 - \Theta(\frac{1}{t^{1/p}D}) \geq 1 - \Theta(\frac{1}{D})$.

Therefore, for $p \in [1, 2)$, if we let $\calE_1$ be the event that for all $i \in [t]$ and $j \in [m]$, we simultaneously have $|S_{ij}| \leq D$, then by a union bound over all the entries in $S$, $\Pr[\calE_1] \geq 1 - \Theta(\frac{td}{D})$. In particular, if we choose $D = \poly(nd)$, then $\calE_1$ occurs with probability at least $1 - \frac{1}{\poly(nd)}$. Note that if $\calE_1$ occurs, then this implies that for all $y \in \R^d$,
\begin{align*}
\|Sy\|_p = \Big(\sum_{i = 1}^t \Big| \sum_{j = 1}^d S_{ij}y_j \Big|^p \Big)^{1/p} \leq \Big(\sum_{i = 1}^t D^p \cdot \Big| \sum_{j = 1}^d y_j \Big|^p \Big)^{1/p} \leq Dt^{1/p} \|y\|_1 \leq D\poly(d) \|y\|_p
\end{align*}

Consider the unit $\ell_p$ ball $B = \{y \in \mathbb{R}^d: \|y\|_p = 1,\, \exists z \in \R^m \text{ s.t. } y = A_{T, j}z\}$ in the column span of $A_{T, j}$. A subset $\mathcal{N} \subset B$ is a $\gamma$-net for $B$ if for all $y \in B$ there exists some $u \in \calN$ such that $\|y - u\|_p \leq \gamma$, for some distance $\gamma > 0$. There exists such a net $\mathcal{N}$ for $B$ of size $|\mathcal{N}| = (\frac{1}{\gamma})^{O(m)}$ by a standard greedy construction, since the column span of $A_{T, j}$ has dimension at most $m + 1$. Let us choose $\gamma$ as follows. First let $K = \poly(nd)$ such that $\|Sy\|_p \leq K \|y\|_p$ (recall that $\|Sy\|_p \leq D \poly(d) \|y\|_p$ if $\calE_1$ holds). Then, we choose $\gamma = \frac{1}{m^2 K}$. Thus, $|\calN| \leq (m^2 K)^{O(m)} = 2^{O(m \log (nd))}$.

Define the event $\calE_2(T, j)$ (here the $T, j$ in parentheses signify that $\calE_2(T, j)$ is defined in terms of $T$ and $j$) as follows: for all $y \in \calN$ simultaneously, $\|Sy\|_p \geq \|y\|_p$. (Note that $\calE_2(T, j)$ depends on $T, j$ since $\calN$ is a net for the column span of $A_{T, j}$.) By applying \textbf{Lemma}~\ref{lemma:pstablenocontraction}, and a union bound over all vectors $y \in \mathcal{N}$, we find that for all $y \in \mathcal{N}$ simultaneously, $\|Sy\|_p \geq \|y\|_p$ with probability at least $1 - \frac{|\calN|}{e^t} = 1 - \frac{2^{O(m \log (nd))}}{e^t}$ --- in other words, $\calE_2(T, j)$ has probability at least $1 - \frac{2^{O(m \log (nd))}}{e^t}$.

Now, consider an arbitrary unit vector $x \in B$. There exists some $y \in \calN$ such that $\|x - y\|_p \leq \gamma = \frac{1}{m^2 K}$. If we assume that both $\calE_1$ and $\calE_2(T, j)$ hold, then the following holds as well:
\begin{align*}
    \|Sx\|_p &\geq \|Sy\|_p - \|S(x - y)\|_p &\text{Triangle Inequality}\\
    &\geq \|y\|_p - \|S(x - y)\|_p &\text{By event $\mathcal{E}_2(T, j)$}\\
    &\geq \|y\|_p - K\|(x - y)\|_p &\text{Implication of event $\mathcal{E}_1$}\\
    &\geq \|y\|_p - K\gamma &\text{By $\|x-y\|_p\leq \gamma$}\\
    &= \|y\|_p - O\Big(\frac{1}{m^2}\Big)\\
    &= \|x\|_p - O\Big(\frac{1}{m^2}\Big) &\text{$\|x\|_p = \|y\|_p = 1$}
\end{align*}

For a sufficiently large $m$, $O(\frac{1}{m^2})$ is at most $\frac{1}{2}$, and thus $\frac{\|x\|_p}{2} = \frac{1}{2}\geq \frac{1}{m^2}$. This implies $\|Sx\|_p \geq \|x\|_p - \frac{\|x\|_p}{2} = \frac{\|x\|_p}{2}$. We can rescale $S$ by a factor of $2$ so that $\|Sx\|_p \geq \|x\|_p$.

We have shown that $\|Sy\|_p \geq \|y\|_p$ holds simultaneously for \textit{all} unit vectors $y$ in the column span of $A_{T, j}$, conditioning on $\calE_1$ and $\calE_2(T, j)$. By linearity, we conclude that $\|Sy\|_p \geq \|y\|_p$ ($1 \leq p < 2$) holds simultaneously for \textit{all} $y$ in the column span of $A_{T, j}$, conditioning on $\calE_1$ and $\calE_2(T, j)$.

\textbf{Step 2:}
Next, we apply a union bound over all possible subsets $T \subset [n]$ of chosen columns from $A$ and all possible single columns $A_{*j}$ for $j \in [n]$, to argue that $\|S(A_T y - A_{*j})\|_p \geq \|A_T y - A_{*j}\|_p$ holds simultaneously for all $y \in \mathbb{R}^{|T|}$ and all $T \subset [n]$ with $|T| = m = k \cdot \poly(\log k)$ and $j \subset [n]$ with high probability.

In \textbf{Step 1}, we showed that $\calE_2(T, j)$ fails with probability $\frac{2^{O(m \log (nd))}}{e^t}$, for any fixed $T, j$. Thus, if we define $\calE_{2, all}$ to be the event that $\calE_2(T, j)$ holds for all $T, j$ (in other words, $\calE_{2, all} = \bigcap_{T, j} \calE_2(T, j)$, then the failure probability of $\calE_{2, all}$ is at most
\begin{align*}
\frac{2^{O(m \log (nd))}}{e^t} \cdot \binom{n}{m} \cdot d \leq \frac{2^{O(m \log (nd))}}{e^t} \cdot n^{O(m)} \cdot d = \frac{2^{O(m \log (nd))}}{e^t}
\end{align*}

In summary, if we let $D = \poly(nd)$, then $\calE_1$ succeeds with probability $1 - \Theta(\frac{td}{D}) \geq 1 - \frac{1}{\poly(nd)}$. In addition, if we let $D = \poly(nd)$, and let $K = D \poly(d) = \poly(nd)$, then $\calE_{2, all}$ holds with probability $1 - \frac{2^{O(m \log (nd)}}{e^t}$. Note that if both $\calE_1$ and $\calE_{2, all}$ hold, then $\calE_1$ and $\calE_2(T, j)$ hold for all $T, j$, meaning that
$$\|A_Ty - A_{*j}\|_p \leq \|S(A_T y - A_{*j})\|_p$$
for all $T \subset [n]$ with $|T| = m = k \cdot \poly(\log k)$, $j \in [n]$ and $y \in \R^{|T|}$. Moreover, $\calE_1$ and $\calE_{2, all}$ simultaneously hold with probability at least $1 - \Theta(\frac{td}{D}) - \frac{2^{O(m \log (nd)}}{e^t}$, which is $1 - \frac{1}{\poly(nd)}$ for $D = \poly(nd)$ and $t = \Theta(m \log (nd))$. This completes the proof of the lemma.
\end{proof}

\begin{customlemma}{2}[Sketched Error Lower Bound]
\label{lemma:pstable}
Let $A \in \R^{d \times n}$ and $k \in \N$. Let $t = k \cdot \poly(\log (nd))$, and let $S \in \R^{t \times d}$ be a matrix whose entries are i.i.d. standard $p$-stable random variables, rescaled by $\Theta(1/t^{\frac{1}{p}})$. Then, with probability $1 - o(1)$, for \textbf{all} $T \subset [n]$ with $|T| = k \cdot \poly(\log k)$ and for \textit{all} $V \in \R^{|T| \times n}$,
\begin{align*}
    \|A_TV - A\|_p \leq \|SA_TV - SA\|_p    
\end{align*}
\end{customlemma}

\begin{proof}
Let $y_j$ denote the $j$-th column of $V$, where $j \in [n]$. By applying \textbf{Lemma}~\ref{lemma:cssboundvector}, and a union bound over all columns of $V$, the following holds with probability $1 - \frac{n}{\poly(nd)} = 1 - o(1)$:
\begin{align*}
\|A_TV - A\|_p &= (\sum_{j=1}^n \|A_Ty_j - A_j\|_p^p)^{\frac{1}{p}}\\
&\leq (\sum_{j=1}^n \|S(A_Ty_j - A_j)\|_p^p)^{\frac{1}{p}}\\
&= \|SA_TV - SA\|_p
\end{align*}
\end{proof}

\subsection{Sketched Error Upper Bound (Lemma 3)}
We show an upper bound on the approximation error of $k$-CSS$_p$ on a sketched subset of columns, $\|SA_TV - SA_T\|_p$, which holds for a fixed subset $A_T$ of columns and for the minimizing right factor $V = \arg\min_{V}\|SA_TV - SA\|_p$ for that subset of columns.

We first adapt Lemma E.17 from~\cite{swz2017lral1norm} to establish an upper bound on the error $\|SA_TV - SA\|_p$ for any fixed $V$ in \textbf{Lemma}~\ref{lemma:afactor}. We then apply \textbf{Lemma}~\ref{lemma:afactor} to the minimizer $V$ to conclude the upper bound in \textbf{Lemma}~\ref{lemma:bestcss}.

\begin{customlemma}{3.1}(An Upper Bound on Norm of A Sketched Matrix) \label{lemma:afactor}
Given $A \in \R^{n \times d}$ and $p \in [1, 2)$, and $U \in \R^{n \times k}$ and $V \in \R^{k \times d}$, if $S \in \R^{t \times n}$ is a dense $p$-stable matrix, whose entries are rescaled by $\Theta\left(\frac{1}{t^{\frac{1}{p}}}\right)$, then with probability at least $1 - o(1)$,
\begin{align*}
    \|SUV - SA\|_p^p \leq O(\log (td)) \|UV - A\|_p^p
\end{align*}
Here, the failure probability $o(1)$ can be arbitrarily small.
\end{customlemma}

\begin{proof}
Lemma E.17 from~\cite{swz2017lral1norm}.
\end{proof}

\begin{customlemma}{3}[Sketched Error Upper Bound (Lemma E.11 of \cite{swz2017lral1norm})]
\label{lemma:bestcss}
Let $A \in \R^{d \times n}$ and $k \in \N$. Let $t = k \cdot \poly(\log (nd))$, and let $S \in \mathbb{R}^{t \times d}$ be a matrix whose entries are i.i.d. standard $p$-stable random variables, rescaled by $\Theta(1/t^{\frac{1}{p}})$. Then, for a fixed subset $T \subset [n]$ of columns with $|T| = k\cdot \poly(\log k)$ and a fixed $V \in \R^{|T| \times n}$, with probability $1 - o(1)$, we have
\begin{align*}
    \min_V \|SA_TV - SA\|_p \leq \min_V O(\log^{1/p}(nd))\|A_TV - A\|_p   
\end{align*}
\end{customlemma}

\begin{proof}
Let $X_1^* = \arg\min_X \|SA_TX - SA\|_p$ and $X_2^* = \arg\min_X \|A_TX - A\|_p$. By \textbf{Lemma}~\ref{lemma:afactor},
\begin{align*}
    \|SA_TX_1^* - SA\|_p^p &\leq \|SA_TX_2^* - SA\|^p_p\\
    &\leq O(\log(k\textrm{poly}(\log n)d))\|A_TX_2^* - A\|^p_p\\
    &\le O(\log(nd))\|A_TX_2^* - A\|^p_p
\end{align*}
Therefore, $$\min_X \|SA_TX - SA\|_p \leq \min_X O(\log^{1/p}(nd))\|A_T X - A\|_p$$.
\end{proof}

\newpage
\section{$\ell_p$ Lewis Weights and Applications} \label{section:lewis_weights_applications}

\subsection{$\ell_p$ Lewis Weights Background}
\label{subsec:lewis_weights}
Our streaming and distributed $k$-CSS algorithms make use of $\ell_{p, 2}$ strong coresets and an $O(1)$-approximation $k$-CSS$_{p, 2}$ subroutine, both of which applies importance sampling of the input matrix, based on the so-called \textit{Lewis weights} (see Definition~\ref{def:lewis_weights}), which can be approximated with repeated computation of the \textit{leverage scores} (see Definition~\ref{def:leverage_scores}) in polynomial time~\cite{cp2015lewisweights}. In this section, we briefly introduce the Lewis weights and the desired properties associated with it. We further introduce \textit{$\ell_p$ sensitivities} and \textit{$\ell_p$ well-conditioned basis} to aid the analysis of the desired property we need from Lewis weights.

\begin{definition}[Statistical Leverage Scores --- Definition 16 of \cite{sketching_as_a_tool}]
\label{def:leverage_scores}
Let $A \in \R^{n \times d}$, and suppose $A = U \Sigma V^T$ is the ``thin'' SVD of $A$. \footnotemark \footnotetext{meaning that if $A$ is of rank $k$, then $U$ and $V$ have $k$ columns, and $\Sigma \in \R^{k \times k}$.} Then, for $i \in [n]$, define $\ell_i(A) = \|U_{i, *}\|_2^2$ --- we say $\ell_i(A)$ is the $i^{th}$ statistical leverage score of $A$.
\end{definition}

\begin{definition}[$\ell_p$ Lewis Weights --- Definition 2.2 of \cite{cp2015lewisweights}]
\label{def:lewis_weights}
Let $1 \leq p < \infty$, and let $A \in \R^{n \times d}$. Then, the $\ell_p$ Lewis weights of $A$ are given by a unique vector $\overline{w} \in \R^n$ such that $\overline{w}_i = \ell_i(\text{diag}(\overline{w})^{1/2 - 1/p} A)$, where $\text{diag}(\overline{w})$ is the $n \times n$ diagonal matrix with the entries of $\overline{w}$ on its diagonal. By Corollaries 3.4 and 4.2 of \cite{cp2015lewisweights}, such a vector $\overline{w}$ exists and is unique. 
\end{definition}

\begin{definition}[$\ell_p$ Sensitivities~\cite{cww2019tukey_regression}]
\label{def:lp_sensitivities}
    Let $1 \leq p < \infty$, and let $A \in \mathbb{R}^{n \times d}$. Let $col(A)$ denote the column span of $A$. The $i^{th}$ $\ell_p$ sensitivity of $A$ is defined as $\sup_{y \in col(A)}\frac{y_i^{p}}{\|y\|_p^p}$.
\end{definition}

\begin{definition}[$\ell_p$ Well-conditioned Basis --- Definition 3 of~\cite{ddhkm2008lp_sampling}]
\label{def:well_conditioned_basis}
    Let $1 \leq p < \infty$ and $A \in \mathbb{R}^{n \times d}$ of rank $k$. Let $q$ be its dual norm. Then an $n \times k$ matrix $U$ is an $(\alpha, \beta, p)$ well-conditioned basis for the column span of $A$, if (1) $\|U\|_{p} \leq \alpha$, and (2) for all $z \in \mathbb{R}^{k}$, $\|z\|_{q} \leq \beta \|U z\|_p$. W will say $U$ is a $p$ well-conditioned basis for the column span of $A$, if $\alpha$ and $\beta$ are $k^{O(1)}$, independent of $d$ and $n$.
\end{definition}

\begin{definition}[$\ell_p$ Leverage score sampling --- Theorem 5 of~\cite{ddhkm09_well_conditioned_bases}]
\label{def:lp_leverage_score_sampling}
    Let $1 \leq p < \infty$ and $A \in \mathbb{R}^{n \times d}$ of rank $k$. Let $U$ be an $(\alpha, \beta, p)$ well-conditioned basis for the column span of $A$. Given approximation error $\eps$ and failure probability $\delta$, for $r \geq C(\eps, \delta, p, k)$, the $\ell_p$ leverage score sampling is any sampling probability $p_i \geq \min\{1, \frac{\|U_{i*}\|_p^{p}}{\|U\|_p^{p}}r\}$, $\forall i \in [n]$.
\end{definition}

The desired property from Lewis weights we need is called an $\ell_p$ subspace embedding (see \textbf{Theorem}~\ref{thm:high_prob_lewis_weights}).
\cite{cp2015lewisweights} shows for a matrix $A \in \R^{n \times d}$, if the rows of $A$ are appropriately sampled using a certain distribution based on the $\ell_p$ Lewis weights of $A$, this property holds with constant probability. 
However, for our construction of strong coresets, we need this property of Lewis weights to hold with high probability $1 - \delta$ for some small $\delta \in (0, \frac{1}{2})$. We explain why this is possible following the works from~\cite{cp2015lewisweights, blm1989zenotope}.

\begin{customthm}{4.2}($\ell_p$-Lewis Weights Subspace Embedding)
\label{thm:high_prob_lewis_weights}
Given an input matrix $A \in \R^{n \times d}$ and $p \in [1, 2)$, there exists a distribution $(\lambda_1, \lambda_2, \ldots, \lambda_n)$ on the rows of $A$, where the distribution is constructed based on Lewis weights sampling. If the following two conditions are met: (1) $n \leq \poly(d/\eps)$, and (2) the minimum (row) Lewis weights of $A$ is at least $1/\poly(d/\eps)$, then for a sampling and rescaling matrix $S$ with $t$ rows, each chosen independently as the $i^{th}$ standard basis vector times $\frac{1}{(t\lambda_i)^{\frac{1}{p}}}$ with probability $\lambda_i$, with $t = O(d \cdot \poly(\log (d/\delta), 1/\eps))$, the following holds for all $x \in \mathbb{R}^d$ simultaneously with probability $1-\delta$:
$$\|SAx\|_p = (1\pm \epsilon)\|Ax\|_p$$
\end{customthm}
\begin{proof}

The Theorem follows Theorem 7.1 of~\cite{cp2015lewisweights}, except that~\cite{cp2015lewisweights} states the above property of Lewis weights holds with constant probability. However, this result can be improved for it to hold with probability $1-\delta$ as follows: Using the results of~\cite{blm1989zenotope}, it is possible to construct a sampling and rescaling matrix $S$ (i.e. a matrix with one non-zero value per row) with $d \poly(\log (d/\delta), 1/\eps)$ rows such that with probability at least $1-\delta$, we have $\|SAx\|_p = (1\pm \eps) \|Ax\|_p$ simultaneously $\forall x \in \mathbb{R}^{d}$. 

To do this, the authors of~\cite{blm1989zenotope} construct a sequence of $v = \poly((\log d)/\eps)$ sets of vectors $\{V_i\}_{i=1}^{v}$, and each vector in a net over the column space can be written approximately as a sum of vectors, one drawn from each set $V_i$. Then they show via Bernstein's inequality that the vectors in all sets have their norms preserved if one samples from the $\ell_p$-Lewis weights of $A$, and the final bound follows from the triangle inequality.
By increasing the number of rows in $S$ by an $O(\log (d/(\eps \delta))$ factor, one can now argue that with probability $1-\delta$, all vectors $y$
in the column span of $A$ have their norm $\|y\|_p$ preserved. This gives a total of $O(d \poly(\log (d/\delta), 1/\eps))$ rows in $S$. 
However, to apply the results from~\cite{blm1989zenotope}, we need two conditions to be satisfied: (1) $n \leq \poly(d/\eps)$, and (2) the minimum (row) Lewis weight of $A$ is at least $1/\poly(d/\eps)$.

We can achieve both conditions by first replacing $A$ with $TA$, where $T$ is a sampling matrix for which $T$ has $\poly(d/\eps)$ rows and with probability $1-\delta$, $\|TAx\|_p = (1\pm \eps) \|Ax\|_p$ simultaneously $\forall x \in \mathbb{R}^{d}$. Many constructions of such $T$ exist with $\poly(d/\eps \log(1/\delta))$ rows, e.g. e.g., based on the $\ell_p$-sensitivities of $A$ (see Definition~\ref{def:lp_sensitivities}) or $\ell_p$ leverage score sampling (see Definition~\ref{def:lp_leverage_score_sampling}), or the $\ell_p$-Lewis weights themselves (see Definition~\ref{def:lewis_weights}). See for example Theorem 10 of \cite{cww2019tukey_regression}. If we choose $\delta > 1/\exp(\poly(d/\eps))$, then the number of rows of $TA$ will be at most $\poly(d/\eps)$, satisfying condition (1) of Theorem \ref{thm:high_prob_lewis_weights}. 
On the other hand, if $\delta \leq 1/\exp(\poly(d/\eps))$, then $\log(1/\delta) > \poly(d/\eps)$.
In this case, we can just sample using $\ell_p$ sensitivities as in Theorem 3.10 of \cite{cww2019tukey_regression}. The number of rows needed will be $\poly(d/\eps)$, which can simply be absorbed into the $\poly(\log(1/\delta))$.

While (1) holds since the number of rows of $T$ is at most $\poly(d/\eps)$, we can ensure (2) also holds by computing the (row) $\ell_p$-Lewis weights of $TA$, and discarding any row $i$ with $\ell_p$-Lewis weight less than $1/\poly(n/\eps)$. 
By Lemma 5.5 of~\cite{cp2015lewisweights} this cannot make the $\ell_p$-Lewis weight of any non-discarded row decrease.
Moreover, since the $\ell_p$-Lewis weights are upper bounds on the $\ell_p$-sensitivities for $1 < p < 2$ (by Lemma 3.8 of \cite{cww2019tukey_regression} and Definition~\ref{def:lp_sensitivities}), discarding such rows $i$ only changes $\|TAx\|_p$ by a $(1\pm\eps)$ factor for any $x \in \mathbb{R}^{d}$, by the triangle inequality. 
Finally, since $n \leq \poly(d/\eps)$, we also have that any $\ell_p$-Lewis weight is now at least $1/\poly(d/\eps)$, as needed to now apply the result of~\cite{blm1989zenotope}. Thus, we can now apply the above $S$ to non-discarded rows of $TA$.
\end{proof}

\subsection{Strong Coresets for $\ell_{p, 2}$ Norm Low Rank Approximation (Lemma 4)}

\begin{customlemma}{4}[Strong Coresets in $\ell_{p,2}$ norm \cite{sw2018strongcoreset}]\label{lemma:coreset}
Let $A \in \R^{d \times n}$, $k \in \N$, $p \in [1, 2)$, and $\eps, \delta \in (0, 1)$. Then, in $\widetilde{O}(nd)$ time, one can find a sampling and reweighting matrix $T$ with
$O(\frac{d}{\eps^2} \poly(\log(d/\eps), \log(1/\delta)))$ columns, such that, with probability $1 - \delta$, for all rank-$k$ matrices $U$, 
\begin{align*}
    \min_{\textrm{rank-k }V} \|UV - AT\|_{p,2} = (1 \pm \eps)\min_{\textrm{rank-k }V}\|UV - A\|_{p,2}
\end{align*}
where $AT$ is called a \textbf{strong coreset} of $A$.
\end{customlemma}

\begin{proof}
We can obtain $T$ with the desired number of columns using the strong coreset construction from \textbf{Lemma 16} in \cite{sw2018strongcoreset}. For our purposes, the matrix $B \in \R^{n \times (d + 1)}$ that we use will be different than the $B$ used in the statement and proof of \textbf{Lemma 16} in \cite{sw2018strongcoreset}. The coreset construction in \cite{sw2018strongcoreset} has the goal of removing a dependence on $d$ in the coreset size. In \cite{sw2018strongcoreset}, $B$ refers to a matrix obtained by projecting $A$ onto a $\poly(k)$-dimensional subspace $S$ obtained by running a dimensionality reduction algorithm (referred to as \textsc{DimensionalityReduction} in \cite{sw2018strongcoreset}) and constructing a coreset by sampling rows from $B$. The rows are sampled according to the $\ell_p$ Lewis weights of $B$.

In our case, we do not want our coreset size to have a polynomial dependence on $k$, while a linear dependence on $d$ suffices. Thus, instead of using the dimensionality reduction subroutine in \cite{sw2018strongcoreset}, we simply let $B$ be the input matrix $A$, concatenated with a column of $0$'s (the column span of $A$ will be the subspace $S$ referred to in the statement of \textbf{Lemma 16} of \cite{sw2018strongcoreset}). The desired number of rows and running time then follows from \textbf{Lemma 16} of \cite{sw2018strongcoreset}. \footnotemark \footnotetext{The proof of \textbf{Lemma 16} of \cite{sw2018strongcoreset} mentions that for $p > 1$, Lewis weight sampling requires $(f/\eps)^{O(p)}$ rows for a matrix with $C$ columns --- this is a typo, and $f \poly(\frac{\log f}{\eps})$ rows suffice.} 
Based on $1-\delta$ $\ell_p$ Lewis weights subspace embedding, the size of the coreset grows linearly in $\poly(\log(1/\delta))$.
\end{proof}

\subsection{Bi-criteria $O(1)$-approximation algorithm for $k$-CSS$_{p, 2}$ (Theorem 1)}
\label{appendix:kcss_in_lp2_norm}
We introduce an  $O(1)$-approximate bi-criteria $k$-CSS$_{p, 2}$ algorithm (\textbf{Algorithm}~\ref{algorithm:kCSSp2}), which is a modification of the algorithm from~\cite{cw2015subsapproxl2}. The major difference is that we use $\ell_p$-Lewis weight sampling, instead of $\ell_p$ leverage score sampling, which reduces the number of output columns from $O(k^2)$ to $O(k\poly(\log k))$.

We first show how to use a sparse embedding matrix $S$ to obtain an $O(1)$-approximate left factor in Section~\ref{subsec:osnap_sparse_embedding}. We then show how to apply the $\ell_p$-Lewis weight sampling to select a subset of $\widetilde{O}(k)$ columns that gives an $O(1)$-approximation in Section~\ref{subsec:lewis_weights_sampling}. Finally, we conclude the analysis of our $O(1)$-approximate bi-criteria $k$-CSS$_{p, 2}$ algorithm in Section~\ref{subsec:full_analysis_kcssp2}.

\begin{algorithm}
\caption{polynomial time, $O(1)$-approximation for $k$-CSS$_{p, 2}$ ($1 \leq p < 2$)}
\label{algorithm:kCSSp2}
\begin{algorithmic}
\STATE \textbf{Input:} The data matrix $A \in \R^{d \times n}$, rank $k \in \mathbb{N}$
\STATE \textbf{Output:} The left factor $U \in \R^{d \times \widetilde{O}(k)}$, the right factor $V \in \R^{\widetilde{O}(k) \times n}$ such that $\|UV - A\|_{p, 2} \leq O(1)\min_{\textrm{rank-k} A_k}\|A_k - A\|_{p, 2}$
\STATE $S \gets $ $\widetilde{O}(k)\times d$ sparse embedding matrix, with sparsity $s = \poly(\log k)$.
\STATE {$S' \gets $ $n \times \widetilde{O}(k)$ sampling matrix, each column of which is a standard basis vector chosen randomly according to the $\ell_p$ Lewis weights of columns of $SA$.}
\STATE Return $U \gets AS'$, $V \gets (AS')^{\dag}A$ 
\COMMENT{$\dag$ denotes the Moore-Penrose pseudoinverse.}
\end{algorithmic}
\end{algorithm}

\subsubsection{Sparse Embedding Matrices}
\label{subsec:osnap_sparse_embedding}
The \textbf{sparse embedding matrix} $S \in \mathbb{R}^{\widetilde{O}(k) \times d}$ of  \cite{nn2013OSNAP}, and used by \cite{cw2015subsapproxl2}, is constructed as follows: each column of $S$ has exactly $s$ non-zero entries chosen in uniformly random locations. Each non-zero entry is a random value $\pm \frac{1}{\sqrt{s}}$ with equal probability. $s$ is also called the \textit{sparsity} of $S$. Let $h$ be the hash function that picks the location of the non-zero entries in each column of $S$ and $\sigma$ be the hash function that determines the sign $\pm$ of each non-zero entry.

Applying the sparse embedding matrix $S$ to $A$ enables us to obtain a rank-$k$ right factor that is at most a factor of $O(1)$ worse than the best rank-$k$ approximation error in the $\ell_{p, 2}$ norm. We adapt Theorem 32 from~\cite{cw2015subsapproxl2} to show this in \textbf{Theorem}~\ref{thm:dim_reduction_OSNAP}. Notice that in Theorem 32 of \cite{cw2015subsapproxl2}, the number of rows required for $S$ is $O(k^2)$, but this can be reduced to $\widetilde{O}(k)$ through a different choice of hyperparameters when constructing the sparse embedding matrix $S$.

We note two choices of hyperparameters, i.e., the number $m$ of rows and sparsity $s$, of $S$ in \textbf{Theorem}~\ref{thm:osnap_construct_1} and \textbf{Theorem}~\ref{thm:osnap_construct_2}, both of which give the same result. The proof of Theorem 32 from~\cite{cw2015subsapproxl2} uses the hyperparameters from \textbf{Theorem}~\ref{thm:osnap_construct_1}. We instead use the  hyperparameters from \textbf{Theorem}~\ref{thm:osnap_construct_2} and show in \textbf{Lemma}~\ref{lemma:lemma29} that $\widetilde{O}(k)$ rows of $S$ suffice to preserve certain desired properties. We then combine \textbf{Lemma}~\ref{lemma:lemma29} and \textbf{Lemma}~\ref{lemma:lemma31} adapted from~\cite{cw2015subsapproxl2}, to conclude our result in \textbf{Theorem}~\ref{thm:dim_reduction_OSNAP}, following the analysis from~\cite{cw2015subsapproxl2}.

\begin{customthm}{5.1}(Theorem 3 from~\cite{nn2013OSNAP})
\label{thm:osnap_construct_1}
    For a sparse embedding matrix $S \in \mathbb{R}^{m \times n}$ with sparsity $s = 1$ and a data matrix $U \in \mathbb{R}^{n \times d}$, let $\epsilon \in (0, 1)$. With probability at least $1 - \delta$ all singular values of $SU$ are $(1 \pm \epsilon)$ as long as $m \geq \delta^{-1}(d^2 + d)/(2\epsilon - \epsilon^2)^2$. For the hash functions used to construct $S$, $\sigma$ is 4-wise independent and $h$ is pairwise independent.
\end{customthm}

\begin{customthm}{5.2}(Theorem 9 from~\cite{nn2013OSNAP})
\label{thm:osnap_construct_2}
    For a sparse embedding matrix $S \in \mathbb{R}^{m \times n}$ with sparsity $s = \Theta(\log^3(d/\delta)/\epsilon)$ and a data matrix $U \in \mathbb{R}^{n \times d}$, let $\epsilon \in (0, 1)$. With probability at least $1 - \delta$ all singular values of $SU$ are $(1 \pm \epsilon)$ as long as $m = \Omega(d\log^8(d/\delta)/\epsilon^2)$. For the hash functions used to construct $S$, we have that $\sigma, h$ are both $\Omega(\log(d/\delta))$-wise independent.
\end{customthm}

\begin{customlemma}{5.3}
\label{lemma:lemma29}
   Let $\mathcal{C}$ be a constraint set and $A \in\mathbb{R}^{n \times d}, B\in\mathbb{R}^{n \times d'}$ be two arbitrary matrices. For a sparse embedding matrix $S \in \mathbb{R}^{m \times n}$, there is $m = O(\frac{d\log^8(\frac{d}{\epsilon^{p+1}})}{\epsilon^{2(p+1)}})$, such that with constant probability, the following hold:
    \begin{align*}
        &i)\ \|S(AX - B)\|_{p, 2} \geq (1-\epsilon)\|AX-B\|_{p, 2} \text{ for all } X \in \R^{d \times d'}\\
        &ii)\ \|S(AX^* - B)\|_{p, 2} \leq (1 + \epsilon)\|AX^* - B\|_{p, 2},\ where\ X^* = \arg\min_{X \in \mathcal{C}}\|AX - B\|_{p, 2}
    \end{align*}
\end{customlemma}

\begin{proof}
The proof is the same as the proof of Lemma 29 from~\cite{cw2015subsapproxl2}, except that we use a different choice of hyperparameters in constructing $S$, i.e., sparsity $s$ and the number $m$ of rows. In the proof of Lemma 29 from~\cite{cw2015subsapproxl2}, the construction of $S$ follows \textbf{Theorem}~\ref{thm:osnap_construct_1}, where the sparsity $s = 1$, but requires $m = O(d^2)$ rows. We replace the construction by \textbf{Theorem}~\ref{thm:osnap_construct_2}, where we pick $\delta = \epsilon^{p+1}$. Now the sparsity $s$ is larger but this construction reduces the number of rows required to $m = \widetilde{O}(d)$.

If we use the construction in Theorem \ref{thm:osnap_construct_2} with the parameters $\eps, \delta$ both being $\eps^{p + 1}$, then the rest of the proof follows from Lemma 27 of \cite{cw2015subsapproxl2} (using the same argument as in Lemma 29 of \cite{cw2015subsapproxl2}). As in Lemma 29 of \cite{cw2015subsapproxl2}, properties (i) and (ii) of Lemma 27 of \cite{cw2015subsapproxl2} follow simply because we have chosen the parameters $\eps, \delta$ of $S$ to be $\eps^{p + 1}$ (thus, by Theorem \ref{thm:osnap_construct_2}, $S$ is an $\eps$-subspace embedding for $A$ in the $\ell_2$ norm, and $S$ is an $\eps^{p + 1}$-subspace embedding for $[A, B_{*, i}]$ with probability $1 - \eps^{p + 1}$, for all $i$). Finally, to show that property (iii) in Lemma 27 of \cite{cw2015subsapproxl2}, the only property of the matrix $S$ that is needed by \cite{cw2015subsapproxl2} is Equation (20) of \cite{nn2013OSNAP}, which also holds when $S$ is constructed as in Theorem 9 of \cite{nn2013OSNAP}.
\end{proof}

\begin{customlemma}{5.4} 
\label{lemma:lemma31}
Consider a data matrix $A \in \mathbb{R}^{n \times d}$. Let the best rank-$k$ matrix in the $\ell_{p, 2}$ norm be $A_k = \arg\min_{\textrm{rank-k }A_k}\|A_k - A\|_{p, 2}$.
For $R \in \mathbb{R}^{d \times m}$, if $R^T$ satisfies both of the following two conditions for all $X \in \mathbb{R}^{n \times n}$:
\begin{align*}
    &i)\ \|R^T(A_k^TX - A^T)\|_{p, 2} \geq (1 - \epsilon)\|A_k^TX - A^T\|_{p, 2}\\
    &ii)\ \|R^T(A_k^TX^* - A^T)\|_{p, 2} \leq (1 + \epsilon)\|A_k^TX^* - A^T\|_{p, 2},\ where\ X^* = \arg\min_{X}\|A_k^TX - A^T\|_{p, 2}
\end{align*}
then
\begin{align*}
    \min_{\textrm{rank-k }X} \|XR^TA_k^T - A^T\|_{p, 2}^p \leq (1+3\epsilon)\|A_k^T - A^T\|_{p, 2}^p
\end{align*}
\end{customlemma}

\begin{proof}
    Lemma 31 from~\cite{cw2015subsapproxl2}.
\end{proof}

\begin{customthm}{5.5}($\ell_{p, 2}$-Low Rank Approximation)
\label{thm:dim_reduction_OSNAP}
Let the data matrix be $A \in \R^{d \times n}$ and $k \in \mathbb{N}$ be the desired rank. Let $S \in \R^{m \times d}$ be a sparse embedding matrix with $m = O(k \poly(\log k) \poly(\frac{1}{\eps}))$ rows, and sparsity $s = \poly(\log k)$. Then, the following holds with constant probability:
$$\min_{\textrm{rank-k } X} \|XSA - A\|_{p, 2} \leq (1 + 3\eps) 
\min_{\textrm{rank-k } A_k} \|A_k - A\|_{p, 2}$$
\end{customthm}

\begin{proof}
The proof is the same as the proof of Theorem 32 in~\cite{cw2015subsapproxl2}, except that we adapt a different construction of the sparse embedding matrix $S$, which reduces the number of rows from $O(k^2)$ to $\widetilde{O}(k)$ with increased sparsity $s$.

Consider $A_k = \arg\min_{\textrm{rank-k }A_k}\|A_k - A\|_{p, 2}$. 
Let $V_k$ be a basis for the column space of $A_k$.
By applying \textbf{Lemma}~\ref{lemma:lemma29} and \textbf{Lemma}~\ref{lemma:lemma31} on the basis $V_k$, we conclude the above theorem by setting the number $m$ of rows to $m = O\Big(\frac{k\log^8(\frac{k}{\eps^2})}{\eps^4}\Big)$, and sparsity $s = \poly(\log k)$ in the sparse embedding matrix $S$.
\end{proof}

\subsubsection{Using Lewis Weight Sampling for Column Subset Selection} \label{subsec:lewis_weights_sampling}

Here we show how to use $\ell_p$ Lewis weights sampling (discussed above in at the beginning of Section \ref{section:lewis_weights_applications}) for $k$-CSS$_{p, 2}$. We first introduce a technical tool, a version of Dvoretzky's Theorem (\textbf{Theorem}~\ref{thm:dvoretzky_for_lp}) which allows us to embed $\ell_2^n$ into $\ell_2^{O(n/\eps^2)}$ with only $(1 \pm \eps)$ distortion, and thus enables us to switch between the $\ell_p$ norm and the $\ell_2$ norm and use $\ell_p$ Lewis weight sampling. Based on \textbf{Theorem}~\ref{thm:high_prob_lewis_weights} and \textbf{Theorem}~\ref{thm:dvoretzky_for_lp}, we show in \textbf{Theorem}~\ref{thm:subset_cols_lewis_weight} that Lewis weight sampling provides a good subset of columns, on which our later analysis of $k$-CSS$_{p, 2}$ is based.

\begin{customthm}{5.7}(Randomized Dvoretzky's Theorem)
\label{thm:dvoretzky_for_lp}
Let $n \in \N$, and $\eps \in (0, 1)$. Let $r = \frac{n}{\eps^2}$. Let $G \in \R^{r \times n}$ be a random matrix whose entries are i.i.d. standard Gaussian random variables, rescaled by $\frac{1}{\sqrt{r}}$. For $r = \frac{n}{\epsilon^2}$, the following holds with probability $1 - e^{-\Theta(n)}$, for all $y \in \R^n$, $$\|Gy\|_p = (1 \pm \eps)\|y\|_2$$
\end{customthm}

\begin{proof}
This follows from Theorem 1.2 from~\cite{pvz_dvoretzky_lp_norm}.
\end{proof}

\begin{customthm}{5.8}(Subset of Columns by Lewis Weights Sampling)
\label{thm:subset_cols_lewis_weight}
Let $A \in \mathbb{R}^{d \times n}$. Let $S \in \mathbb{R}^{m \times d}$ be a sparse embedding matrix, with $m = O(k\cdot \poly(\log k)\poly(\frac{1}{\epsilon}))$. Further, let $S' \in \mathbb{R}^{n \times t}$ be a sampling matrix whose columns are random standard basis vectors generated according to the $\ell_p$ Lewis weights of columns of $SA$ (that is, the row sampling matrix $(S')^T$ is generated based on the Lewis weights of $(SA)^T$), with $t = k \cdot \poly(\log k)$. Then, for $\hat{X} = \arg\min_{\textrm{rank-k } X} \|XSAS' - AS'\|_{p, 2}$, the following holds with probability $1 - o(1)$: 
$$\|\hat{X}SA - A\|_{p, 2} \leq \Theta(1) \min_{\textrm{rank-k }A_k}\|A_k - A\|_{p, 2}$$
\end{customthm}

\begin{proof}
Let $X^* = \arg\min_{\textrm{rank-k }X^*}\|X^*SA - A\|_{p, 2}$. By the triangle inequality,
$$\|\hat{X}SA - A\|_{p, 2} \leq \|X^*SA - \hat{X}SA\|_{p, 2} + \|X^*SA - A\|_{p, 2}$$

Our goal is to bound $\|X^*SA - \hat{X}SA\|_{p, 2}$. By Lemma D.28 and Lemma D.29 from~\cite{swz2017lral1norm}, for any column sampling matrix $S$ and for any fixed matrix $Y$, it can be shown that $\mathbf{E}[\|YS\|_p^p] = \|YS\|_p^p$. In our case, since $S'$ is a sampling matrix, we have $\mathbf{E}[\|YS'\|_p^p] = \|YS'\|_p^p$ for any fixed matrix $Y$.

Now let $G \in \mathbb{R}^{\Theta(d) \times d}$ be a rescaled random matrix whose entries are i.i.d. standard Gaussian random variables as in \textbf{Theorem}~\ref{thm:dvoretzky_for_lp}. We apply \textbf{Theorem}~\ref{thm:dvoretzky_for_lp} to transform between the $\ell_p$ space and the Euclidean space. Since transformation of both directions can be done with very small distortion, we obtain a $\Theta(1)$ approximation.
With constant probability, we have
\begin{align*}
    &\|X^*SA - \hat{X}SA\|_{p, 2}\\
    & = \Theta(1)\|G(X^* - \hat{X})SA\|_p &\text{By \textbf{Theorem}~\ref{thm:dvoretzky_for_lp}}\\
    & = \Theta(1)\|G(X^* - \hat{X})SAS'\|_p  &\text{By \textbf{Theorem}~\ref{thm:high_prob_lewis_weights}}\\
    & = \Theta(1)\|(X^* - \hat{X})SAS'\|_{p, 2} &\text{By \textbf{Theorem}~\ref{thm:dvoretzky_for_lp}}\\
    & \leq \Theta(1) \Big(\|X^*SAS' - AS'\|_{p, 2} + \|\hat{X}SAS' - AS'\|_{p, 2}\Big) &\text{Triangle Inequality}\\
    & \leq \Theta(1) \|X^*SAS' - AS'\|_{p, 2} &\text{Since $\hat{X} = \arg\min_{\textrm{rank-k } X} \|XSAS' - AS'\|_{p, 2}$}\\
    &=\Theta(1)\|G(X^*SA - A)S'\|_p &\text{By \textbf{Theorem}~\ref{thm:dvoretzky_for_lp}}\\
    & \leq \Theta(1) \|G(X^*SA - A)\|_p &\text{By Markov Bound on $\mathbf{E}[\|YS'\|_p^p] = \|YS'\|$}\\
    & = \Theta(1)\|X^*SA - A\|_{p, 2} &\text{By \textbf{Theorem}~\ref{thm:dvoretzky_for_lp}}
\end{align*}

Therefore,
\begin{align*}
    \|\hat{X}SA - A\|_{p, 2} &\leq \|X^*SA - \hat{X}SA\|_{p, 2} + \|X^*SA - A\|_{p, 2}\\
    &\leq \Theta(1)\|X^*SA - A\|_{p, 2}\\
    &\leq \Theta(1)\min_{\textrm{rank-k }A_k}\|A - A_k\|_{p, 2} &\text{By \textbf{Theorem}~\ref{thm:dim_reduction_OSNAP}}
\end{align*}
as desired. Note that we can achieve $o(1)$ failure probability by increasing the number of columns in $S'$ by a logarithmic factor --- see \textbf{Theorem}~\ref{thm:high_prob_lewis_weights}.
\end{proof}

\subsubsection{Analysis for $k$-CSS$_{p, 2}$}
\label{subsec:full_analysis_kcssp2}

We now conclude our proof that Algorithm~\ref{algorithm:kCSSp2} for bi-criteria $k$-CSS$_{p, 2}$ achieves an $O(1)$ approximation factor with polynomial running time. This result is stated as \textbf{Theorem}~\ref{thm:bicriteria_k_css_p2}.

\begin{customthm}{1}[Bicriteria $O(1)$-Approximation Algorithm for $k$-CSS$_{p, 2}$]
\label{thm:bicriteria_k_css_p2}
Let $A \in \R^{d \times n}$ and $k \in \N$. There is an algorithm with $(\nnz(A) + d^2)\cdot k\poly(\log k)$ runtime that outputs a rescaled subset of columns $U \in \R^{d \times \widetilde{O}(k)}$ of $A$ and a right factor $V \in \R^{\widetilde{O}(k) \times n}$ for which $V = \min_{V}\|UV - A\|_{p, 2}$, such that with probability $1 - o(1)$, 
\begin{align*}
    \|UV - A\|_{p, 2} \leq O(1) \cdot \min_{\text{ rank-k}\ A_k} \|A_k - A\|_{p, 2}
\end{align*}
\end{customthm}

\begin{proof}
\textbf{Approximation Factor. }
First notice that the minimizer $\hat{X}$ of $\|\hat{X}SAS' - AS'\|_{p, 2}$ has to be in the column span of $AS'$. Thus we can write $\hat{X} = (AS')Y$ for some matrix $Y$. By \textbf{Theorem}~\ref{thm:subset_cols_lewis_weight}, 
\begin{align*}
    \|\hat{X}SA - A\|_{p, 2} = \|(AS')YSA - A\|_{p, 2} \leq \Theta(1) \min_{\textrm{rank-k }A_k}\|A - A_k\|_{p, 2}
\end{align*}

We denote $YSA = V$. We take the left factor $U = AS'$ and solving for $\min_{V} \|UV - A\|_{p, 2}$ will give us a $\Theta(1)$ approximation to $ \min_{\textrm{rank-k }A_k}\|A - A_k\|_{p, 2}$. A good minimizer for the right factor $V$ in the Euclidean space is $V = (AS')^{\dag}A$. This concludes our result. Notice that since $S'$ is a sampling matrix with $\widetilde{O}(k)$ columns, we get a rank-$k$ left factor $U$ as a subset of columns of $A$ as desired.

\textbf{Running time. } 
First notice that $S$ is a sparse embedding matrix with $\poly(\log k)$ non-zero entries. Thus computing $SA$ takes time $\nnz(A) \cdot \poly(\log k)$. By~\cite{cp2015lewisweights},
computing the Lewis weights of $SA$ takes time $\nnz(SA) + \poly(k) \leq \nnz(A) \poly(\log k) + \poly(k)$, and computing the output left factor $U = AS'$ takes time $\nnz(A)$. 
Computing $(AS')^{\dag}$ takes time $d^2\cdot k\poly(\log
k)$. Computing the right factor $V = (AS')^{\dag}A$ takes $\nnz(A)k\poly(\log
k)$.
Therefore, the overall running time is $(\nnz(A) + d^2)\cdot k\poly(\log k)$.

\textbf{Failure Probability. } For the failure probability of the first step, note that we can select the parameter $\delta$ for the sparse embedding matrix $S$ to be $\frac{1}{\poly(d)}$, at the cost of a logarithmic factor in the number of rows. Similarly, the failure probability of the second step is $o(1)$ as mentioned in Theorem \ref{thm:subset_cols_lewis_weight}.
\end{proof}

\newpage
\section{The Streaming Algorithm and Full Analysis (Section 4)}
We give a single-pass streaming algorithm for $k$-CSS$_{p}$ ($1 \leq p < 2$) in \textbf{Algorithm}~\ref{alg:streaming}, which is based on the Merge-and-Reduce framework (see, e.g.~\cite{mcgregor14_graph_stream_survey}) previously used in graph streaming algorithms, for instance.

\begin{algorithm}[H]
\caption{A one-pass streaming algorithm for bi-criteria $k$-CSS$_p$ in the column-update streaming model.}
\label{alg:streaming}
\begin{algorithmic}
\STATE \textbf{Input: } A matrix $A \in \R^{d \times n}$ whose columns arrive one at a time, $p \in [1, 2)$, rank $k \in \N$ and batch size $r = k \poly(\log(nd))$.
\STATE \textbf{Output: } A subset of $\widetilde{O}(k)$ columns $A_I$.
\STATE Generate a dense $p$-stable sketching matrix $S \in \mathbb{R}^{k\poly(\log(nd)) \times d}$.
\STATE A list $C \leftarrow \{\}$ of strong coresets and corresponding level numbers.
\STATE A list $D \leftarrow \{\}$ of unsketched column subsets of $A$ corresponding to the list $C$ of strong coresets and their level numbers.
\STATE A list of sketched columns $M \leftarrow \{\}$.
\STATE A list of corresponding unsketched columns $L \leftarrow \{\}$.
\FOR{Each column $A_{*j}$ seen in the data stream}
\STATE $M \leftarrow M \cup SA_{*j}$
\STATE $L \leftarrow L \cup A_{*j}$
\IF{length of $M == r$}
\STATE $C \leftarrow C \cup (M, 0), D \leftarrow D \cup L$
\STATE $C, D \leftarrow$ \textbf{Recursive Merge}($C, D$) \COMMENT{// Algorithm~\ref{alg:recursive_merge}}
\STATE $M \leftarrow \{\}, L \leftarrow \{\}$
\ENDIF
\ENDFOR
\STATE $C \leftarrow C \cup (M, 0), D \leftarrow D \cup L$
\STATE $C, D \leftarrow$ \textbf{Recursive Merge}($C, D$) \COMMENT{// Algorithm~\ref{alg:recursive_merge}}
\STATE $I \gets $ The column indices obtained by applying to the concatenation of the strong coresets in $C$. (Here, $I$ is a set of column indices and $|I| = k \cdot \poly(\log k)$.)
\STATE Finally, recover $A_I$ by mapping the selected indices $I$ to unsketched columns in $D$.

\end{algorithmic}
\end{algorithm}

\begin{algorithm}[H]
\caption{Recursive Merge}
\label{alg:recursive_merge}
\begin{algorithmic}
\STATE \textbf{Input: } A list $C$ of strong coresets and their corresponding level numbers. A list $D$ of (unsketched) column subsets of $A$ corresponding to the sketched columns in $C$.
\STATE \textbf{Output: } New $C$, where the list of strong coresets is greedily merged, and the corresponding new $D$.
\IF{length of $C$ == 1}
\STATE Return $C, D$.
\ELSE
\STATE Let $(C_{-2}, l_{-2}), (C_{-1}, l_{-1})$ be the second to last and last elements of the list $C$ (i.e. the second to last and last sets of columns $C_{-2}, C_{-1}$ with their corresponding level $l_{-2}, l_{-1}$ from list $C$).
\IF{$l_{-2}$ == $l_{-1}$}
\STATE Remove $(C_{-2}, l_{-2}), (C_{-1}, l_{-1})$ from $C$.
\STATE Remove the corresponding $D_{-2}, D_{-1}$ from $D$.
\STATE Compute a strong coreset $C_0$ of (i.e., sample and rescale columns from) $C_{-2} \cup C_{-1}$, as described in the proof of Lemma \ref{lemma:coreset} --- $C_0$ has at least $k \cdot \poly(\log nd)$ columns. Record the original indices $I$ in $C_{-2} \cup C_{-1}$ of the columns selected in $C_0$.
\STATE Map indices $I$ to columns in $D_{-2} \cup D_{-1}$ to form a new subset of columns $D_0$.
\STATE $C \leftarrow C \cup (C_0, l_{-1} + 1), D \leftarrow D \cup D_0$.
\STATE \textbf{Recursive Merge}($C, D$).
\ELSE
\STATE Return $C, D$.
\ENDIF
\ENDIF
\end{algorithmic}
\end{algorithm}

To analyze our streaming algorithm, we first need \textbf{Lemma}~\ref{lemma:approx_merge} to show how the approximation error propagates through each level of the binary tree induced by the merge operator. It shows how a strong coreset $C_0$, computed at level $l$ of the tree, approximates the projection error for the union of all columns at the leaves of the subtree rooted at $C_0$.

\begin{customlemma}{5}[Approximation Error from Merging]
\label{lemma:approx_merge}
Let $C_0$ be a strong coreset constructed in a step of Algorithm \ref{alg:recursive_merge} (\textbf{Recursive Merge}), i.e. $C_{-1}$ and $C_{-2}$ are two strong coresets at levels $l - 1$ of the binary tree, and $C_0$ is a strong coreset at level $l$ obtained by taking a strong coreset for the concatenation of $C_{-1}$ and $C_{-2}$. Then,
\begin{itemize}
    \item If $C_0$ is a strong coreset of $C_{-1} \cup C_{-2}$, constructed as described in Lemma \ref{lemma:coreset}, with at least $\frac{k}{\gamma^2} \cdot \poly(\log (nd/\gamma))$ columns, then with probability at least $1 - \frac{1}{n^2}$,
    $$\min_V \|UV - C_0\|_{p, 2} = (1 \pm \gamma) \min_V \|UV - (C_{-1} \cup C_{-2})\|_{p, 2}$$
    for all matrices $U$ of rank at most $k$, simultaneously.
    \item If $\mathbf{C}$ is a strong coreset at level $l$ of the binary tree, and $M$ is the union of the columns of $SA$ represented as leaves of the subtree rooted at $\mathbf{C}$, then with probability at least $1 - \frac{q}{n^2}$,
    $$\min_V \|UV - \mathbf{C}\|_{p, 2} = (1 \pm \gamma)^l \|UV - M\|_{p, 2}$$
    for all rank-$k$ matrices $U$, simultaneously, as long as the coresets computed in each merge operation have at least $\frac{k}{\gamma^2} \cdot \poly(\log(nd/\gamma))$ columns. Here, $q$ is the number of nodes in the subtree rooted at $\mathbf{C}$.
\end{itemize}
\end{customlemma}

\begin{proof}

The first statement is a direct consequence of Lemma \ref{lemma:coreset} --- we are applying Lemma \ref{lemma:coreset}, setting approximation error $\eps = \gamma$, failure probability $\delta = \frac{1}{n^2}$ and rank $k = k$. Note that the coreset construction described in the proof of Lemma \ref{lemma:coreset} requires $O(d \log d/\eps^2) \cdot \log(1/\delta)$ columns to be sampled, but in our case, $d = k \cdot \poly(\log nd)$, since $SA$ has $k \cdot \poly(\log nd)$ rows.

We show the second statement by using the first statement, together with induction on the number of merge operations $l$, and a union bound over all the merge operations performed. The union bound is simply as follows: note that for each node $C_0$ in the subtree rooted at $\mathbf{C}$, with probability $1 - \frac{1}{n^2}$,
$$\min_V \|UV - C_0\|_{p, 2} = (1 \pm \gamma) \min_V \|UV - (C_{-1} \cup C_{-2})\|_{p, 2}$$
for all $U$ of rank $k$, where $C_{-1}$ and $C_{-2}$ are the two coresets corresponding to the children of $C_0$ (this is simply by the first statement of Lemma \ref{lemma:approx_merge}). 
Thus, by a union bound, this holds simultaneously for all nodes $C_0$ in the subtree rooted at $\mathbf{C}$, with probability at least $1 - \frac{q}{n^2}$, where $q$ is the number of nodes in the subtree rooted at $\mathbf{C}$. In other words, let $\calE$ be the event that for all nodes $C_0$, $C_0$ is a strong coreset of $C_{-1} \cup C_{-2}$ --- then $\calE$ occurs with probability at least $1 - \frac{q}{n^2}$. Note that since $q \leq 2n$, a failure probability $\delta = \frac{1}{n^2}$ suffices to pay for the union bound.

Assuming $\calE$ holds, we can apply induction. First let us consider the base case where $l = 1$ --- here, the desired result clearly holds since it is implied by the event $\calE$ (since the subtree rooted at $\mathcal{C}$ only has two other nodes, $C_{-1}$ and $C_{-2}$). 

Now suppose $l > 1$, and the second statement of Lemma \ref{lemma:approx_merge} holds for $\mathbf{C}$ at levels less than $l$. Now, suppose $\mathbf{C}$ is at level $l$, and let $C_{-1}$ and $C_{-2}$ be the coresets corresponding to the children of $\mathbf{C}$ in the binary tree. Let $M_{-1}, M_{-2}$ each be the contiguous submatrices of $SA$ represented by the leaves of the subtrees rooted at $C_{-1}$ and $C_{-2}$ respectively. Note that by its definition, $M = M_{-1} \cup M_{-2}$, where $M$ is as defined in the second statement of Lemma \ref{lemma:approx_merge}. 
Let $q_1, q_2$ be the number of nodes in the subtree rooted at $C_{-1}$ and $C_{-2}$ respectively, and $q = q_1 + q_2$ be the number of nodes in the subtree rooted at $\mathbf{C}$.
By the induction hypothesis, since $C_{-1}$ and $C_{-2}$ are at levels $l - 1$, for all rank-$k$ matrices $U$, with probability $1-\frac{q_1}{n^2}$,
$$\min_V \|UV - C_{-1}\|_{p, 2} = (1 \pm \gamma)^{l - 1} \min_V \|UV - M_{-1}\|_{p, 2}$$
and with probability $1-\frac{q_2}{n^2}$,
$$\min_V \|UV - C_{-2}\|_{p, 2} = (1 \pm \gamma)^{l - 1} \min_V \|UV - M_{-2}\|_{p, 2}$$
Thus, for any matrix $U$ of rank $k$, with probability $1-\frac{q}{n^2}$,
\begin{align*}
\min_V \|UV - \mathbf{C}\|_{p, 2}
& = (1 \pm \gamma) \min_V \|UV - (C_{-1} \cup C_{-2})\|_{p, 2} \\
& = (1 \pm \gamma) \Big(\min_V \|UV - C_{-1}\|_{p, 2}^p + \min_V \|UV - C_{-2}\|_{p, 2}^p\Big)^{1/p} \\
& = (1 \pm \gamma) \cdot\Big( (1 \pm \gamma)^{p(l - 1)} \cdot \min_V \|UV - M_{-1}\|_{p, 2}^p + (1 \pm \gamma)^{p(l - 1)} \cdot \min_V \|UV - M_{-2}\|_{p, 2}^p \Big)^{1/p} \\
& = (1 \pm \gamma)^l \cdot \Big( \min_V \|UV - M_{-1} \|_{p, 2}^p + \min_V \|UV - M_{-2}\|_{p, 2}^p \Big)^{1/p} \\
& = (1 \pm \gamma)^l \min_V \|UV - M\|_{p, 2}
\end{align*}
Here, the first equality is because the event $\calE$ occurs (meaning $\mathbf{C}$ is a $(1 \pm \gamma)$-approximate strong coreset for $C_{-1} \cup C_{-2}$). The second is by the definition of the $\ell_{p, 2}$ norm (since after raising it to the $p^{th}$ power, it decomposes across columns). The third equality is by the induction hypothesis, and the last equality is again because the $p^{th}$ power of the $\ell_{p, 2}$ norm decomposes across columns. This completes the proof of Lemma \ref{lemma:approx_merge}.
\end{proof}

Finally, using Lemma \ref{lemma:approx_merge}, we give a full analysis of our single-pass streaming algorithm, \textbf{Algorithm \ref{alg:streaming}}.

\begin{customthm}{2}[A One-pass Streaming Algorithm for $k$-CSS$_p$]
\label{thm:streaming_algo}
    In the column-update streaming model, let $A \in \mathbb{R}^{d \times n}$ be the data matrix whose columns arrive one at each time in a data stream. Given $p \in [1, 2)$ and a desired rank $k \in \mathbb{N}$, Algorithm~\ref{alg:streaming} outputs a subset of columns $A_I \in \mathbb{R}^{d \times k\poly(\log(k))}$ in $\widetilde{O}(\nnz(A)k + nk + k^3)$ time, such that with probability $1 - o(1)$,
    \begin{align*}
        \min_{V}\|A_IV - A\|_p \leq \widetilde{O}(k^{1/p - 1/2})\min_{L \subset [n], |L|=k}\|A_LV - A\|_p
    \end{align*}
    Moreover, Algorithm \ref{alg:streaming} only needs to process all columns of $A$ once and uses $\widetilde{O}(dk)$ space throughout the stream.
\end{customthm}

\begin{proof} \textbf{Approximation Factor:} 
Note that $n/r$, the number of leaves, might not be a power of $2$, and so we might get a list of coresets instead of a single one at the end of the stream.
Consider the list $C$, of coresets and their corresponding level numbers, left at the end of the stream, before Algorithm~\ref{alg:streaming} applies \textbf{Recursive Merge} after the data stream for the last time (to get a single output coreset). Denote these coresets by $\mathbf{C}_1, \mathbf{C}_2, \ldots, \mathbf{C}_t$, where $t = |C|$.

First, since strong coresets are subsets consisting of subsampled and reweighted columns of $SA$, we can let $SAT$ denote the concatenation of the $\mathbf{C}_i$, where $T$ is a sampling and reweighting matrix. In addition, note that for each $\mathbf{C}_i$, the subtree rooted at $\mathbf{C}_i$ has depth at most $\log(n/r)$, since all leaves represent contiguous blocks of $r$ columns, and each coreset also has $r$ columns. We bound $\min_V \|A_IV - A\|_p$ using Lemma \ref{lemma:approx_merge} and these observations. In the following, let $L \subset [n]$, $|L| = k$ denote the subset of $k$ columns of $A$ that gives the minimum $k$-CSS$_p$ cost, i.e. the one minimizing $\min_V \|A_LV - A\|_p$. First note that with probability $1-o(1)$,
\begin{align*}
\min_V \|A_IV - A\|_p
        &\leq \|A_IV' - A\|_p  &\text{(where $V' = \argmin_V \|SA_IV - SA\|_{p, 2}$)}\\
        & \leq \|SA_I V' - SA\|_p  &\text{By \textbf{Lemma}~\ref{lemma:pstable}}\\
        &=\widetilde{O}(k^{\frac{1}{p} - \frac{1}{2}})\|SA_IV' - SA\|_{p, 2}  &\text{By \textbf{Lemma}~\ref{lemma:norm}}
\end{align*}
where $I$ is the subset of column indices output by the bi-criteria $O(1)$-approximation algorithm for $k$-CSS$_{p, 2}$ that we apply at the end of Algorithm \ref{alg:streaming}. Let $(SAT)^*$ denote the best rank $k$ approximation to $SAT$ in the $\ell_{p, 2}$-norm. By \textbf{Theorem}~\ref{thm:bicriteria_k_css_p2}, with probability $1-o(1)$,
\begin{align*}
\widetilde{O}(k^{\frac{1}{p} - \frac{1}{2}})\|SA_I V' - SA\|_{p, 2}
&\leq \widetilde{O}(k^{\frac{1}{p} - \frac{1}{2}})\cdot O(1) \|(SAT)^* - SAT\|_{p, 2}\\
&\leq \widetilde{O}(k^{\frac{1}{p} - \frac{1}{2}})\min_{V} \|SA_LV - SAT\|_{p, 2}
\end{align*}
Now, recall that $SAT$ is the concatenation of the coresets $\mathbf{C}_i$ --- since the $p^{th}$ power of the $\ell_{p, 2}$ norm decomposes across columns,
\begin{align*}
\min_V \|SA_L V - SAT\|_{p, 2}^p
& = \sum_{i = 1}^t \min_V \|SA_L V - \mathbf{C}_i \|_{p, 2}^p
\end{align*}
Suppose that the subtree rooted at $\mathbf{C}_i$ has $q_i$ nodes, and has depth $l_i$, and in addition, let $M_i$ be the contiguous range of columns of $SA$ which are represented by the leaves of the subtree rooted at $\mathbf{C}_i$. Then, by the second statement of Lemma \ref{lemma:approx_merge}, with probability at least $1 - \frac{q_i}{n^2}$,
$$\min_V \|SA_L V - \mathbf{C}_i\|_{p, 2} = (1 \pm \gamma)^{l_i} \min_V \|SA_LV - M_i\|_{p, 2}$$
Thus, by a union bound, this occurs simultaneously for \textit{all} $i \in [t]$ with probability at least $1 - \sum_{i = 1}^t \frac{q_i}{n^2} = 1 - \frac{1}{n}$ (since the subtrees rooted at the $\mathbf{C}_i$'s together contain all coresets ever created by the streaming algorithm). Thus, with probability at least $1 - \frac{1}{n}$,
\begin{equation}
\begin{split}
\min_V \|SA_L V - SAT\|_{p, 2}^p
& = \sum_{i = 1}^t \min_V \|SA_L V - \mathbf{C}_i \|_{p, 2}^p \\
& = \sum_{i = 1}^t (1 \pm \gamma)^{pl_i} \min_V \|SA_L V - M_i\|_{p, 2}^p \\
& = (1 \pm \gamma)^{p\log(n/r)} \min_V \|SA_L V - SA\|_{p, 2}^p
\end{split}
\end{equation}
where the last equality is because $SA$ is the concatenation of the $M_i$, by their definition, and the subtree rooted at $\mathbf{C}_i$ has depth at most $\log(n/r)$. Taking $p^{th}$ roots, with probability $1-1/n$,
$$\min_V \|SA_L V - SAT\|_{p, 2} = (1 \pm \gamma)^{\log(n/r)} \min_V \|SA_L V - SA\|_{p, 2}$$
Setting $\gamma = \frac{\eps}{2\log(n/r)}$, we obtain the following with probability $1-1/n$,
$$\min_V \|SA_L V - SAT\|_{p, 2} = (1 \pm \eps) \min_V \|SA_L V - SA\|_{p, 2}$$
Thus, by a union bound over all the events, with probability $1-o(1)$,
\begin{align*}
\widetilde{O}(k^{\frac{1}{p} - \frac{1}{2}})\min_{V} \|SA_LV - SAT\|_{p, 2}
&\leq \widetilde{O}(k^{\frac{1}{p} - \frac{1}{2}}) \min_{V}\|SA_LV - SA\|_{p, 2}\\
&\leq \widetilde{O}(k^{\frac{1}{p} - \frac{1}{2}}) \min_{V}\|SA_LV - SA\|_{p} &\text{By \textbf{Lemma}~\ref{lemma:norm}}\\
&\leq \widetilde{O}(k^{\frac{1}{p} - \frac{1}{2}})\cdot \log^{1/p}(nd) \min_{V}\|A_LV - A\|_p &\text{By \textbf{Lemma}~\ref{lemma:bestcss}}
\end{align*}
and we conclude that $\min_V \|A_IV - A\|_p \leq \widetilde{O}(k^{\frac{1}{p} - \frac{1}{2}})\min_V \|A_LV - A\|_p$ with probability $1-o(1)$.

\paragraph{Space Complexity:} Since the nodes are merged greedily during the data stream, and within the list $C$ are in decreasing order according to their level, at most one node at each level $l$ is in the list $C$ at any time. Since the number of columns at each node in the binary tree is $\widetilde{O}(k)$ (i.e. the size of one coreset), the total space complexity is $\widetilde{O}(kd)$, suppressing logarithmic factors in $n, d, k$.

\paragraph{Running time:} Since generating a single $p$-stable random variable takes $O(1)$ time, generating the dense $p$-stable sketching matrix $S$ takes $O(dk \cdot \poly(\log(nd)))$ time. Computing $SA_{*j}, \forall j \in [n]$ takes a total of $O(\nnz(A)\cdot k\poly(\log (nd)))$ time. By Lemma~\ref{lemma:coreset}, merging two coresets, which are matrices of size $\widetilde{O}(k) \times \widetilde{O}(k)$, takes $\widetilde{O}(k^2)$ time. The merging operation is performed at most $O(n/k)$ times, so the total time it takes for merging is $\widetilde{O}(nk)$. By Theorem~\ref{thm:bicriteria_k_css_p2}, the $k$-CSS$_{p, 2}$ algorithm takes at most $k^3\poly(\log(knd))$ time to find the final subset of columns. Since the number of selected columns is $k\poly(\log k)$, it takes $k\poly(\log k)$ time to map the indices and recover the original columns $A_I$. Therefore, the overall running time is $\widetilde{O}(\nnz(A)k + nk + k^3)$, suppressing a low degree polynomial dependency on $\log(knd)$.
\end{proof}

\newpage
\section{The Distributed Protocol and Full Analysis (Section 5)}
We give our one-round distributed protocol for $k$-CSS$_{p}$ ($1 \leq p < 2$) in \textbf{Algorithm}~\ref{alg:protocol} and the full analysis below.

\begin{algorithm}[H]
   \caption{A one-round protocol for bi-criteria $k$-CSS$_p$ in the column partition model}
   \label{alg:protocol}
\begin{algorithmic}
    \STATE \textbf{Initial State:} \\
    Server $i$ holds matrix $A_i \in \mathbb{R}^{d \times n_i}$, $\forall i \in [s]$.
   \STATE {\bfseries Coordinator:}\\
     Generate a dense $p$-stable sketching matrix $S \in \mathbb{R}^{k \textrm{ poly}(\log (nd)) \times d}$.\\ Send $S$ to all servers.
   \STATE {\bfseries Server $i$:}\\ Compute $SA_i$. \\
   Let the number of samples in the coreset be $t = O(k \textrm{poly}(\log (nd))\log(1/\delta))$.
    Construct a coreset of $SA_i$ under the $\ell_{p,2}$ norm by applying a sampling matrix $D_i$ of size $n_i \times t$ and a diagonal reweighting matrix $W_i$ of size $t \times t$. \\
    Let $T_i = D_iW_i$.
    Send $SA_iT_i$ along with $A_iD_i$ to the coordinator.
    \STATE {\bfseries Coordinator:}\\ Column-wise stack $SA_iT_i$ to obtain $SAT = [SA_1T_1, SA_2T_2, \dots, SA_sT_s]$.\\
    Apply $k$-CSS$_{p, 2}$ on $SAT$ to obtain the indices $I$ of the subset of selected columns with size $O(k \cdot \poly(\log k))$. \\
    Since $D_i$'s are sampling matrices, the coordinator can recover the original columns of $A$ by mapping indices $I$ to $A_iD_i$'s. \\
    Denote the final selected subset of columns by $A_I$.
    Send $A_I$ to all servers.
    \STATE {\bfseries Server $i$:}\\
    Solve $\min_{V_i} \|A_IV_i - A_i\|_p$ to obtain the right factor $V_i$. $A_I$ and $V$ will be factors of a rank-$k \cdot \poly(\log k)$ factorization of $A$, where $V$ is the (implicit) column-wise concatenation of the $V_i$.
\end{algorithmic}
\end{algorithm}

\begin{customthm}{3}[A One-round Protocol for Distributed $k$-CSS$_p$]
\label{thm:analysis_of_distributed_protocol}
In the column partition model, let $A \in \R^{d \times n}$ be the data matrix whose columns are partitioned across $s$ servers and suppose server $i$ holds a subset of columns $A_i \in \mathbb{R}^{d \times n_i}$, where $n = \sum_{i \in [s]} n_i$. Then, given $p \in [1, 2)$ and a desired rank $k \in \N$, Algorithm \ref{alg:protocol} outputs a subset of columns $A_I \in \R^{d \times k \poly(log(k))}$ 
in $\widetilde{O}(\nnz(A)k + kd + k^3)$ time, such that with probability $1 - o(1)$,
$$\min_V \|A_I V - A\|_p \leq \widetilde{O}(k^{1/p - 1/2}) \min_{L \subset [n], |L| = k} \|A_L V - A\|_p$$
Moreover, Algorithm \ref{alg:protocol} uses one round of communication and $\widetilde{O}(sdk)$ words of communication.
\end{customthm}

\begin{proof}
\textbf{Approximation Factor:} 
In the following proof, let $L \subset [n], |L| = k$ denote the best possible subset of $k$ columns of $A$ that gives the minimum $k$-CSS$_p$ cost, i.e., the cost $\min_V \|A_LV - A\|_p$ achieves minimum. First, note that with probability $1-o(1)$,
\begin{align*}
\min_V \|A_IV - A\|_p 
& \leq \|A_I V' - A\|_p &\text{$V' := \argmin_{V} \|SA_IV - SA\|_{p, 2}$}\\ 
&\leq \|S A_I V' - SA \|_p \quad  &\text{By \textbf{Lemma}~\ref{lemma:pstable}} \\
&= \widetilde{O}(k^{\frac{1}{p} - \frac{1}{2}}) \|S A_I V' - SA\|_{p,2}  &\text{By \textbf{Lemma}~\ref{lemma:norm}}
\end{align*}
$SA_I$ is the selected columns output from the bi-criteria $O(1)$-approximation $k$-CSS$_{p, 2}$ algorithm. Let $(SAT)^*$ denote the best rank $k$ approximation to $SAT$. By \textbf{Theorem}~\ref{thm:bicriteria_k_css_p2}, with probability $1-o(1)$,
\begin{align*}
\widetilde{O}(k^{\frac{1}{p} - \frac{1}{2}}) \|S A_I V' - SA\|_{p,2}
& \leq \widetilde{O}(k^{\frac{1}{p} - \frac{1}{2}})\cdot O(1) \|(SAT)^* - SAT\|_{p, 2}\\
&\leq \widetilde{O}(k^{\frac{1}{p} - \frac{1}{2}}) \min_V \|SA_LV - SAT\|_{p, 2}
\end{align*}
Note that $SAT = [SA_1T_1, \dots, SA_sT_s]$ is a column-wise concatenation of all coresets of $SA_i$, $\forall i \in [s]$. By \textbf{Lemma}~\ref{lemma:coreset}, and a union bound over the $i \in [s]$, with probability $1 - s\delta = 1-o(1)$,
\begin{align*}
(\min_V \|SA_LV - SAT\|_{p, 2}^p)^{1/p}
&= (\sum_{i=1}^s \min_{V_i} \|SA_LV_i - SA_iT_i\|_{p, 2}^p)^{1/p}\\
&= (\sum_{i=1}^s (1\pm \epsilon)^p \min_{V_i}\|SA_LV_i - SA_i\|_{p, 2}^p)^{1/p}\\
&= (1 \pm \epsilon) (\sum_{i=1}^s \min_{V_i}\|SA_LV_i - SA_i\|_{p,2}^p)^{1/p}\\
&= (1\pm \epsilon) \min_V \|SA_LV - SA\|_{p, 2}
\end{align*}
Hence, by a union bound over all the events, with probability $1-o(1)$,
\begin{align*}
\widetilde{O}(k^{\frac{1}{p} - \frac{1}{2}}) \min_V \|SA_LV - SAT\|_{p, 2}
&\leq \widetilde{O}(k^{\frac{1}{p} - \frac{1}{2}}) \min_V \|SA_LV - SA\|_{p, 2}\\
&\leq \widetilde{O}(k^{\frac{1}{p} - \frac{1}{2}}) \min_V \|SA_LV - SA\|_{p} & \text{By \textbf{Lemma}~\ref{lemma:norm}}\\
&\leq \widetilde{O}(k^{\frac{1}{p} - \frac{1}{2}}) \cdot \log^{1/p}(nd)\min_V \|A_LV - A\|_p & \text{By \textbf{Lemma}~\ref{lemma:bestcss}}
\end{align*}
Thus, $\min_V \|A_IV - A\|_p \leq \widetilde{O}(k^{\frac{1}{p} - \frac{1}{2}}) \min_V \|A_LV - A\|_p$ with probability $1-o(1)$.

\paragraph{Communication Cost:}
Sharing the dense $p$-stable sketching matrix $S$ with all servers costs $O(sdk \cdot \poly(\log(nd)))$ communication (this can be removed with a shared random seed). Sending all coresets $SA_iT_i$ ($\forall i \in [s]$) and the corresponding columns $A_iD_i$ to the coordinator costs $\widetilde O(sdk)$ communication, since each coreset contains only $\widetilde{O}(k)$ columns (note that since we compute $s$ coresets, each coreset computation should have a failure probability of $\frac{1}{\poly(s)}$ to allow us to union bound --- this only increases the communication cost by a $\log(s)$ factor, however).
Finally, the coordinator needs $\widetilde{O}(sdk)$ words of communication to send the $\widetilde{O}(k)$ selected columns to each server. Therefore, the overall communication cost is $\widetilde{O}(sdk)$, suppressing a logarithmic factor in $n, d$.
\paragraph{Running time:}
Since generating a single $p$-stable random variable takes $O(1)$ time, generating the dense $p$-stable sketching matrix $S$ takes $O(dk \cdot \poly(\log(nd)))$ time. 
Computing all $SA_i$'s takes $O(\nnz(A)k \cdot \poly(\log (nd)))$ time. 
By Lemma~\ref{lemma:coreset}, computing all coresets for $SA_iT_i, \forall i \in [s]$ takes time $\widetilde{O}(kd)$.
By Theorem~\ref{thm:bicriteria_k_css_p2}, the $k$-CSS$_{p, 2}$ algorithm takes time $(\nnz(SAT) + k^2\poly(\log nd)) \cdot k\poly(\log k) \leq k^3\poly(\log (knd))$ to find the set of selected columns. Since the number of selected columns is $O(k \poly(\log k))$, it then takes the protocol $O(k \poly(\log k))$ time to map the indices and recover the original columns $A_I$. Therefore, the overall running time is $\widetilde{O}((\nnz(A)k + kd + k^3)$, suppressing a low degree polynomial dependency on $\log (knd)$.
After the servers receive $A_I$, it is possible to solve $\min_{V_i}\|A_IV_i - A_i\|_{p}$ in $\widetilde{O}(\nnz(A_I)) + \poly(d\log n)$ time
, $\forall i \in [s]$ due to~\cite{ww2019lpobliviousemb, ycrm2018wsgd}.
\end{proof}

\newpage
\section{A High Communication Cost Protocol for $k$-CSS$_p$ ($p \geq 1$) (Section 1)}
\label{appendix:badprotocol}

We describe in detail the naive protocol for distributed $k$-CSS$_p$ mentioned in Section 1, which works for all $p \geq 1$, in the column partition model, and which achieves an $O(k^2)$-approximation to the best rank-$k$ approximation, using $O(1)$ rounds and polynomial time but requiring a communication cost that is linear in $n + d$. The inputs are a column-wise partitioned data matrix $A \in \mathbb{R}^{d \times n}$ distributed across $s$ servers and a rank parameter $k \in \mathbb{N}$. Each server $i$ holds part of the data matrix $A_i \in \mathbb{R}^{d \times n_i}$, $\forall i \in [s]$, and such that $\sum_{i=1}^s n_i = n$.

We use a single machine, polynomial time bi-criteria $k$-CSS$_p$ algorithm as a subroutine of the protocol, e.g., Algorithm 3 in ~\cite{cgklpw2017lplra}, which selects a subset of $\widetilde{O}(k)$ columns $A_T$ of the data matrix $A \in \mathbb{R}^{d \times n}$ in polynomial time, for which $\min_X \|A_TX - A\|_p \leq O(k) \min_{\textrm{rank-k }A_k}\|A - A_k\|_p$, $\forall p \geq 1$.

\begin{algorithm}
\caption{A protocol for $k$-CSS$_p$ ($p \geq 1$)}
\begin{algorithmic}
\STATE \textbf{Initial State: } Server $i$ holds matrix $A_i \in \mathbb{R}^{d \times n_i}$, $\forall i \in [s]$.
\STATE {\bfseries Server $i$:} \\
Apply polynomial time bi-criteria $k$-CSS$_p$ on $A_i$ to obtain a subset $B_i$ of columns as the left factor. Solve for the right factor $V_i = \arg\min_{V_i} \|U_iV_i - A_i\|_p$. Send $U_i$ and $V_i$ to the coordinator.
\STATE {\bfseries Coordinator:}\\
Column-wise concatenate the $U_iV_i$ to obtain $UV = [U_1V_1, \dots, U_sV_s]$. Apply a polynomial time bi-criteria $k$-CSS$_p$ algorithm on $UV$ to obtain a subset $C$ of columns. Send $C$ to each server.
\STATE {\bfseries Server $i$:}\\
Solve $\min_{X_i} \|CX_i - A_i\|_p$ to obtain the right factor.
\end{algorithmic}
\end{algorithm}

\paragraph{Approximation Factor.}
Let $UV$ denote the column-wise concatenation of the $U_iV_i$.
Let $X^* = \arg\min_X \|CX - A\|_p$. Then,
\begin{align*}
    \|CX^* - AS\|_p &\leq \|CX^* - UV\|_p + \|UV - A\|_p &\text{By the triangle inequality}\\
    &\leq O(k)\min_{\textrm{rank-k }(UV)_k}\|UV - (UV)_k\|_p + \|UV - A\|_p &\text{By the $O(k)$-approximation of $k$-CSS$_p$}\\
    &\leq O(k)\|UV - A\|_p\\
    &=O(k)(\sum_{i=1}^s \|U_iV_i- A_i\|_p)\\
    &\leq O(k)(\sum_{i=1}^s O(k)\min_{\textrm{rank-k }A_i^*}\|A_i - A_i^*\|_p) &\text{By the $O(k)$-approximation of $k$-CSS$_p$}\\
    &\leq O(k^2)\sum_{i=1}^s\|A_i - (A^*)_i\|_p &\text{$A^* = \arg\min_{\textrm{rank-k }A^*}\|A - A^*\|_p$}\\
    &= O(k^2)\|A - A^*\|_p
\end{align*}

\paragraph{Communication Cost. } Since $U_i \in \mathbb{R}^{d \times \widetilde{O}(k)}$ and $V_i \in \mathbb{R}^{\widetilde{O}(k) \times n_i}$, sending $U_i$ and $V_i$ costs $\widetilde{O}(skn)$. Since $C \in \mathbb{R}^{d \times \widetilde{O}(k)}$, sending $C$ from the coordinator to all servers costs $\widetilde{O}(sdk)$. Thus the overall communication cost is $\widetilde{O}(s(n+d)k)$.

\paragraph{Running time.} According to~\cite{cgklpw2017lplra}, applying the $k$-CSS$_p$ algorithm and solving $\ell_p$ regression can both be done in polynomial time. Thus the overall running time of the protocol is polynomial.

\paragraph{Problems with this protocol.} Although this protocol works for all $p \geq 1$, a communication cost that linearly depends on the large dimension $n$ is too high, and furthermore, the output $C$ is not a subset of columns of $A$, because the protocol applies $k$-CSS$_p$ on a concatenation of both the left factor $U_i$ and the right factor $V_i$. $U_i$ is a subset of columns of $A_i$ but $V_i$ is not necessarily a sampling matrix. One might wonder whether it is possible that each server only sends $U_i$ and the coordinator then runs $k$-CSS$_p$ on a concatenation of the $U_i$. This will not necessarily give a good approximation to $\min_{\text{rank-k }A_k}\|A - A_k\|_p$ because the columns not selected in the $U_i$ locally on each server might become globally important. Finally, although it is possible to improve the approximation factor to $\widetilde{O}(k)$ by making use of an $\widetilde{O}(\sqrt{k})$-approximation algorithm for $\ell_p$-low rank approximation that also selects a subset of columns \cite{mw2019opt_l1_css_lra}, this protocol would still suffer from all of the aforementioned problems.

\newpage
\section{Greedy $k$-CSS$_{p, 2}$ and Full Analysis (Section 6)} \label{appendix:greedy_analysis}

\begin{algorithm}
\caption{Greedy $k$-CSS$_{p, 2}$. }
\label{algorithm:fast_greedy_kCSSp2}
\begin{algorithmic}
\STATE \textbf{Input:} The data matrix $A \in \R^{d \times n}$. A desired rank $k \in \mathbb{N}$ and $p \in [1, 2)$. The number of columns to be selected $r \leq n$. Failure probability $\delta \in (0, 1)$.
\STATE \textbf{Output:} A subset of $r$ columns $A_T$. 
\STATE Indices of selected columns $T \leftarrow \{\}$.
\FOR {$i = 1$ to $r$}
    \STATE $C \leftarrow$ Sample $\frac{n}{k}\log(\frac{1}{\delta})$ indices from $\{1, 2, \dots, n\} \setminus T$ uniformly at random.
    \STATE Column index $j^* \gets \argmin_{j \in C} (\min_{V} \|A_{T \cup j}V - A\|_{p,2}$)
    \STATE $T \leftarrow T \cup j^*$. 
\ENDFOR
\STATE Map indices $T$ to get the selected columns $A_T$.
\end{algorithmic}
\end{algorithm}

We propose a greedy algorithm for selecting columns in $k$-CSS$_{p, 2}$ (\textbf{Algorithm}~\ref{algorithm:fast_greedy_kCSSp2}) for $p \in [1, 2)$.
We give a detailed analysis on the first additive approximation compared to the error of the optimal column subset for Greedy $k$-CSS$_{p, 2}$.
Our analysis is inspired by the analysis of the Frobenius norm Greedy $k$-CSS$_2$ algorithm in~\cite{abfmrz2016greedycssfrobenius}.

Notice that during each iteration, the algorithm needs to evaluate the error $\min_V \|A_{T \cup j}V - A\|_{p, 2}$ to greedily pick the next column $j$.
A standard greedy algorithm which considers all unselected columns in $[n] \setminus T$ would need $O(nr)$ evaluations of the regression error $\min_V \|A_{T \cup j} V - A\|_{p, 2}$, which is too expensive. To improve the running time, we adopt the \textit{Lazier-than-lazy} framework for greedy algorithms originally introduced in~\cite{mbkvk15_lazier_than_lazy_greedy} and used by~\cite{abfmrz2016greedycssfrobenius} in greedy $k$-CSS$_2$. Instead of considering all unselected columns at each iteration, we first randomly sample $\frac{n}{k}\log(\frac{1}{\delta})$ candidate columns from $[n] \setminus T$ and greedily pick the next column only among those candidates. This reduces the number of evaluations of $\min_V \|A_{T \cup j} V - A\|_{p, 2}$ to $O(n\log (\frac{1}{\delta}))$.

To aid the analysis, we first define a utility function that quantifies how well the selected columns approximate the original matrix in \textbf{Notation} below as in~\cite{abfmrz2016greedycssfrobenius}.
We show in \textbf{Lemma}~\ref{lemma:single_vector} an improvement of the utility function with one additional column when projecting a single vector, based on \textbf{Lemma}~\ref{lemma:greedy_frobenius_single_vector} from~\cite{abfmrz2016greedycssfrobenius}. We then show an improvement of the utility function when projecting a matrix in \textbf{Lemma}~\ref{lemma:matrix_greedy_improvement}, by applying \textbf{Lemma}~\ref{lemma:single_vector} and Jensen's Inequality, following the analysis in~\cite{abfmrz2016greedycssfrobenius}. 
With \textbf{Lemma}~\ref{lemma:matrix_greedy_improvement}, we show a large expected improvement in the utility function by choosing a column from a subsampled candidates, based on \textbf{Lemma 6} from~\cite{abfmrz2016greedycssfrobenius}.
Finally, we conclude by giving the convergence rate and the running time for \textit{Lazier-than-lazy} based Greedy $k$-CSS$_{p,2}$ in \textbf{Theorem}~\ref{thm:convergence_of_fast_greedy_kCSSp2}.

\paragraph{Notation} Consider the input matrix $A \in \mathbb{R}^{d \times n}$ ($n \gg d$). 
Let $B$ be the matrix of normalized columns of $A$, where the $j$-th column of $B$ is $B_{*j} = A_{*j}/\|A_{*j}\|_2$. 
Let $\pi_T$ be the projection matrix onto the column span of $A_T$ or equivalently $B_T$. Let $\sigma_{\min}(M)$ denote the minimum singular value of some matrix $M$.

To aid our analysis, we define a \textit{utility function} $\Phi$ as follows,  inspired by \cite{abfmrz2016greedycssfrobenius}. For a subset $T \subset [n]$ and a matrix $M \in \R^{d \times t}$ (or a vector $M \in \mathbb{R}^d$), 
$$\Phi_M(T) = \|M\|_{p, 2}^{p} - \|M - \pi_T M\|_{p, 2}^{p} = \sum_{i = 1}^t \Big(\|M_{*i}\|_{2}^{p} - \|M_{*i} - \pi_T M_{*i}\|_{2}^{p}\Big) = \sum_{i = 1}^t \Phi_{M_{*i}}(T)$$

Observe that as the number of columns selected and added to $T$ increases, we get a more accurate estimation of $M$ and thus the approximation error $\|M - \pi_T M\|_{p,2}$ decreases, which results in an increase in the utility function $\Phi_M(T)$.

\begin{customlemma}{7.1}
\label{lemma:greedy_frobenius_single_vector} Let $S, T \subset [n]$ be two sets of column indices, with $S = \{i_1, \ldots, i_k\}$ and $\|\pi_S u\|_2 \geq \|\pi_T u\|_2$ for some vector $u \in \mathbb{R}^d$. Then,
$$\sum_{j = 1}^k \Big(\|\pi_{T_j'} u\|_2^2 - \|\pi_T u\|_2^2\Big) \geq \sigma_{min}(B_S)^2 \frac{(\|\pi_S u\|_2^2 - \|\pi_T u\|_2^2)^2}{4\|\pi_S u\|_2^2}$$
where $T_j' = T \cup \{i_j\}$ for all $j \in [k]$.
\end{customlemma}
\begin{proof}
    \textbf{Lemma 2} from~\cite{abfmrz2016greedycssfrobenius}, except that we replace the condition for $S$ and $T$, i.e., $\Phi_u(S) \geq \Phi_u(T)$ in~\cite{abfmrz2016greedycssfrobenius} with $\|\pi_S u\|_2 \geq \|\pi_T u\|_2$. The two conditions are equivalent, since 
    \begin{align*}
        \Phi_u(S) &\geq \Phi_u(T)\\
        \Leftrightarrow \|u\|_2^p - \|u - \pi_S u\|_2^p &\geq \|u\|_2^p - \|u - \pi_T u\|_2^p \\
        \Leftrightarrow \|u - \pi_S u\|_{2} &\leq \|u - \pi_T u\|_2\\
        \Leftrightarrow \|u\|_2^2 - \|\pi_S u\|_2^2 &\leq \|u\|_2^2 - \|\pi_T u\|_2^2\\
        \Leftrightarrow \|\pi_S u\|_2 &\geq \|\pi_T u\|_2
    \end{align*}
\end{proof}

\begin{customlemma}{7.2} (Utility Improvement by Projecting a Single Vector)
\label{lemma:single_vector}
Give $p \in [1, 2)$.
Let $S, T \subset [n]$ be two sets of column indices, with $\Phi_u(S) \geq \Phi_u(T)$ for some vector $u \in \mathbb{R}^d$. Let $k = |S|$, and for $i \in S$, let $T_i' = T \cup \{i\}$. Then,
$$\sum_{i \in S} \Big(\Phi_u(T_i') - \Phi_u(T) \Big) \geq \frac{p \sigma_{min}(B_S)^2}{16} \cdot \frac{(\Phi_u(S) - \Phi_u(T))^{2/p+1}}{\Phi_u(S)^{2/p}}$$
\end{customlemma}

\begin{proof}
To aid the analysis, we define the decreasing function $g: (-\infty, \|u\|_2^2] \to \R$ by 
$$g(x) = (\|u\|_2^2 - x)^{p/2}$$
and the derivative of $g$ is
$$|g'(x)| = \frac{p}{2} (\|u\|_2^2 - x)^{p/2 - 1}$$
which is an increasing function for $p < 2$. Then,
\begin{align*}
\sum_{i = 1}^k \Big(\Phi_u(T_i') - \Phi_u(T) \Big)
& = \sum_{i = 1}^k \Big(\|u - \pi_T u\|_2^p - \|u - \pi_{T_i'} u\|_2^p \Big) & \text{By definition of $\Phi$}\\
& = \sum_{i = 1}^k \Big((\|u\|_2^2 - \|\pi_T u\|_2^2)^{p/2} - (\|u\|_2^2 - \|\pi_{T_i'} u\|_2^2)^{p/2}\Big) & \text{By Pythagorean Theorem}\\
& = \sum_{i = 1}^k \Big( g(\|\pi_T u\|_2^2) - g(\|\pi_{T_i'} u\|_2^2) \Big) & \text{By definition of $g$} \\
& \geq \sum_{i = 1}^k |g'(\|\pi_T u\|_2^2)| \Big(\|\pi_{T_i'} u\|_2^2 - \|\pi_T u\|_2^2 \Big) & \text{Mean Value Theorem and} \\
& & \|\pi_T u\|_2 \leq \|\pi_{T_i'} u\|_2 \\
& = |g'(\|\pi_T u\|_2^2)| \sum_{i = 1}^k \Big( \|\pi_{T_i'} u\|_2^2 - \|\pi_T u\|_2^2 \Big) \\
& = \frac{p}{2} \Big(\|u\|_2^2 - \|\pi_T u\|_2^2\Big)^{p/2 - 1} \sum_{i = 1}^k \Big(\|\pi_{T_i'} u\|_2^2 - \|\pi_T u\|_2^2 \Big) \\
& = \frac{p}{2} \|u - \pi_T u\|_2^{p - 2} \sum_{i = 1}^k \Big(\|\pi_{T_i'} u\|_2^2 - \|\pi_T u\|_2^2 \Big) \\
& \geq \frac{p}{2} \|u - \pi_T u\|_2^{p - 2} \cdot \sigma_{min}(B_S)^2 \frac{(\|\pi_S u\|_2^2 - \|\pi_T u\|_2^2)^2}{4\|\pi_S u\|_2^2} & \text{Lemma~\ref{lemma:greedy_frobenius_single_vector}} \\
& = \frac{p \sigma_{min}(B_S)^2}{2} \cdot \frac{(\|u - \pi_T u\|_2^2 - \|u - \pi_S u\|_2^2)^2}{4 \|\pi_S u\|_2^2 \|u - \pi_T u\|_2^{2 - p}}
\end{align*}
Now we can lower bound
\begin{align*}
\|u - \pi_T u\|_2^2 - \|u - \pi_S u\|_2^2
& = \|u - \pi_T u\|_2^{2 - p} \|u - \pi_T u\|_2^p - \|u - \pi_S u\|_2^{2 - p} \|u - \pi_S u\|_2^p \\
& \geq \|u - \pi_T u\|_2^{2 - p} \Big(\|u - \pi_T u\|_2^p - \|u - \pi_S u\|_2^p \Big) & \text{since} \|u - \pi_S u\|_2 \leq \|u - \pi_T u\|_2\\
& = \|u - \pi_T u\|_2^{2 - p} \Big(\Phi_u(S) - \Phi_u(T)\Big)
\end{align*}

Thus,
\begin{align*}
\sum_{i = 1}^k \Big(\Phi_u(T_i') - \Phi_u(T) \Big)
& \geq \frac{p \sigma_{min}(B_S)^2}{2} \cdot \frac{(\|u - \pi_T u\|_2^2 - \|u - \pi_S u\|_2^2)^2}{4 \|\pi_S u\|_2^2 \|u - \pi_T u\|_2^{2 - p}} \\
& \geq \frac{p\sigma_{min}(B_S)^2}{2} \cdot \frac{\Big( \Phi_u(S) - \Phi_u(T)\Big)^2 \cdot \|u - \pi_T u\|_2^{2(2 - p)}}{4 \|\pi_S u\|_2^2 \|u - \pi_T u\|_2^{2 - p}} \\
& = \frac{p\sigma_{min}(B_S)^2}{2} \cdot \frac{\Big( \Phi_u(S) - \Phi_u(T)\Big)^2 \cdot \|u - \pi_T u\|_2^{2 - p}}{4 \|\pi_S u\|_2^2}
\end{align*}

Now to finish the proof, let us lower bound $\frac{\|u - \pi_T u\|_2^{2 - p}}{\|\pi_S u\|_2^2}$. First, observe that $\|u\|_2^{p} \|u - \pi_S u\|_2^{2-p} - \|u\|_2^{2-p}\|u-\pi_S u\|_2^{p} \geq 0$. To see why, observe that the following equivalences hold:
\begin{align*}
    \|u\|_2^{p} \|u - \pi_S u\|_2^{2-p} &\geq \|u\|_2^{2-p}\|u-\pi_S u\|_2^{p}\\
    \Leftrightarrow \|u\|_2^{p}/\|u\|_2^{2-p} &\geq \|u-\pi_S u\|_2^{p}/\|u-\pi_S u\|_2^{2-p}\\
    \Leftrightarrow \|u\|_2^{2p-2} &\geq \|u-\pi_S u\|_2^{2p-2}
\end{align*}
and the last statement is true, since $\|u\|_2 \geq \|u - \pi_S u\|_2$ and $f(x) = x^{2p-2}$ is a monotone function. Thus,
\begin{align*}
    \frac{\|u - \pi_T u\|_2^{2 - p}}{\|\pi_S u\|_2^2}
    &= \frac{\|u - \pi_T u\|_2^{2 - p}}{\|u\|_2^2 - \|u - \pi_S u\|_2^2} \\
    &\geq \frac{\|u - \pi_T u\|_2^{2 - p}}{\|u\|_2^2 - \|u - \pi_S u\|_2^2 + \|u\|_2^{p} \|u - \pi_S u\|_2^{2-p} - \|u\|_2^{2-p}\|u-\pi_S u\|_2^{p}}\\
    &= \frac{\|u - \pi_T u\|_2^{2 - p}}{(\|u\|_2^p - \|u - \pi_S u\|_2^p)(\|u\|_2^{2-p} + \|u - \pi_S u\|_2^{2-p})}\\
    &= \frac{1}{\Phi_{u}(S)}\cdot \frac{\|u - \pi_T u\|_2^{2 - p}}{\|u\|_2^{2-p} + \|u - \pi_S u\|_2^{2-p}}\\
    &\geq \frac{1}{2\Phi_{u}(S)}\cdot \frac{\|u - \pi_T u\|_2^{2 - p}}{\|u\|_2^{2-p}} &\text{Since $\|u\|_2^{2-p} \geq \|u - \pi_S u\|_2^{2-p}$}\\
    &= \frac{1}{2\Phi_{u}(S)} \cdot \Big(\frac{\|u - \pi_T u\|_2^{p}}{\|u\|_2^{p}}\Big)^{2/p-1}\\
    &= \frac{1}{2\Phi_{u}(S)} \cdot \Big(\frac{\|u\|_2^{p} - \Phi_{u}(T)}{\|u\|_2^{p}}\Big)^{2/p-1} &\text{By definition of $\Phi$}\\
    &\geq \frac{1}{2\Phi_{u}(S)}\cdot \Big(1 - \frac{\Phi_{u}(T)}{\Phi_{u}(S)}\Big)^{2/p-1} &\text{Since $\Phi_{u}(S) \leq \|u\|_2^{p}$}\\
    &= \frac{(\Phi_{u}(S) - \Phi_{u}(T))^{2/p-1}}{2\Phi_{u}(S)^{2/p}}
\end{align*}

Combining all the above inequalities gives
\begin{align*}
\sum_{i = 1}^k \Big(\Phi_u(T_i') - \Phi_u(T) \Big)
& \geq \frac{p\sigma_{min}(B_S)^2}{2} \cdot \frac{\Big( \Phi_u(S) - \Phi_u(T)\Big)^2 \cdot \|u - \pi_T u\|_2^{2 - p}}{4 \|\pi_S u\|_2^2} \\
& \geq \frac{p \sigma_{min}(B_S)^2}{8} \cdot \Big(\Phi_u(S) - \Phi_u(T)\Big)^2 \cdot \frac{(\Phi_{u}(S) - \Phi_{u}(T))^{2/p-1}}{2\Phi_{u}(S)^{2/p}}\\
&= \frac{p \sigma_{min}(B_S)^2}{16} \cdot \frac{(\Phi_{u}(S) - \Phi_{u}(T))^{2/p+1}}{\Phi_{u}(S)^{2/p}}
\end{align*}
This completes the proof.
\end{proof}

\begin{customlemma}{7.3}(Utility Improvement by Projecting a Matrix) \label{lemma:matrix_greedy_improvement}
Given $p \in [1, 2)$. Let $A \in \R^{d \times n}$, and $T, S \subset [n]$ be two sets of column indices, with $\Phi_A(S) \geq \Phi_A(T)$. Furthermore, let $k = |S|$. Then, there exists a column index $i \in S$ such that
$$\Phi_A(T \cup \{i\}) - \Phi_A(T) \geq p\sigma_{min}(B_S)^2 \frac{(\Phi_A(S) - \Phi_A(T))^{2/p+1}}{16k \Phi_A(S)^{2/p}}$$
\end{customlemma}

\begin{proof}
The proof mostly follows the proof of \textbf{Lemma 1} in \cite{abfmrz2016greedycssfrobenius}. We combine Lemma \ref{lemma:single_vector} with Jensen's inequality to conclude an improvement of the utility function with one additional column when projecting a matrix instead of a single column.

For $j \in [n]$, we define $\delta_j = \min(1, \frac{\Phi_{A_{*j}}(T)}{\Phi_{A_{*j}}(S)})$. Note that $\delta_j$ is $1$ if the $j$-th column $A_{*j}$ has a larger projection onto $B_T$ than $B_S$, and $\frac{\Phi_{A_{*j}}(T)}{\Phi_{A_{*j}}(S)}$ otherwise.
Note that $f(x) = x^{2/p+1}$ is convex on $x \in [0, \infty)$, $\forall 1\leq p < 2$, based on which we will apply Jensen's inequality.

Let $k = |S|$. For $i \in [k]$, let $T_i' = T \cup \{i\}$.
\begin{align*}
    \frac{1}{p\sigma_{min}(B_S)^2} \sum_{i = 1}^k \Big(\Phi_A(T_i') - \Phi_A(T)\Big)
    & = \frac{1}{p\sigma_{min}(B_S)^2} \sum_{j = 1}^n \sum_{i = 1}^k \Big(\Phi_{A_{*j}}(T_i') - \Phi_{A_{*j}}(T)\Big) 
    & \text{By definition of $\Phi$}\\
    & \geq \sum_{j = 1}^n \frac{(1 - \delta_j)^{2/p+1}}{16} \cdot \Phi_{A_{*j}}(S) 
    & \text{By \textbf{Lemma}~\ref{lemma:single_vector}}\\
    & = \frac{\Phi_A(S)}{16} \sum_{j = 1}^n (1 - \delta_j)^{2/p+1} \cdot \frac{\Phi_{A_{*j}}(S)}{\sum_{i=1}^n \Phi_{A_{*i}}(S)} 
    &\text{Note $\Phi_A(S) = \sum_{i=1}^n \Phi_{A_{*i}}(S)$}\\
    & \geq \frac{\Phi_A(S)}{16} \Big(\sum_{j = 1}^n (1 - \delta_j) \cdot \frac{\Phi_{A_{*j}}(S)}{\sum_{i=1}^n \Phi_{A_{*i}}(S)} \Big)^{2/p+1}
    &\text{By Jensen's Inequality}\\
    & = \frac{1}{16 \Phi_A(S)^{2/p}} \Big(\sum_{j = 1}^n (1 - \delta_j) \cdot \Phi_{A_{*j}}(S) \Big)^{2/p+1} \\
    & \geq \frac{1}{16 \Phi_A(S)^{2/p}} \Big(\sum_{j = 1}^n (\Phi_{A_{*j}}(S) - \Phi_{A_{*j}}(T)) \Big)^{2/p+1}     &\text{Since $1 - \delta_j \geq 1 - \frac{\Phi_{A_{*j}}(T)}{\Phi_{A_{*j}}(S)}$}\\
    & &\text{$\Rightarrow (1 - \delta_j) \cdot \Phi_{A_{*j}}(S) $}\\
    & &\text{$\geq \Phi_{A_{*j}}(S) - \Phi_{A_{*j}}(T)$}\\
    & = \frac{(\Phi_A(S) - \Phi_A(T))^{2/p+1}}{16\Phi_A(S)^{2/p}}    
\end{align*}

Hence,
$$\sum_{i = 1}^k \Big(\Phi_A(T_i') - \Phi_A(T)\Big) \geq p\sigma_{min}(B_S)^2 \frac{(\Phi_A(S) - \Phi_A(T))^{2/p+1}}{16\Phi_A(S)^{2/p}}$$
This implies there is at least one column of $B_S$, with index $i \in S$, such that when $i$ is added to $T$, the utility function $\Phi_A(T)$ increases by at least $\frac{1}{k}\cdot p\sigma_{min}(B_S)^2 \frac{(\Phi_A(S) - \Phi_A(T))^{2/p+1}}{16\Phi_A(S)^{2/p}}$.
\end{proof}

\begin{customlemma}{7.4}[Expected Increase in Utility]
\label{lemma:expected_increase_in_utility}
Given $p \in [1, 2)$.
Let $A \in \R^{d \times n}$, and let $T, S \subset [n]$ be two sets of column indices, with $k := |S|$ and $\Phi_A(S) \geq \Phi_A(T)$. Let $\overline{T}$ be a set of $\frac{n \log(1/\delta)}{k}$ column indices of $A$, chosen uniformly at random from $[n] \setminus T$. Then,
$$\mathbb{E}[\max_{i \in \overline{T}} \Phi_A(T \cup \{i\})] - \Phi_A(T) \geq (1 - \delta) \cdot p\sigma_{min}(B_S)^2 \cdot \frac{(\Phi_A(S) - \Phi_A(T))^{2/p+1}}{16k \Phi_A(S)^{2/p}}$$
\end{customlemma}

\begin{proof}
The proof is nearly identical to the proof of \textbf{Lemma 6} of \cite{abfmrz2016greedycssfrobenius} --- we include the full proof for completeness. The first step in the proof is showing that $\overline{T} \cap (S \setminus T)$ is nonempty with high probability. Then, by conditioning on $\overline{T} \cap (S \setminus T)$ being nonempty, we can show that the expected increase in utility is large. For the purpose of this analysis, we assume that the columns of $\overline{T}$ are sampled independently with replacement. At the end of the proof, we discuss sampling the columns of $\overline{T}$ without replacement.

First, observe that
\begin{align*}
\Pr[\overline{T} \cap (S \setminus T) = \varnothing]
& = \prod_{t = 1}^{O\big(\frac{n \log(1/\delta)}{k}\big)} \Big(1 - \frac{|S \setminus T|}{n - |T|}\Big) \\
& = \Big(1 - \frac{|S \setminus T|}{n - |T|}\Big)^{O\big(\frac{n\log(1/\delta)}{k}\big)} \\
& \leq e^{-\frac{|S \setminus T|}{n - |T|} \cdot \frac{n \log(1/\delta)}{k}} & \text{By $1 - x \leq e^{-x}$} \\
& \leq e^{-\frac{|S \setminus T| \log(1/\delta)}{k}} & \text{Because $n - |T| < n$}
\end{align*}
meaning that
\begin{align*}
\Pr[\overline{T} \cap (S \setminus T) \neq \varnothing]
& \geq 1 - e^{-\frac{|S \setminus T| \log(1/\delta)}{k}} \\
& = 1 - \delta^{\frac{|S \setminus T|}{k}} \\
& \geq (1 - \delta) \frac{|S \setminus T|}{k} & \text{Since $|S \setminus T| \leq k$, and $1 - \delta^x \geq (1 - \delta)x$ for $x, \delta \in [0, 1]$}
\end{align*}

Therefore,
\begin{align*}
& \E[\max_{i \in \overline{T}} \Phi_A(T \cup \{i\}) - \Phi_A(T)] \\
& \geq \Pr[\overline{T} \cap (S \setminus T) \neq \varnothing] \cdot \E\Big[\max_{i \in \overline{T}} \Phi_A(T \cup \{i\}) - \Phi_A(T) \Big| \overline{T} \cap (S \setminus T) \neq \varnothing \Big] \\
& \geq (1 - \delta) \frac{|S \setminus T|}{k} \cdot \E\Big[\max_{i \in \overline{T}} \Phi_A(T \cup \{i\}) - \Phi_A(T) \Big| \overline{T} \cap (S \setminus T) \neq \varnothing \Big] \\
& \geq (1 - \delta) \frac{|S \setminus T|}{k} \cdot \E\Big[\max_{i \in \overline{T}} \Phi_A(T \cup \{i\}) - \Phi_A(T) \Big| |\overline{T} \cap (S \setminus T)| = 1 \Big] \\
& \,\,\,\,\,\,\,\, \text{(Since it is always better for $\overline{T} \cap (S \setminus T)$ to be larger)} \\
& = (1 - \delta) \frac{|S \setminus T|}{k} \cdot \frac{\sum_{i \in S \setminus T} (\Phi_A(T \cup \{i\}) - \Phi_A(T))}{|S \setminus T|} \\ 
& \,\,\,\,\,\,\,\, \text{(Since the single element of $\overline{T} \cap (S \setminus T)$ is uniformly random in $(S \setminus T)$)} \\
& = (1 - \delta) \cdot \frac{\sum_{i \in S} (\Phi_A(T \cup \{i\}) - \Phi_A(T))}{|S|} \\
& \,\,\,\,\,\,\,\, \text{(Since $\Phi_A(T \cup \{i\}) = \Phi_A(T)$ for $i \in T$)} \\
& \geq (1 - \delta) \cdot \frac{1}{k} \cdot p\sigma_{min}(B_S)^2 \frac{(\Phi_A(S) - \Phi_A(T))^{2/p+1}}{16\Phi_A(S)^{2/p}} \\ 
& \,\,\,\,\,\,\,\, \text{(By \textbf{Lemma} \ref{lemma:matrix_greedy_improvement}.)}
\end{align*}
This proves the lemma in the case where the columns are sampled with replacement. Now, we discuss what happens when sampling without replacement. Note that the expected increase in utility can only be higher if the columns of $\overline{T}$ are sampled without replacement. Intuitively, this is because if $\overline{T}$ has some repeated columns, then it is always better to replace those repeated columns with other columns of $A$. Thus, for each instance of $\overline{T}$ where some columns are sampled multiple times, we can ``move'' all of the probability mass from this instance of $\overline{T}$ to other sets $\overline{T}' \subset [n] \setminus T$, which contain $\overline{T}$ but do not have repeated elements. This leads to the uniform distribution on subsets of $[n] \setminus T$ with no repeated elements, i.e., the distribution that results from sampling without replacement.
\end{proof}

Using this lemma, we analyze the convergence rate and the running time of \textbf{Algorithm}~\ref{algorithm:fast_greedy_kCSSp2}:

\begin{customthm}{4}[Greedy $k$-CSS$_{1, 2}$]
\label{thm:convergence_of_fast_greedy_kCSSp2}
    Let $p \in [1, 2)$.
    Let $A \in \mathbb{R}^{d \times n}$ be the data matrix and $k \in \mathbb{N}$ be the desired rank. Let $A_L$ be the best possible subset of $k$ columns, i.e., $A_L = \argmin_{A_L} \min_{V}\|A_LV - A\|_{p,2}$. Let $\sigma$ be the minimum non-zero singular value of the matrix $B$ of normalized columns of $A_L$, (i.e., the $j$-th column of $B$ is $B_{*j} = (A_L)_{*j}/\|(A_L)_{*j}\|_2$).
    Let $T \subset [n]$ be the subset of output column indices selected by \textbf{Algorithm}~\ref{algorithm:fast_greedy_kCSSp2}, for $\epsilon, \delta \in (0, 1)$, for $|T| = \Omega(\frac{k}{p\sigma^2 \epsilon^2})$, with probability $1 - \delta$,  $$\mathbb{E}[\min_{V} \|A_TV - A\|_{p,2}] \leq \min_V \|A_LV - A\|_{p,2}  + \epsilon \|A\|_{p,2}$$ 
    The overall running time is $O(\frac{n}{p\sigma^2\epsilon^2}\log(\frac{1}{\delta})\cdot (\frac{dk^2}{p^2\sigma^4\epsilon^4} + \frac{ndk}{p\sigma^2\epsilon^2}))$.
\end{customthm}

\begin{proof}

\textbf{Convergence Rate. }
The proof uses the same strategy as that of Theorem 5 of \cite{abfmrz2016greedycssfrobenius}, with minor modifications. 
Let $S := B_{L}$ be the best subset of columns of $B$.
Let $T_t$ be the subset of columns of $B$ selected by \textbf{Algorithm} \ref{algorithm:fast_greedy_kCSSp2} after $t$ iterations (in particular, $T_0 = \varnothing$). 
In addition, let $F = \Phi_A(S) = \Phi_A(S) - \Phi_A(T_0)$ be the distance from the current value of the utility function to the best achievable value. Let $\Delta_t$ denote a small amount at time $t$ that is used to quantify our progress as how good the currently selected subset of columns approximates the matrix.
When no column is selected, $\Delta_0 = F$. Let $\Delta_{i + 1} = \frac{\Delta_i}{2}$. 
Now, fix a time $t$ such that for some $i$, $\Delta_i \geq \Phi_A(S) - \Phi_A(T_t) \geq \Delta_{i + 1} = \frac{\Delta_i}{2}$. Then, we bound the number of additional iterations $t'$ needed so that
$$\E[\Phi_A(S) - \Phi_A(T_{t + t'}) \mid T_t] \leq \Delta_{i + 1}$$

For convenience, for each $k \geq 0$, define $E_k := E[\Phi_A(T_{t + k}) \mid T_t]$. Then, our goal is to find $t'$ such that
$$\Phi_A(S) - E_{t'} \leq \Delta_{i + 1}$$
However, observe that from Lemma \ref{lemma:expected_increase_in_utility} above, we obtain
\begin{align*}
E_{k + 1} - E_k
& = \E\Big[\Phi_A(T_{t + k + 1}) - \Phi_A(T_{t + k}) \Big| T_t \Big] \\
& = \E\Big[ \E\big[\Phi_A(T_{t + k + 1}) - \Phi_A(T_{t + k}) \big| T_{t + k}\big] \Big| T_t \Big] & \text{ By }\E[\E[X|Y]] = \E[X] \\
& \geq \E\Big[ (1 - \delta) \cdot p\sigma_{min}(B_S)^2 \cdot \frac{(\Phi_A(S) - \Phi_A(T_{t + k}))^{2/p+1}}{16k\Phi_A(S)^{2/p}} \Big| T_t \Big] & \text{ By Lemma \ref{lemma:expected_increase_in_utility}} \\
& = \frac{(1 - \delta) \cdot p \sigma_{min}(B_S)^2}{16k\Phi_A(S)^{2/p}} \cdot \E\big[(\Phi_A(S) - \Phi_A(T_{t + k}))^{2/p+1} | T_t\big] \\
& \geq \frac{(1 - \delta) \cdot p\sigma_{min}(B_S)^{2}}{16k\Phi_A(S)^{2/p}} \cdot \Big(\E[\Phi_A(S) - \Phi_A(T_{t + k}) | T_t]\Big)^{2/p+1} & \text{ By Jensen's Inequality} \\
& = (1 - \delta) \cdot p\sigma_{min}(B_S)^2 \cdot \frac{(\Phi_A(S) - E_k)^{2/p+1}}{16k\Phi_A(S)^{2/p}}
\end{align*}
Now, suppose that $\Delta_i \geq \Phi_A(S) - E_s \geq \Delta_{i + 1}$, for $s = 0, \ldots, t' - 1$. Then, for all such $s$, $E_{s + 1} - E_s \geq \frac{(1 - \delta) p\sigma_{min}(B_S)^2 \Delta_{i + 1}^{2/p+1}}{16kF^{2/p}}$. 
Summing these inequalities for $s = 0, \ldots, t' - 1$, we find that
$$E_{t'} - E_0 \geq \frac{(1 - \delta) p\sigma_{min}(B_S)^2}{16kF^{2/p}} \cdot \Delta_{i + 1}^{2/p+1} \cdot t'$$
and for the increase from $E_0$ to $E_{t'}$ to be greater than $\Delta_{i + 1}$, it suffices to have
$$t' \geq \frac{32kF^{2/p}}{\Delta_{i + 1}^{2/p} \cdot (1 - \delta)p \sigma_{min}(B_S)^2}$$

In summary, if $\Phi_A(S) - \E[\Phi_A(T_t)] \leq \Delta_i$, then in at most $s = \frac{32kF^{2/p}}{\Delta_{i + 1}^{2/p} \cdot (1 - \delta)p \sigma_{min}(B_S)^2}$ iterations, $\Phi_A(S) - \E[\Phi_A(T_{t + s})] \leq \Delta_{i + 1}$. Thus, if we let $N \in \N$ such that $\Delta_{N + 1} \leq \frac{\eps}{(1 - \delta)^{p/2}} F \leq \Delta_N$, then the number of iterations $t$ needed to have $\Phi_A(S) - \E[\Phi_A(T_t)] < \Delta_{N + 1}$ is at most

\begin{align*}
\sum_{i = 0}^N \frac{32kF^{2/p}}{\Delta_{i + 1}^{2/p} \cdot (1 - \delta) p\sigma_{min}(B_S)^2}
& = \frac{32kF^{2/p}}{(1 - \delta)p\sigma_{min}(B_S)^2} \sum_{i = 0}^N \frac{1}{\Delta_{i + 1}^{2/p}} \\
& = \frac{32kF^{2/p}}{(1 - \delta)p\sigma_{min}(B_S)^2} \sum_{i = 0}^N \frac{1}{4^{(N - i)/p}} \cdot \frac{1}{\Delta_{N + 1}^{2/p}} \\
& \leq \frac{32kF^{2/p}}{(1 - \delta)p\sigma_{min}(B_S)^2} \cdot \frac{4^{1/p}(1 - \delta)}{\eps^{2/p}F^{2/p}} \sum_{i = 0}^{N} \frac{1}{4^{(N - i)/p}} & \text{ Since }\Delta_{N + 1}^{2/p} \geq (\frac{\eps F}{2(1 - \delta)^{p/2}})^{2/p} \\
& \leq \frac{128k}{ p \sigma_{min}(B_S)^2 \eps^{2/p}} \sum_{i = 0}^N \frac{1}{2^{i}} &\text{Since $p < 2$}\\ 
& \leq \frac{256k}{p\sigma_{min}(B_S)^2 \eps^{2/p}}
\end{align*}

Thus, after $t = O(\frac{k}{p\sigma_{min}(B_S)^2 \eps^{2/p}})$ iterations,
$$\Phi_A(S) - \E[\Phi_A(T_t)] \leq \frac{\eps}{(1 - \delta)^{p/2}} \Phi_A(S) \leq \frac{\eps}{\sqrt{1 - \delta}} \Phi_A(S)$$
meaning
$$\|A\|_{p, 2}^p - \|A - \pi_S A\|_{p, 2}^p - \E[\|A\|_{p, 2}^p - \|A - \pi_T A\|_{p, 2}^p] \leq \frac{\eps}{\sqrt{1 - \delta}} \|A\|_{p, 2}^p - \frac{\eps}{\sqrt{1 - \delta}} \|A - \pi_S A\|_{p, 2}^p$$
and rearranging gives
$$\E[\|A - \pi_T A\|_{p, 2}^p] \leq \Big(1 - \frac{\eps}{\sqrt{1 - \delta}}\Big) \|A - \pi_S A\|_{p, 2}^p + \frac{\eps}{\sqrt{1 - \delta}} \|A\|_{p, 2}^p$$
and observe that if we select $\delta = \eps$, then $\frac{1}{\sqrt{1 - \delta}} = O(1)$ for $\eps < \frac{1}{2}$. Therefore,
\begin{align*}
E[\|A - \pi_T A\|_{p, 2}]
& \leq E[\|A - \pi_T A\|_{p, 2}^p]^{1/p} & \text{(By Jensen's inequality since $x^{1/p}$ is concave)} \\
& \leq \Big((1 - O(\eps)) \|A - \pi_S A\|_{p, 2}^p + O(\eps) \|A\|_{p, 2}^p \Big)^{1/p} \\
& \leq (1 - O(\eps))^{1/p} \|A - \pi_S A\|_{p, 2} + O(\eps)^{1/p} \|A\|_{p, 2} & (x + y)^{1/p} \leq x^{1/p} + y^{1/p} \\
& \leq \|A - \pi_S A\|_{p, 2} + O(\eps)^{1/p} \|A\|_{p, 2}
\end{align*}
In summary, if we select $O(\frac{k}{p\sigma_{min}(B_S)^2\eps^{2/p}})$ columns, then
$$E[\|A - \pi_T A\|_{p, 2}] \leq \|A - \pi_S A\|_{p, 2} + O(\eps)^{1/p} \|A\|_{p, 2}$$
and replacing $\eps$ with $O(\eps^p)$, we find that
$$E[\|A - \pi_T A\|_{p, 2}] \leq \|A - \pi_S A\|_{p, 2} + \eps \|A\|_{p, 2}$$
after $O(\frac{k}{p\sigma_{min}(B_S)^2 \eps^2})$ iterations.

\textbf{Running Time. }
Each evaluation of the error $\min_V \|A_{T\cup j}V - A\|_{p, 2}$ takes $O(d|T|^2 + nd|T|)$ time by taking the pseudo-inverse of $A_{T \cup j}$. Since the algorithm samples $\frac{n}{k}\log(\frac{1}{\delta})$ columns at each iteration, the time it takes for each iteration is $O(\frac{n}{k}\log(\frac{1}{\delta})(d|T|^2 + nd|T|))$.
If we set the number of iterations, i.e. the number of selected columns $r = \max |T| = \frac{k}{p\sigma^2 \epsilon^2}$, the overall running time is then $O(\frac{k}{p\sigma^2 \epsilon^2}\cdot \frac{n}{k}\log(\frac{1}{\delta}) \cdot (\frac{dk^2}{p^2\sigma^4 \epsilon^4} + \frac{ndk}{p\sigma^2 \epsilon^2})) = O(\frac{n}{p\sigma^2 \epsilon^2}\log(\frac{1}{\delta})\cdot(\frac{dk^2}{p^2\sigma^4 \epsilon^4} + \frac{ndk}{p\sigma^2 \epsilon^2}))$.

\end{proof}

\newpage
\section{Additional Experimental Details}
\label{appendix:add_details_experiments}

\textbf{Hyperparameters for $k$-CSS$_{1,2}$. }
There are two additional hyperparameters for our $O(1)$-approximation bi-criteria $k$-CSS$_{1, 2}$ (see Section~\ref{appendix:kcss_in_lp2_norm} for a complete description of this algorithm), i.e. the size of the sparse embedding matrix and its sparsity, which we use to generate a rank-$k$ left factor that gives an $O(1)$-approximation (see Section~\ref{subsec:osnap_sparse_embedding} for details on the embedding matrix and sparsity). In the experiments, we set both the sparsity and the size of the sketching matrix we use to be $\frac{k}{2}$, where $k$ is the number of output columns.

\end{appendices}

\end{document}